\documentclass[a4paper,12pt]{article} 
\usepackage{amsmath,amsthm,amssymb}
\newtheorem{theorem}{Theorem}[section]

\newtheorem{proposition}[theorem]{Proposition}
\newtheorem{lemma}[theorem]{Lemma}
\newtheorem{corollary}[theorem]{Corollary}

\newcommand{\bm}[1]{\mbox{\boldmath$#1$}}
\newcommand{\be}{\begin{equation}}
\newcommand{\ee}{\end{equation}}
\newcommand{\bea}{\begin{eqnarray}}
\newcommand{\eea}{\end{eqnarray}}
\newcommand{\non}{\nonumber}
\begin{document}

\title{Reduction formula of form factors for the integrable spin-$s$ 
XXZ chains and application to the correlation functions}
\author{
Tetsuo Deguchi\footnote{e-mail deguchi@phys.ocha.ac.jp} 
}
\date{ }
%\date{\today}
\maketitle
\begin{center}   
Department of Physics, Graduate School of Humanities and Sciences, 
Ochanomizu University \\ 
 2-1-1 Ohtsuka, Bunkyo-ku, Tokyo 112-8610, Japan 
\end{center} 

\begin{abstract} 
For the integrable spin-$s$ XXZ chain 
we express explicitly any given spin-$s$ form factor 
in terms of a sum over the scalar products of the spin-1/2 operators.  
Here they are given by the operator-valued matrix elements of  
the monodromy matrix of the spin-1/2 XXZ spin chain. 
In the paper  we call an arbitrary matrix element of a local operator 
between two  Bethe eigenstates a form factor of the operator.   
We derive all important formulas of the fusion method in detail.  
We thus revise the derivation of the higher-spin XXZ form factors 
given in a previous paper. The revised method has  
several interesting applications in mathematical physics. For instance,  
we express the spin-$s$ XXZ correlation function of an arbitrary entry 
at zero temperature in terms of a sum of multiple integrals.  
\end{abstract}

\newpage 
%%%%%%%%%%%%%%%%%%%%%%%%%%%%%%%%%%%%%%%%%%%%%%%%%%%%%%%%%%%
%
%  Section 1 
%
\setcounter{equation}{0} 
\renewcommand{\theequation}{1.\arabic{equation}}
\section{Introduction} 

The multiple-integral representations of correlation functions 
of the spin-1/2 XXZ spin chain have attracted much interest 
during the last two decades in the mathematical physics 
of integrable quantum spin chains 
\cite{Korepin,Slavnov,Miki,Jimbo-Miwa,MS2000,KMT1999,KMT2000,Goehmann-CF}. 
They are also derived for the integrable higher-spin XXX 
spin chains through the algebraic Bethe-ansatz method 
\cite{Kitanine2001,Castro-Alvaredo}. 
The multiple-integral representations of 
the finite-temperature correlation functions 
of the integrable isotropic spin-1 chain have been explicitly derived 
\cite{GSS}.

The Hamiltonian of the spin-1/2 XXZ spin chain 
under the periodic boundary conditions (P.B.C.) is given by  
\begin{align}
{\cal H}_{\rm XXZ} =  {\frac 1 2} 
\sum_{j=1}^{L} \left(\sigma_j^X \sigma_{j+1}^X +
 \sigma_j^Y \sigma_{j+1}^Y + \Delta \sigma_j^Z \sigma_{j+1}^Z  \right) \, . 
\label{hxxz} 
\end{align}
Here $\sigma_j^{a}$ ($a=X, Y, Z$) are the Pauli matrices defined 
on the $j$th site and $\Delta$ denotes the anisotropy 
of the exchange coupling. The P.B.C. are given by 
$\sigma_{L+1}^{a} = \sigma_{1}^{a}$ for $a=X, Y, Z$.   
In terms of the $q$ parameter of the quantum group $U_q(sl_2)$, 
we express $\Delta$ by  
%\begin{equation}
$\Delta = {\frac 1 2}(q+q^{-1})$.  
%\end{equation}
We define parameter $\eta$ by $q = \exp \eta$. 
The transfer matrix of the XXZ spin chain has free parameters which we call 
the inhomogeneity parameters $w_j$ for $j=1, 2, \ldots, L$.

Recently, a systematic method for evaluating the form factors 
and correlation functions of the integrable higher-spin 
XXZ spin chain has been proposed by applying the fusion method 
\cite{DM1,DM2,DM3}. However, the proposed method was not completely correct 
\cite{DM4}. 
In the paper, we revise the previous method for evaluating 
the spin-$\ell/2$ form factors, and formulate systematic formulas 
by which we can express any given spin-$\ell/2$ form factor 
in terms of a sum over the scalar products of the spin-1/2 operators. 
We show the derivation of the revised method, explicitly. 
As an application of of the fusion method we derive 
a concise multiple-integral representation for 
the spin-$s$ XXZ correlation function of an arbitrary 
entry in a region of the gapless regime, 
which is now expressed in terms of a sum of the multiple integrals.  
Here we remark that we call 
an arbitrary  matrix element of a local spin-$s$ operator 
between two Bethe eigenvectors 
a form factor of the operator in the paper, 
following Refs. \cite{Slavnov,KMT1999}.   

Let us review the fusion method for evaluating 
the higher-spin form factors briefly, and point out where it was wrong. 
We consider the integrable spin-$\ell/2$ XXZ spin chain 
for an integer $\ell$ with $\ell > 1$. 
In the fusion method we construct the spin-$\ell/2$ XXZ transfer matrix 
in the following two steps: 
we first construct the spin-1/2 transfer matrix with $w_j=w_j^{(\ell)}$ 
from the product of the spin-1/2 $R$-matrices with 
their rapidities shifted by the inhomogeneity parameters $w_j$ 
which are given by the $N_s$ pieces of the complete $\ell$-strings 
$w_j^{(\ell)}$ such as 
$w_j=w_j^{(\ell)}$ for $j=1, 2, \ldots, L$ (see \S 2.4); 
we then multiply the product with the spin-$\ell/2$ projection operators.    
Here, the spin-$\ell/2$ chain with $N_s$ sites 
is defined on the spin-1/2 chain with $L$ sites, where $L=\ell N_s$.

When we evaluate the spin-$\ell/2$ form factors with the fusion method, 
we reduce each of the spin-$\ell/2$ operators 
into a sum of products of the local spin-1/2 operators,  
which we want to express in terms of the operator-valued matrix elements 
of the spin-1/2 monodromy matrix 
through the formula of the quantum inverse-scattering problem (QISP).  
We then want to calcuate the expectaton value or the matrix elements 
of the sums of products of the local spin-1/2 operators 
with respect to the Bethe states 
by making use of Slavnov's formula of scalar products 
for the spin-1/2 Bethe-ansatz operators.   
Here we recall that the inhomogeneity parameters $w_j$ 
of the spin-1/2 monodromy matrix 
are given by the $N_s$ pieces of the complete $\ell$-strings, 
$w_j=w_j^{(\ell)}$, for $j=1, 2, \ldots, L$. 

However, the QISP formula does not hold, 
if one of the transfer matrices appearing in it is nonregular.  
Here we remark that it has the product of the inverse operators 
of the transfer matrices where the spectral parameters $\lambda$ are 
given by some of inhomogeneity parameters $w_j$.   
In fact, we can show that the spin-1/2 transfer matrix 
with $w_j=w_j^{(\ell)}$ is non-regular at 
$\lambda=w_{\ell(k-1)+1}^{(\ell)}$, 
the first rapidity of the $k$th complete $\ell$-string for 
an integer $k$ of $1 \le k \le N_s$ (see \S 3.6). 
Consequently, the QISP formula does not hold in the straightforward form 
for the fusion method. 
Thus, we want to avoid such special values of the spectral parameter 
when we evaluate the matrix elements or expectation values 
of the higher-spin local operators through the fusion method. 
In the revised method 
we avoid directly putting the complete $\ell$-strings $w_j^{(\ell)}$ 
into the inhomogeneity parameters $w_j$, as we shall see later.  

The main result of the present paper should have several interesting 
applications in the mathematical physics of exactly solvable models. 
For instance, with the revised method we can evaluate 
the form factors of various solvable quantum spin chains 
associated with the affine quantum group \cite{DMo1}. 
Then, for the spin-$s$ XXZ spin chain 
we can derive the multiple-integral representation of correlation functions 
through the revised method, as we have mentioned in the above. 
Moreover, it should be an interesting problem to calculate some form factors 
and matrix elements of the local spin operators 
for the superintegrable chiral Potts chains through the present method 
\cite{ND-CP2,SCP}.  
Furthermore, there is another interesting possible application. 
We can construct integrable quantum impurity models such as 
consisting of one spin-$S$ site with $N$ spin-1/2 sites 
through the fusion method. We can then calculate the form factors 
for the spin-$S$ site and the correlation functions between the spin-$S$ site 
and the other spin-1/2 sites, by  applying the method in the present paper. 
These interesting topics should be discussed in separate papers.

The paper consists of the following. In section 2 we introduce 
the finite-dimensional representations of the quantum group and 
construct the monodromy matrix of the integrable higher-spin XXZ spin chain 
through the fusion method. We then introduce 
the complete $\ell$-strings, $w_j^{(\ell)}$, where  
the $k$th complete $\ell$-string $w_{\ell(k-1)+\alpha}^{(\ell)}$ 
for $1 \le \alpha \le \ell$ is given by the sequence of $\ell$ complex numbers 
shifted by $\eta$ successively, such as 
$\xi_k, \xi_k-\eta, \ldots, \xi_k-(\ell-1)\eta$.   
In section 3, we define the higher-spin elementary operators 
$E^{i, \, j}$, 
which have only one nonzero matrix element of entry $(i,j)$ 
with respect to the basis vectors and their conjugate vectors. Then,  
we explicitly derive a formula (Proposition \ref{prop:E-w}), 
by which we can reduce any given product 
of the higher-spin elementary operators 
into a sum of products of the spin-1/2 elementary operators. 
In order to prove it we derive two expressions for a 
given product of the spin-1/2 elementary operators. 
Interestingly, the overall factor for the form factor 
of a higher-spin operator depends on whether the operator 
is associated with principal grading or homogeneous grading.  
In order to make the form factors independent of the grading $w$   
we define the general spin-$\ell/2$ elementary operators 
$\widehat{E}^{i, \, j \,  (\ell \, w)}$. 
We show in \S 3.6 that the spin-1/2 transfer matrix 
with $w_j=w_j^{(\ell)}$ is non-regular 
at $\lambda=w_{\ell(k-1)+1}^{(\ell)}=\xi_k$.  
%which is the first rapidity 
%of the $k$th complete $\ell$-string ($1 \le k \le N_s$).   
In section 4, we reduce the form factor of a product of  
the spin-$\ell/2$ elementary operators  
into those of the spin-1/2 elementary operators 
(Proposition \ref{prop:<gen-E(w)>-e}).  
We introduce the {\it almost complete $\ell$-strings}, 
$w_j^{(\ell; \, \epsilon)}$ ($1 \le j \le L)$, 
a set of inhomogeneity parameters that are slightly different from 
the complete $\ell$-strings $w_j^{(\ell)}$ by the order of 
a small parameter $\epsilon$.  
We reduce the spin-$\ell/2$ form factors into a sum of 
the spin-1/2 scalar products 
by making use of the QISP formula with inhomogeneity parameters given by 
the almost complete $\ell$-strings and by sending $\epsilon$ to 0 
\cite{Assumption}. 
We can revise the expressions of the higher-spin form factors 
given in Ref. \cite{DM1} making use of Proposition 4.5 \cite{ER1}. 
In section 5, we show an explicit formula by which we can calculate 
every spin-$\ell/2$ form factor in terms of the scalar products 
of the spin-1/2 operators. Finally, in section 6, we express  
the correlation function of an arbitrary entry for 
 the integrable spin-$s$ XXZ spin chain 
in terms of a sum of the multiple integrals. 
Moreover, the normalization factors 
are systematically shown for the general spin-$\ell/2$ elementary operators 
$\widehat{E}^{i, \, j \,  (\ell \, w)}$.  
Here, the expression of the correlation functions 
is different from 
that of Ref. \cite{DM2} mainly 
with respect to the sum over the multiple integrals 
\cite{ER2}.  Due to the spin inversion symmetry, 
however, the spin-1 one point functions 
are expressed in terms of single multiple integrals, 
which are the same with those of Ref. \cite{DM2},  
.

%\newpage
%%%%%%%%%%%%%%%%%%%%%%%%%%%%%%%%%%%%%%%%%%%%%%%%%%%%%%%%%%%%%%
%
%  Section 2 
%
\setcounter{equation}{0} 
 \renewcommand{\theequation}{2.\arabic{equation}}
\section{Affine quantum group and the monodromy matrix 
%$R$-matrix
}

%%%%%%%%%%%%%%%%%%%%%%%%%%%%
% 
%    Sec 2.1
%
%%%%%%%%%%%%%%%%%%%%%%%%%%%
\subsection{Spin-$\ell/2$ representations of the quantum group $U_q(sl_2)$}

The quantum algebra $U_q(sl_2)$ 
is an associative algebra over ${\bf C}$ generated by  
$X^{\pm}, K^{\pm}$  with the following relations: 
\cite{Jimbo-QG,Jimbo-Hecke,Drinfeld}
\begin{eqnarray} 
K X^{\pm} K^{-1}  & = & q^{\pm 2} X^{\pm} \, ,  \quad 
K K^{-1} = K^{-1} K  = 1 \, , \nonumber \\
{[} X^{+}, X^{-} {]} & = &  
{\frac   {K - K^{-1}}  {q- q^{-1}} } \, . 
\end{eqnarray}
The algebra $U_q(sl_2)$ is also a Hopf algebra over ${\bf C}$ 
with comultiplication 
\begin{eqnarray} 
\Delta (X^{+}) & = & X^{+} \otimes 1 + K \otimes X^{+}  \, , 
 \quad 
\Delta (X^{-})  =  X^{-} \otimes K^{-1} + 1 \otimes X^{-} \, ,  \nonumber \\
\Delta(K) & = & K \otimes K  \, , 
\end{eqnarray} 
and antipode:  
$S(K)=K^{-1} \, , S(X^{+})= - K^{-1} X^{+} \, , S(X^{-}) = -  X^{-} K$, and   
coproduct: $\epsilon(X^{\pm})=0$ and $\epsilon(K)=1$.

%\subsection{Basis vectors of spin-$\ell/2$ representation of $U_q(sl_2)$}

Let us introduce some symbols of $q$-analogues. 
For any given integer $n$ we define the $q$-integer of $n$ 
by $[n]_q= (q^n-q^{-n})/(q-q^{-1})$, and the $q$-factorial of $n$ by 
\begin{equation} 
[n]_q ! = [n]_q \, [n-1]_q \, \cdots \, [1]_q \, .  
\end{equation}
For integers $m$ and $n$ satisfying $m \ge n \ge 0$ 
we define the $q$-binomial coefficients as follows
\begin{equation} 
\left[ 
\begin{array}{c} 
m \\ 
n 
\end{array}  
 \right]_q 
= {\frac {[m]_q !} {[m-n]_q ! \, [n]_q !}}  \, . 
\end{equation}

We first formulate the spin-1/2 representation $V^{(1)}$.  
Let $| \alpha \rangle$ ($\alpha=0, 1$) be the basis vectors. 
We have  $X^{-} | 0 \rangle = |1 \rangle $, $X^{-} | 1 \rangle = 0$, 
$X^{+} | 0 \rangle = 0$, $X^{+} | 1 \rangle = |0 \rangle$, 
and $K | \alpha \rangle = q^{1 - 2\alpha} | \alpha \rangle$ for 
$\alpha=0, 1$.  

We now construct the spin-$\ell/2$ representation $V^{(\ell)}$ of $U_q(sl_2)$ 
for a non-negative integer $\ell$,  
%It is an irreducible ($\ell$+1)-dimensional representation. 
and in the tensor product space $(V^{(1)})^{\otimes \ell}$ of 
the spin-1/2 representations $V^{(1)}$. 
Let us  introduce the basis vectors of $V^{(\ell)}$, 
$|| \ell, n  \rangle$, for $n=0, 1, \ldots, \ell$.  
First, we define the highest weight vector $||\ell, 0 \rangle$ by 
\begin{equation} 
||\ell , 0 \rangle = |0 \rangle_1 \otimes |0 \rangle_2 \otimes 
\cdots \otimes |0 \rangle_\ell \, . 
\end{equation}   
Here  $|\alpha \rangle_j$ for $\alpha=0, 1$ 
denote the basis vectors of the spin-1/2 representation defined 
on the $j$th component of the tensor product $(V^{(1)})^{\otimes \ell}$. 
We remark that 0 and 1 corresponds to up-spin, $\uparrow$, 
and down-spin, $\downarrow$, respectively.  We also remark that  
$K ||\ell, 0 \rangle = q^{\ell} || \ell, 0 \rangle$. 
We define $|| \ell, n \rangle$ for $n \ge 1$ by 
\be 
|| \ell, n \rangle = \left( \Delta^{(\ell-1)} (X^{-}) \right)^n ||\ell, 0 \rangle \,  {\frac 1 {[n]_q!}} \, . 
\ee
We then have the following: \cite{DM1}    
\begin{equation} 
|| \ell, i \rangle  = 
\sum_{1 \le a(1) < \cdots < a(i) \le \ell} 
\sigma_{a(1)}^{-} \cdots \sigma_{a(i)}^{-} ||\ell,  0 \rangle \, 
q^{a(1)+ a(2) + \cdots + a(i)  - i \ell + i(i-1)/2} \quad \mbox{for} 
\, \, i=0, 1, \ldots, \ell.  
\label{eq:|ell,i>} 
\end{equation}
Here the sum is taken over all such integers $a(1), a(2), \ldots, a(i)$ 
that satisfy $1 \le a(1) < \cdots < a(i) \le \ell$. We denote by 
$\sigma_j^{\pm}$ the Pauli matrices acting on the $j$th site.

We denote by $F(\ell, n)$ the ``square length''s of vectors $|| \ell, n \rangle$ as follows.    
\be 
F(\ell, n) =  \left( || \ell, n \rangle \right)^{T} \cdot || \ell, n \rangle . 
\ee
Here the superscript $T$ denotes the matrix transposition. 
Setting $\left( || \ell, 0 \rangle \right)^{T} \cdot || \ell, 0 \rangle =1$, 
we have 
\be 
F(\ell, n) = \left[
\begin{array}{c}
 \ell \\
  n 
\end{array} 
\right]_q \, q^{-n(\ell-n)} \, . \label{eq:normalization}
\ee
We thus define conjugate vectors $\langle \ell, j||$ by 
\be 
\langle \ell, j || = \left( || \ell, j \rangle \right)^{T} /F(\ell, j) \quad 
\mbox{for} \, \, j=0, 1, \ldots, \ell. 
\ee
Explicitly we have the following:  
\begin{equation} 
\langle \ell, j || =  
\left[ 
\begin{array}{c} 
\ell \\ 
j 
\end{array}  
 \right]_q^{-1} \, q^{j(\ell-j)} \, 
\sum_{1 \le b(1) < \cdots < b(j) \le \ell} 
\langle \ell, 0 || \sigma_{b(1)}^{+} \cdots \sigma_{b(j)}^{+} \, 
q^{b(1) + b(2) + \cdots + b(j) - j \ell + j(j-1)/2}   \, . 
\label{eq:<ell,n|}
\end{equation}

It is easy to show the normalization factor (\ref{eq:normalization}) 
by making use of the following lemma:  
\begin{lemma} For an integer $n$ with $0 < n \le \ell$ we have 
\be 
\sum_{1 \le a(1) < \cdots < a(n) \le \ell} 
q^{2a(1) + \cdots + 2 a(n)} = q^{n (\ell+1)} \, 
\left[
\begin{array}{c}
\ell \\
n
\end{array}
\right]_q \, . \label{eq:sum-q-factors}
\ee
\end{lemma} 
\begin{proof}
Let us consider the $q$-binomial theorem:
\be 
\prod_{k=1}^{\ell}(1 - zq^{2k}) = \sum_{n=o}^{\ell}
 (-1)^n z^n q^{n(n+1)} \, 
\left[
\begin{array}{c}
\ell \\
n
\end{array}
\right]_q \, . \label{eq:q-binomial}
\ee
Expanding the left hand side of (\ref{eq:q-binomial})
with respect to $z$ we have  
\be 
\prod_{k=1}^{\ell}(1 - zq^{2k})
= \sum_{n=0}^{\ell} (-z)^n \sum_{1 \le a(1) < \cdots < a(n) \le \ell} 
q^{2 a(1) + \cdots + 2 a(n)} \, .  
\ee
Hence we have (\ref{eq:sum-q-factors}). 
\end{proof} 
%%%%%%%%%%%%%%

We remark that when $q$ is complex and not real,
the conjugate vector $\langle \ell, j ||$ is different from 
the Hermitian conjugate of a given vector $|| \ell, j \rangle$.  
Thus, the pairing of $\langle \ell, j||$ and $|| \ell, k \rangle$ 
does not give a standard scalar product if $q$ is complex and not real. 
However, the Hermitian conjugate of a vector $|| \ell, j \rangle$ is not 
covariant with respect to the quantum group $U_q(sl_2)$ 
if $q$ is complex and not real, while the 
transposed vector is covariant.   
Therefore, we express the transposed vector 
as the conjugate vector $\langle \ell, j ||$.

%\subsection{Projection operators}

We now introduce the projection operator which maps the tensor product 
of the spin-1/2 representations $(V^{(1)})^{\otimes \ell}$ 
to the spin-$\ell$ representation $V^{(\ell)}$. 
In terms of the basis vectors and their conjugate vectors we define it 
by  
\be 
P^{(\ell)}= \sum_{n=0}^{\ell} || \ell, n \rangle \langle \ell, n ||.  
\label{eq:proj}
\ee
We shall denote it also 
by $P^{(\ell)}_{1 \cdots \ell}$, since it is defined on the 
tensor product $(V^{(1)})^{\otimes \ell}$.

%%%%%%%%%%%%%%%%%%%%%%%%%%%%%%%%%%%%%%%%%%%%%%%%%%%%%%%%%%%%%%
%
%  Sec 2.2
%
\subsection{Operators in the tensor product space}

Let us consider the tensor product 
$V_1^{(\ell)} \otimes \cdots \otimes V_{N_s}^{(\ell)}$ 
of the $(\ell+1)$-dimensional vector spaces $V^{(\ell)}_j$ 
with parameter $\lambda_j$ 
for $j=1, 2 , \ldots, N_s$. Here we assume  $L= \ell N_s$. 
We call the tensor product 
$V_1^{(\ell)} \otimes \cdots \otimes V_{N_s}^{(\ell)}$ 
 the quantum space. 
In the most general cases, 
we consider the tensor product 
of the auxiliary space $V_0^{(2s_0)}$  and the quantum space 
$\left( V_1^{(2s_1)} \otimes \cdots V_{r}^{(2s_{r})} \right)$ 
such as 
$V_0^{(2s_0)} \otimes \left( V_1^{(2s_1)} \otimes \cdots 
\otimes V_{r}^{(2s_{r})} \right)$ with $2s_1 + \cdots + 2s_r = L$,  
where $V_{j}^{(2s_j)}$ have spectral parameters 
$\lambda_j$ for $j=1, 2, \ldots, r$. Here 
$s_j$ are given by integers or half-integers for $j=0, 1, 2, \ldots, r$.      

We denote by $e^{a, \, b}$ such two-by-two matrices 
that have only one nonzero element equal to 1 
at the entry of $(a, b)$ for $a, b= 0, 1$. 
We also express by $E^{a, \, b \, (2s)}$ the 
$(2s+1)$-by-$(2s+1)$ matrices with 
unique nonzero element 1 at the entry of $(a, b)$ 
for $a, b= 0, 1, \ldots, 2s$. If it is defined 
on the $j$th component of the quantum space,
 we denote it by $E^{a, \, b \, (2s)}_j$.  
For a given set of matrix elements 
${\cal A}^{a, \, \alpha}_{b, \, \beta}$  for $a,b=0, 1, \ldots, 2s_j$ and  
$\alpha, \beta=0, 1, \ldots, 2s_k$,  
we define operators ${\cal A}_{j, k}$ by 
\bea 
{\cal A}_{j, k} & = & \sum_{a,b=1}^{2s_j} \sum_{\alpha, \beta=1}^{2s_k} 
{\cal A}^{a, \, \alpha}_{b, \, \beta} I^{(2s_0)} \otimes
I^{(2s_1)} \otimes \cdots \otimes I^{(2s_{j-1})}  \non \\ 
& & \quad \otimes E^{a, \, b \, (2s_j)} \otimes 
I^{(2s_{j+1})}  \otimes \cdots \otimes I^{(2s_{k-1})}  
\otimes E^{\alpha, \, \beta \, (2s_k)} \otimes I^{(2s_{k+1})} 
 \otimes \cdots \otimes I^{(2s_r)} \, .   
\label{eq:defAjk}
\eea
Similarly, for a set of matrix elements 
${\cal B}^{a}_{b}$  for $a, b=0, 1, \ldots, 2s_j$,  
we define operators ${\cal B}_{j}$ by 
\be 
{\cal B}_{j} = \sum_{a,b=1}^{2s_j} 
{\cal B}^{a}_{b} I^{(2s_0)} \otimes
I^{(2s_1)} \otimes \cdots \otimes I^{(2s_{j-1})} 
\otimes E^{a, \, b \, (2s_j)} \otimes I^{(2s_{j+1})} 
\otimes  \cdots \otimes I^{(2s_r)} \, .   
\label{eq:defBj}
\ee

%%%%%%%%%%%%%%%%%%%%%%%%%%%%%%%%%%%%%%%%%%%%%%%%%%%%%%%%%%%%%%
%
%  Sec 2.3 
%
\subsection{$R$-matrix and the monodromy matrix of types $(1, 1^{\otimes L})$}

Let us introduce the $R$-matrix of the XXZ spin chain 
\cite{Korepin,MS2000,KMT1999,KMT2000}. 
Let $V_1$ and $V_2$ be two-dimensional vector spaces. 
The $R$-matrix acting on $V_1 \otimes V_2$ associated with homogeneous grading 
of type $w=+$ is given by  
\be 
{R}^{(1 \, +)}_{1 2} (\lambda_1-\lambda_2) = \sum_{a,b,c,d=0,1} 
R^{(1 \, +)}(u)^{a \, b}_{c \, d} \, 
\, e_1^{a, \, c} \otimes e_2^{b, \, d} = 
\left( 
\begin{array}{cccc}
1 & 0 & 0 & 0 \\ 
0 & b(u) & c^{-}(u) & 0 \\
0 & c^{+}(u) & b(u) & 0 \\
0 & 0 & 0 & 1 \\
\end{array} 
\right)_{[1,2]} \, \label{eq:R+},  
\ee
where $u=\lambda_1-\lambda_2$, 
$b(u) = \sinh u/\sinh(u + \eta)$ and 
$c^{\pm}(u) = \exp( \pm u) \sinh \eta/\sinh(u + \eta)$. 
Here, the suffix $[1,2]$ 
in eq. (\ref{eq:R+}) denotes that the matrix acts on 
the tensor product of $V_1$ and $V_2$. 

We remark that the $R$-matrix associated with homogeneous grading 
of type $w=-$, ${R}^{(1 \, -)}_{1 2} (\lambda_1-\lambda_2)$,
 is given by exchanging all the $\pm$ signs in (\ref{eq:R+}) 
\cite{DM1,DM2}. 

In the massless regime, we set $\eta= i \zeta$ by a real number  
$\zeta$,   
and we have $\Delta= \cos \zeta$. In the paper we mainly consider 
the region $0 \le \zeta < \pi/2s$ for the correlation functions.  
In the massive regime, we assign $\eta$ a real nonzero number 
and we have $\Delta = \cosh \eta > 1$.

We denote by $R^{(1 \,  p)}(u)$ or simply by $R(u)$ the symmetric $R$-matrix 
where $c^{\pm}(u)$ of (\ref{eq:R+}) are replaced by 
$c(u)= \sinh \eta/\sinh(u+\eta)$ \cite{DM1}. The symmetric $R$-matrix 
is compatible with principal grading  of 
the affine quantum group $U_q(\widehat{sl}_2)$ \cite{DM1}.

Let us now consider the $(L+1)$th tensor product 
of the spin-1/2 representations, 
which consists of the tensor product of the auxiliary space $V_0^{(1)}$ and 
the quantum space which is given by the $L$th tensor product of $V_j^{(1)}$ 
for $j=1, 2, \ldots, L$; i.e., 
$V_0^{(1)} \otimes \left( V_1^{(1)} \otimes \cdots \otimes V_L^{(1)} \right)$. 
We call the tensor product that of type $(1,1^{\otimes L})$.

Let $\{ w_j \}_L$ denote the set of free parameters 
$w_{j}$ for $j=1 , 2, \ldots, L$. We call them inhomogeneity parameters. 
We define the spin-1/2 XXZ monodromy matrices 
associated with grading of types $w= \pm, p$ by 
\be 
T^{(1, \, 1 \, w)}_{0, \, 1 2 \ldots L}(\lambda; \{ w_j \}_L)=
R^{(1 \, w)}_{0, 1 2 \cdots L} 
 = R^{(1 \, w)}_{0 L}(\lambda-w_L) \cdots R^{(1 \, w)}_{0 1}(\lambda-w_1) \, . 
\ee
Here $R^{(1 \, w)}_{jk}$ denote 
the $R$-matrices associated with grading of types $w = \pm$ and $p$  
with inhomogeneity parameters $\{w_j \}_L$.

%%%%%%%%%%%%%%%%%%%%%%%%%%%%%%%%%%%%%%%%%%%%%%%%%%%%%%%%%%%%%%
%
%  S 2.4
%
\subsection{Rapidities forming $n$-strings}

We introduce a set of rapidities for a positive integer $n$, 
which we call a complete $n$-string.  
It is given by a set of $n$ rapidities of the following form: 
\be 
\lambda_j = \Lambda + (n - 2j +1) \eta/2 \, , \quad \mbox{for} \, \, 
j=1, 2, \ldots, n. 
\ee
Here we call parameter $\Lambda$ the center of the complete $n$-string. 

Let $\epsilon$ be an infinitesimally small number; i.e., we have 
$| \epsilon| \ll 1$. We take generic parameters $r_b$ for $b=1, 2, \ldots, n$. 
We define an ``almost complete $n$-string'' by 
the following set of $n$ rapidities 
\be 
\lambda_j = \Lambda + (n - 2j +1) \eta/2 + \epsilon r_b 
\, , \quad \mbox{for} \, \, 
j=1, 2, \ldots, n. 
\ee 
They are different from the complete $n$-string with center $\Lambda$ 
by the small numbers $\epsilon r_b$.  

We introduce  $N_s$ sets of almost complete $\ell$-strings 
 $w_j^{(\ell; \, \epsilon)}$ for $1 \le j \le L$ as follows. 
\be 
w_j^{(\ell; \, \epsilon)} = \xi_b - (\beta-1) \eta + \epsilon r_b^{(\beta)} 
\quad \mbox{ for} \, \, \beta= 1, 2, \ldots, \ell; \, b=1, 2, \ldots, N_s .  
\label{eq:complete-string}
\ee
Here we assume that parameters $r_b^{(\beta)}$ are generic. 
We recall $\ell N_s = L$.

When $\epsilon=0$ the set of inhomogeneity parameters 
$w_j^{(\ell; \, \epsilon)}$ gives $N_s$ pieces of complete $\ell$-strings. 
We denote them by $w_j^{(\ell)}$; i.e., $w_j^{(\ell)}=w_j^{(\ell; \, 0)}$ 
for $j=1, 2, \ldots, L$. 
We shall put $w_j^{(\ell; \, \epsilon)}$ 
and  $w_j^{(\ell)}$ into the homogeneous parameters 
$w_j$ ($1 \le j \le L$) of the spin-1/2 transfer matrix defined 
on the $L$ sites, later.

%%%%%%%%%%%%%%%%%%%%%%%%%%%%%%%%%%%%%%%%%%%%%%%%%%%%%%%%%%%%%%
%
%  Sec 2.5
%
\subsection{Monodromy matrix of type  $(1, \ell^{\otimes N_s})$ }

We shall now construct 
the spin-$\ell/2$ monodromy matrix on the spin-$\ell/2$ chain 
with $N_s$ sites. We define $L$ by $L=\ell N_s$. We consider 
the tensor product of the spin-1/2 auxiliary space $V^{(1)}$ and 
the $N_s$th tensor product $(V^{(\ell)})^{\otimes N_s}$ 
of quantum spaces $V^{(\ell)}$; i.e.  
$V^{(1)} \otimes (V^{(\ell)})^{\otimes N_s}$.  
Here we construct $(V^{(\ell)})^{\otimes N_s}$ 
in $(V^{(1)})^{\otimes L}$.

In the fusion construction of the higher-spin monodromy matrices  
it is useful to 
introduce the following symbols for the spin-1/2 monodromy matrices:  
%$T^{(\ell \, w; \, \epsilon)}_{0, \, 1 2 \cdots L}(\lambda)$ with $w= \pm, p$:
\be 
T^{(1, \, \ell \, w; \, \epsilon)}_{0, \, 1 2 \cdots L}(\lambda) 
= T^{(1, \, 1 \, w)}(\lambda; \{w_j^{(\ell; \, \epsilon)} \}_L) \, , \quad  
\mbox{for} \, \, w= \pm, p,   
\ee
where we put $w_j=w_j^{(\ell; \, \epsilon)}$ for $j=1, 2, \ldots, L$.  
It is nothing but the spin-1/2 monodromy matrix $T^{(1, \, 1 \, w)}(\lambda)$ 
with special inhomogeneity parameters.  

We introduce the following symbol: 
\be 
T^{(1, \, \ell \, +; \, 0)}(\lambda) 
= \lim_{\epsilon \rightarrow 0} T^{(1, \, \ell \, + ; \, \epsilon)}(\lambda) . 
\ee 
We thus express 
the (0,1)-element of the spin-1/2 monodromy matrix 
$T^{(1, \, \ell \, +; \, 0)}(\lambda)$ as 
\be 
B^{(\ell \, +; \, 0)}(\lambda) = 
\left(T^{(1, \, \ell \, + ; \, 0)}(\lambda) \right)_{0, \, 1} .  
\ee

Let us denote by $P^{(\ell)}_{\ell (b-1)+1}$ 
the projection operator which maps 
the tensor product of the spin-1/2 representations 
$V_{\ell (b-1)+1}^{(1)} \otimes \cdots \otimes V_{\ell (b-1)+\ell}^{(1)}$ 
to the $b$th component of the $N_s$th tensor product 
$(V^{(\ell)})^{\otimes N_s}$.  
Here $b$ is an integer satisfying $1 \le b \le N_s$, and 
the tensor product  
$V_{\ell (b-1)+1}^{(1)} \otimes \cdots \otimes V_{\ell (b-1)+\ell}^{(1)}$ 
corresponds to the $\ell (b-1)+1$th to $\ell (b-1)+\ell$th 
components of the $L$th tensor product $(V^{(1)})^{\otimes L}$.  
We define $P^{(\ell)}_{1 \cdots L}$ by 
\be 
P^{(\ell)}_{1 \cdots L}= \prod_{b=1}^{N_s} 
P^{(\ell)}_{\ell (b-1)+1} \, . 
\ee 
We construct the spin-$\ell/2$ monodromy matrix 
$T^{(1, \, \ell\, +)}_{0, \, 1 2 \cdots N_s}(\lambda)$ 
associated with homogeneous grading 
by applying the projection operator $P^{(\ell)}_{1 \cdots L}=
\prod_{b=1}^{N_s} P^{(\ell)}_{\ell (b-1)+1}$ as follows 
\cite{DM1}. 
\be 
T^{(1, \, \ell \, +)}_{0, \, 1 2 \cdots N_s}
(\lambda; \{ \xi_b \}_{N_s})= P^{(\ell)}_{1 \cdots L} 
T^{(1, \, \ell \, + ; \, 0)}_{0, \, 1 2 \cdots L}(\lambda) 
P^{(\ell)}_{1 \cdots L} \, . 
\ee

%%%%%%%%%%%%%%%%%%%%%%%%%%%%%%%%%%%%%%%%%%%%%%%%%%%%%
% 
%  Sec 2.6
%
%%%%%%%%%%%%%%%%%%%%%%%%%%%%%%%%%%%%%%%%%%%%%%%%%%%%
\subsection{Higher-spin monodromy matrix of type 
$(\ell, \, (2s)^{\otimes N_s})$ }

We set the inhomogeneity parameters $w_j$ for $j=1, 2, \ldots, L$,  
as $N_s$ sets of complete $2s$-strings \cite{DM1}. 
We define $w_{(b-1)\ell+ \beta}^{(2s)}$ for $\beta = 1, \ldots, 2s$, 
as follows.  
\be 
w_{2s(b-1)+ \beta}^{(2s)} = \xi_b - (\beta-1) \eta \, , \quad 
 \mbox{for} \quad b = 1, 2, \ldots, N_s . 
\label{eq:ell-strings}
\ee

We now define the monodromy matrix of type $(1, (2s)^{\otimes N_s})$ 
associated with homogeneous grading.  
We first recall the definition of the monodromy matrix as follows.   
\bea 
{T}^{(1, \, 2s \, +)}_{0, \, 1 2 \cdots N_s}
(\lambda_0; \{ \xi_b \}_{N_s} ) 
& = & {P}_{12 \cdots L}^{(2s)} 
R_{0, \, 1 \ldots L}^{(1, \, 1 \, +)} 
(\lambda_0; \{ w_{j}^{(2s)} \}_L) 
{P}_{12 \cdots L}^{(2s)}  \non \\  
& = & 
\left( 
\begin{array}{cc} 
{A}^{(2s +)}(\lambda; \{ \xi_b \}_{N_s}) & 
{B}^{(2s +)}(\lambda; \{ \xi_b \}_{N_s}) \\ 
{C}^{(2s +)}(\lambda; \{ \xi_b \}_{N_s}) & 
{D}^{(2s +)}(\lambda; \{ \xi_b \}_{N_s})  
\end{array} 
\right) \, . 
\eea

We shall now define the monodromy matrix of type 
$(\ell, \, (2s)^{\otimes N_s})$ associated with homogeneous grading. 
It acts on the tensor product of the auxiliary space $V_{a_1 \cdots a_{\ell}}$ 
and the quantum space $V_1^{(2s)} \otimes \cdots \otimes V_{N_s}^{(2s)}$.   
Let us express the tensor product  
$V_0^{(\ell)} \otimes \left( V_1^{(2s)} \otimes \cdots 
\otimes V_{N_s}^{(2s)} \right)$  by the following symbol   
\be 
(\ell, \, (2s)^{\otimes N_s}) 
= (\ell, \, \overbrace{2s, 2s, \ldots, 2s}^{N_s}) \, . 
\ee
Here we recall that $V_0^{(\ell)}$ abbreviates 
$V_{a_1 a_2 \ldots a_{\ell}}^{(\ell)}$.  
For the auxiliary space $V_0^{(\ell)}$ 
we define the monodromy matrix of type 
$(\ell, \, (2s)^{\otimes N_s})$ by 
\be 
{T}^{(\ell, \, 2s \, +)}_{0, \, 1 2 \cdots N_s} 
 =  {P}^{(\ell)}_{a_1 a_2 \cdots a_{\ell}} \,   
{T}_{a_1, \, 1 2 \cdots N_s}^{(1, \, 2s \, +)}(\lambda_{a_1}) 
{T}_{a_2, \, 1 2 \cdots N_s}^{(1, \, 2s \, +)}(\lambda_{a_1}-\eta) 
\cdots 
{T}_{a_{\ell}, \, 1 2 \cdots N_s}^{(1, \,  2s \, +)}
(\lambda_{a_1}-(\ell-1)\eta) \,  
{P^{(2s)}}_{a_1 a_2 \cdots a_{\ell}} \, . 
\ee 
Here we remark that it is associated with homogeneous grading.

%%%%%%%%%%%%%%%%%%%%%%%%%%%%%%%%%%%%%%%%%%%%%%%%%%%%%%%%%%%%%%
%
%  Sec 2.7 
%
\subsection{Gauge transformations}

Let us consider a two-by-two diagonal matrix 
$\Phi(w)={\rm diag}(1, \exp(w))$. 
In terms of the tensor-product notation (\ref{eq:defBj}) 
we introduce $\Phi_j(w)$ for $j=0, 1, \ldots, L$, 
in the tensor product $(V^{(1)})^{\otimes L}$ 
of the spin-1/2 representations $V^{(1)}$.  
We define the gauge transformation $\chi_{1 2 \cdots L}$ by
\be 
\chi_{1 2 \cdots L} = \Phi_1(w_1) \Phi_2(w_2) \cdots \Phi_{L}(w_L) \, . 
\ee
Here $w_j$ denote the inhomogeneity parameters 
of the spin-1/2 transfer matrix of the XXZ spin chain 
for $j=1, 2, \ldots, L$.  
%where  denote diagonal two-by-two matrices 
%$\Phi_j={\rm diag}(1, \exp(w_j))$ acting on the $j$th component 
%$V_j^{(1)}$ of the tensor product 
%$(V^{(1)})^{\otimes \ell}$ for $j=1, \ldots, L$. 

In the $N_s$th tensor product 
of the spin-$\ell/2$ representations,  
%$V^{(\ell)}$ 
$(V^{(\ell)})^{\otimes N_s}$,  
which we call the quantum space of the spin-$\ell/2$ XXZ spin chain, 
we now introduce the gauge transformation $\chi^{(\ell)}_{1 2 \cdots N_s}$
for the spin-$\ell/2$ monodromy matrix $T^{(1, \, \ell \, +)}(\lambda)$  
with inhomogeneity parameters $w_{\ell(k-1)}^{(\ell)}$ 
for $k=1, 2, \ldots, N_s$. Here we recall 
that they are given by the complete $\ell$-strings 
(\ref{eq:complete-string}) with parameters $\xi_b$ for $b=1, 2, \ldots, N_s$. 
We introduce the $(\ell+1)$-by-($\ell+1$) diagonal matrix $\Phi^{(\ell)}(w)$ by \be 
\Phi^{(\ell)}(w)  \, || \ell, n \rangle 
= \exp({n w}) \, || \ell, n \rangle  
\quad \mbox{for} \, \, n=0, 1, \ldots, \ell.  
\ee
We then define the spin-$\ell/2$ gauge transformation 
$\chi^{(\ell)}_{1 2 \cdots N_s}$ on 
the quantum space $(V^{(\ell)})^{\otimes N_s}$ by 
\be 
\chi^{(\ell)}_{1 2 \cdots N_s} = \Phi^{(\ell)}_1(\Lambda_1) 
\Phi^{(\ell)}_2(\Lambda_2) 
\cdots \Phi^{(\ell)}_{N_s}(\Lambda_{N_s}) \, ,   
\ee 
where $\Lambda_b$ denote the string centers as follows:  
\be
\Lambda_b= \xi_b -(\ell-1)\eta/2   \, , \quad \mbox{for} \, \, 
b=1, 2, \ldots, N_s . 
\ee
Considering the tensor product of the auxiliary space $V^{(1)}$ 
and the quantum space $(V^{(\ell)})^{\otimes N_s}$ we define   
the gauge transformation $\chi^{(1, \, \ell)}_{0, \, 1 2 \cdots N_s}$ 
on $V^{(\ell)} \otimes (V^{(\ell)})^{\otimes N_s}$ by 
\be 
\chi^{(1, \, \ell)}_{0, \, 1 2 \cdots N_s} 
= \Phi_0 \Phi^{(\ell)}_1 \cdots \Phi^{(\ell)}_{N_s} .  
\ee 
Similarly, we define $\chi_{0, \, 1 2 \cdots L}$ by  
$\chi_{0, \, 1 2 \cdots L} = \Phi_0(\lambda) \Phi_1(w_1) \cdots \Phi_{L}(w_L)$. 

%%%%%%%%%%%%%%%%%%%%%%%%%%%%%%%%%%%%%%%%%%%%%%%%%%%%%%%%%%%%%%
%
%  Sec 2.8
%
\subsection{Spin-$\ell/2$ monodromy matrices associated with principal grading}

We now construct the higher-spin monodromy matrix 
of type $(1, (\ell)^{\otimes N_s})$ associated with principal grading. 
We denote it by $T^{(1, \, \ell \, p)}_{0, \, 1 2 \cdots N_s}(\lambda)$, which 
 acts on the quantum space $V_{1}^{(\ell)} \otimes \cdots \otimes V_{N_s}^{(\ell)}$. 
In the fusion construction we define it by 
\bea 
T^{(1, \, \ell \, p)}_{0, \, 1 2 \cdots N_s}(\lambda) & = & 
\left( \chi^{(1, \, \ell)}_{0, \, 1 2 \cdots N_s} \right)^{-1} 
T^{(1, \, \ell \, +)}_{0, \, 1 2 \cdots N_s}(\lambda) 
\left( \chi^{(1, \, \ell)}_{0, \, 1 2 \cdots N_s} \right)
\non \\ 
& = & 
\left( \chi^{(1, \, \ell)}_{0, \, 1 2 \cdots N_s} \right)^{-1} 
\, 
\left( P^{(\ell)}_{1 2 \cdots L} 
T^{(1, \, \ell \, +; \, 0)}_{0, \, 1 2 \cdots L}(\lambda)  
P^{(\ell)}_{1 2 \cdots L} \right) \, 
\chi^{(1, \, \ell)}_{0, \, 1 2 \cdots N_s} \, . 
\eea

Let us construct the higher-spin monodromy matrix 
of type $(\ell, (2s)^{\otimes N_s})$
associated with principal grading, which acts 
on the quantum space $V_{1}^{(2s)} \otimes \cdots \otimes V_{N_s}^{(2s)}$. 
From the higher-spin monodromy matrices associated with homogeneous grading 
we derive them through the inverse 
of the gauge transformation as follows 
\be 
T^{(\ell, \, 2s \, p)} 
= \left(\chi_{a_1 \cdots a_{\ell}, \, 
1 2 \ldots N_s}^{(\ell, \, 2s)} \right)^{-1} 
T^{(\ell, \, 2s \, +)}(\lambda) 
\left(\chi_{a_1 \cdots a_{\ell}, \, 1 2 \ldots N_s}^{(\ell, \, 2s)} \right) \, .  \ee
Here $\chi^{(\ell, \, 2s)}_{a_1 \cdots a_{\ell},  \, 1 2 \ldots N_s}$ 
denote the following: 
\be  
\chi^{(\ell, \, 2s)}_{a_1 \cdots a_{\ell}, \, 1 2 \ldots N_s}
= \Phi^{(\ell)}_{a_1 \cdots a_{\ell}}(\Lambda_0) 
\Phi^{(2s)}_1(\Lambda_1) \cdots 
\Phi^{(2s)}_{N_s}(\Lambda_{N_s}) \, ,  
\ee 
where $\Lambda_0$ denotes the string center, 
 $\Lambda_0=\lambda_{a_1} - (\ell-1) \eta/2$. 

Hereafter we shall denote $T^{(1, \, \ell, \, w)}(\lambda)$ 
by $T^{(\ell, \, w)}(\lambda)$, briefly.

%%%%%%%%%%%%%%%%%%%%%%%%%%%%%%%%%%%%%%%%%%%%%%%%%%%%%%%%%%%
%
%  Section 3 
%
\setcounter{equation}{0} 
\renewcommand{\theequation}{3.\arabic{equation}}
\section{Reduction of higher-spin elementary operators}
%%%%%%%%%%%%%%%%%%%%%%%%%%%%%%%%
%
% Sec 3.1
%
\subsection{Spin-$\ell/2$ elementary operators associated 
with homogeneous and principal gradings}

Let us consider the spin-$\ell/2$ representation $V^{(\ell)}$ constructed 
in the $\ell$th tensor product space $(V^{(1)})^{\otimes \ell}$. 
We define the spin-$\ell/2$ elementary operators associated 
with homogeneous grading, 
$E^{i, \, j \, (\ell \, +)}$, by 
\be 
E^{i, \, j \, (\ell \, +)} =  || \ell, i \rangle \langle \ell, j || 
\quad {\rm for} \, \, i,j = 0, 1, \ldots, \ell. 
\ee 
We define the spin-$\ell/2$ elementary matrices associated 
with principal grading, $E^{i, \, j \, (\ell \, p)}$, also by 
\be 
E^{i, \, j \, (\ell \, p)} =  || \ell, i \rangle \langle \ell, j || 
\quad {\rm for} \, \, i,j = 0, 1, \ldots, \ell. 
\ee 
In the paper we define it by the same operator 
as that of homogeneous grading. We have 
\be 
E^{i, \, j \, (\ell \, +)} = 
E^{i, \, j \, (\ell \, p)} =  || \ell, i \rangle \langle \ell, j || \, . 
\ee

Through (\ref{eq:proj}), which expresses the projection operator 
in terms of the basis vectors and conjugate vectors,  
we have the following:  
\begin{lemma} In the $(\ell+1)$-dimensional representation $V^{(\ell)}$ 
for the spin-$\ell/2$ elementary operators with grading of $w=\pm, p$ 
and the spin-$\ell/2$ projection operator we have 
\be 
 P^{(\ell)} E^{i, \, j \, (\ell \, w)} 
= E^{i, \, j \, (\ell \, w)} P^{(\ell)} 
= E^{i, j \, (\ell \, w)} \, . 
\ee
\label{lem:PE=EP=E}
\end{lemma}

Let us recall that we have set $L=\ell N_s$, and 
the quantum space $(V^{(\ell)})^{\otimes N_s}$ is constructed in the 
$L$th tensor product space $(V^{(1)})^{\otimes L}$. 
We now introduce the spin-$\ell/2$ elementary operators 
associated with grading of $w$,  $E_k^{i, \, j \, (\ell \, w)}$,  
acting on the $k$th component of the quantum space 
$(V^{(\ell)})^{\otimes N_s}$ as follows.  
\be 
E_k^{i, \, j \, (\ell \,  w)} = (I^{(\ell)})^{\otimes (k-1)} \otimes  
E^{i, \, j \, (\ell \,  w)} \otimes (I^{(\ell)})^{\otimes (N_s-k)} \, 
\quad \mbox{for} \, \, k= 1, 2, \ldots, N_s. 
\ee

%%%%%%%%%%%%%%%%%%%%%%%%%%%%%%%%%%%%%%%%%%%%%
%
% Sec 3.2
%
\subsection{Two expressions of a product of spin-1/2 elementary operators} 

For a given product of the spin-1/2 elementary operators 
we shall express it in  another form.   
Let us first consider the simplest example. 
In terms of the highest weight vector  
$|| \ell, 0 \rangle = | 0 \rangle_1 \otimes \cdots \otimes | 0 \rangle_{\ell}$  we have the following: 
\bea  
|| \ell, 0 \rangle \langle \ell, 0 || & = &  
| 0 \rangle_1 \otimes \cdots \otimes | 0 \rangle_{\ell} \, \, 
\langle 0|_1 \otimes \cdots \otimes  \langle 0 |_{\ell} \non \\ 
& = &  
| 0 \rangle_1 \, \langle 0|_1 
\otimes \cdots \otimes | 0 \rangle_{\ell} \, \langle 0 |_{\ell}  
\non \\ 
& = & e_1^{0, \, 0} \cdots e_{\ell}^{0, \, 0} \, .  \label{eq:two-exp} 
\eea
Thus, the product of the spin-1/2 elementary operators 
$e_1^{0, \, 0} \cdots e_{\ell}^{0, \, 0}$ is also expressed as 
$|| \ell, 0 \rangle \langle \ell, 0 ||$. Here we remark that 
$\langle 0 |_1 = (1, 0)$ and $|0 \rangle_1=(1, 0)^{T}$, where 
the superscript ${T}$ denotes the 
matrix transposition.  We thus have 
\be 
| 0 \rangle_1 \langle 0 |_1 = (1, 0)^{T} (1, 0) = 
\left( 
\begin{array}{cc} 
1 & 0 \\ 
0 & 0 
\end{array}
\right) 
= e^{0, \, 0} \, . 
\ee  
We shall generalize relation (\ref{eq:two-exp}) in the following. 

Let us introduce symbols for expressing sequences.  
If a sequence of numbers, $a_1, a_2, \ldots, a_N$, are given,    
we denote it by $(a_j)_{N}$, briefly; i.e., we have 
\be 
(a_j)_{N} = (a_1, a_2, \ldots, a_N) \, .  
\ee
Here we recall that we denote by $\{ \mu_k \}_N$ 
a set of $N$ parameters $\mu_k$; i.e., $\mu_1, \mu_2, \ldots, \mu_N$.  

We now consider two sequences consisting of only two values 0 or 1,  
$(\varepsilon_{\alpha}^{'})_{\ell}$ and 
$(\varepsilon_{\beta})_{\ell}$.  Here, 
the values of $\varepsilon_{\alpha}^{'}$ 
and $\varepsilon_{\beta}$ are given by 0 or 1 
for $\alpha, \beta = 1, 2, \ldots, \ell$. 
For given such sequences 
$(\varepsilon_{\alpha}^{'})_{\ell}$ and $(\varepsilon_{\beta})_{\ell}$ 
%$\{ \varepsilon_{\alpha}^{'}; \, 1 \le \alpha \le \ell \}$ 
%and $\{\varepsilon_{\beta}; \, 1 \le \beta \le \ell \}$ 
we consider the following product of the spin-1/2 elementary operators: 
\be 
\prod_{k=1}^{\ell} e_k^{\varepsilon_k^{'}, \, \varepsilon_k} = 
e_1^{\varepsilon_1^{'}, \, \varepsilon_1} \cdots 
e_{\ell}^{\varepsilon_{\ell}^{'}, \, \varepsilon_{\ell}} \, . 
\label{eq:prod-em}
\ee  
Here we recall that $e_k^{\varepsilon^{'}, \, \varepsilon}$ 
for $\varepsilon^{'}, \varepsilon=0, 1$ 
denote the two-by-two matrices defined on the $k$th sites 
 with unique nonzero element 1 at the entry 
$(\varepsilon^{'}, \varepsilon)$ 
for integers $k$ satisfying $1 \le k \le \ell$.

Let us give another expression of product (\ref{eq:prod-em}). 
We define a set $\mbox{\boldmath$\alpha$}^{-}$ by the set of integers $k$ 
satisfying $\varepsilon_k^{'}=1$ for $ 1 \le k \le \ell$ and a set 
$\mbox{\boldmath$\alpha$}^{+}$ by 
the set of integers $k$ satisfying $\varepsilon_k=0$ for 
$ 1 \le k \le \ell$, respectively:  
\begin{equation} 
\mbox{\boldmath$\alpha$}^{-}(\{ \varepsilon_{\alpha}^{'} \}) 
= \{ \alpha ; \, \varepsilon_{\alpha}^{'}=1 \, 
(1 \le \alpha \le \ell) \}  
\, , \quad 
\mbox{\boldmath$\alpha$}^{+}(\{ \varepsilon_{\beta} \}) 
= \{ \beta ; \, \varepsilon_{\beta} = 0  \, (1 \le \beta \le \ell) \}  \, . 
\label{eq:def-sets-aa'}
\end{equation} 
Let us denote by $\Sigma_{\ell}$ 
the set of integers $1, 2, \ldots, \ell$; 
i.e., $\Sigma_{\ell}=\{1, 2, \ldots, \ell\}$.
In terms of sets ${\bm \alpha}^{\pm}$ we express the product of 
elementary operators given by (\ref{eq:prod-em}) as 
\be 
\prod_{a \in {\bm \alpha}^{-}} \sigma_a^{-}
||\ell,  0 \rangle \, 
\langle \ell, 0 || \prod_{b \in \Sigma_{\ell} \setminus {\bm \alpha}^{+}} 
\sigma_b^{+} \, . \label{eq:2nd-exp}
\ee

We now derive the expression of (\ref{eq:prod-em}) 
from that of (\ref{eq:2nd-exp}), in detail. 
Let us denote by $r$ and $r^{'}$ the number 
of elements of the set $\mbox{\boldmath$\alpha$}^{-}$ 
and $\mbox{\boldmath$\alpha$}^{+}$, respectively. 
%We have $r=i$ and $r^{'}=\ell-j$. 
%
We express the elements of $\mbox{\boldmath$\alpha$}^{-}$ as 
 $a(k)$ for $k=1, 2, \ldots, r$,  and those of 
$\Sigma_{\ell} \setminus {\bm \alpha}^{+}$ as 
 $b(k)$ for $k=1, 2, \ldots, r^{'}$, respectively. 
Expressing $r$ and $\ell-r^{'}$ by $i$ and $j$, respectively,  
we have 
\be 
{\bm \alpha}^{-} = \{ a(1), a(2), \ldots, a(i) \} \, , \quad 
\Sigma_{\ell} \setminus 
{\bm \alpha}^{+} = \{ b(1), b(2), \ldots, b(j) \} \, . 
%{eq:def-sets-aa'}
\label{eq:elements-aa'} 
\ee 
Hereafter if not specified, we shall put them in increasing order: 
$1 \le a(1) < \cdots < a(i) \le \ell$ and 
$1 \le b(1) < \cdots < b(j) \le \ell$, respectively. Here 
we recall $i=r$ and $j=\ell-r^{'}$. 
We thus express the product of the elementary operators 
in terms of $a(k)$ and $b(k)$ as follows:  
\bea  
\prod_{a \in {\bm \alpha}^{-}} \sigma_a^{-}
||\ell,  0 \rangle \, 
\langle \ell, 0 || \prod_{b \in \Sigma_{\ell} \setminus {\bm \alpha}^{+}} 
\sigma_b^{+}  
& = & 
\sigma_{a(1)}^{-} \cdots \sigma_{a(i)}^{-} 
||\ell,  0 \rangle \, 
\langle \ell, 0 || \sigma_{b(1)}^{+} \cdots \sigma_{b(j)}^{+} \,  
\non \\ 
& = & e_{a(1)}^{1, \, 0} \cdots e_{a(i)}^{1, \, 0} 
 \, e^{0, \, 0}_1 \cdots e^{0, \, 0}_{\ell} 
\,  e_{b(1)}^{0, \, 1} \cdots e_{b(j)}^{0, \, 1} 
\, .  
\label{eq:prod-e-sigma_pm}  
\eea  
Calculating products of two-by-two matrices, 
from expression (\ref{eq:2nd-exp}) 
we derive the expression in terms of products of 
the spin-1/2 elementary operators, 
$e_{1}^{\varepsilon_1^{'}, \, \varepsilon_1} 
\cdots e_{\ell}^{\varepsilon_{\ell}^{'}, \, \varepsilon_{\ell}}$,  
such as given in (\ref{eq:prod-em}).  
Here, we derive the sequence 
$( \varepsilon_{\alpha}^{'} )_{\ell}$   by setting 
$\varepsilon_{a(k)}^{'}=1$ for $k=1, 2, \ldots, i$ while 
$\varepsilon_{\alpha}^{'}=0$ for $\alpha \ne a(k)$ 
with $k$ of $1 \le  k \le i$:
\be
\varepsilon_{\alpha}^{'} = 
\left\{ 
\begin{array}{cc} 
1 & \mbox{if} \, \alpha=a(k) \, (1 \le k \le \ell) \, , \\  
0 & {\rm otherwise}   \, . 
\end{array}
\right.
\ee
Similarly, we derive sequence  
$(\varepsilon_{\beta} )_{\ell}$
%$j \in \{ 1, 2, \ldots, \ell \} \setminus \{ a(k); k=1, 2, \ldots, i \} $; 
by setting 
$\varepsilon_{b(k)}=1$ for $k = 1, 2, \ldots, j$
while $\varepsilon_{\beta}=0$ for $\beta \ne b(k)$ 
with $k$ of $1 \le k \le j$. 

Let us introduce  useful notation. Suppose that we have a sequence 
$(\varepsilon_{\alpha}^{'})_{\ell}$ 
such that $\varepsilon_{\alpha}^{'}= 0$ or 1 
for all integers $\alpha$ with $1 \le \alpha \le \ell$ 
and the number of integers $\alpha$ satisfying 
$\varepsilon_{\alpha}^{'}=1$ ($1 \le \alpha \le \ell$) is given by $i$.  
Then, we denote $\varepsilon_{\alpha}^{'}$ by $\varepsilon_{\alpha}^{'}(i)$ 
for each integer $\alpha$ and 
the sequence $(\varepsilon_{\alpha}^{'})_{\ell}$ 
 by $(\varepsilon_{\alpha}^{'}(i))_{\ell}$. 
%
% Let us denote by $(\varepsilon_{\alpha}^{'}(i))_{\ell}$ a sequence 
%of $\varepsilon_{\alpha}^{'}$ such that 
%that the number of integers $\alpha$ satisfying 
% $\varepsilon_{\alpha}^{'}(i)=1$ for $1 \le \alpha \le \ell$ is given by $i$, 
%
In the same way, we denote by $(\varepsilon_{\beta}(j))_{\ell}$ 
a sequence of 0 or 1 
such that the number of integers $\beta$ satisfying 
$\varepsilon_{\beta}(j)=1$ for $1 \le \beta \le \ell$ 
is given by $j$. 
The two expressions of a product of the spin-1/2 elementary operators 
are summarized as follows. 
 
\begin{lemma}
Sequences $(\varepsilon_{\alpha}^{'}(i))_{\ell}$ 
and $( \varepsilon_{\beta}(j) )_{\ell}$ 
are related to integers $a(1) < a(2) < \cdots < a(i)$ and 
 $b(1) < b(2) < \cdots < b(j)$, respectively,  by 
\bea 
e_1^{\varepsilon_1^{'}(i), \, \varepsilon_1(j)} \cdots 
e_{\ell}^{\varepsilon_{\ell}^{'}(i), \, \varepsilon_{\ell}(j)}
& = & 
 e_{a(1)}^{1, \, 0} \cdots e_{a(i)}^{1, \, 0} 
 \, e^{0, \, 0}_1 \cdots e^{0, \, 0}_{\ell} 
\,  e_{b(1)}^{0, \, 1} \cdots e_{b(j)}^{0, \, 1} \, ,  \label{eq:eps-ab} 
\\
\prod_{k=1}^{\ell} e^{\varepsilon_k^{'}(i), \, \varepsilon_k(j)}_k  
& = & 
\prod_{a \in {\bm \alpha}^{-}} \sigma_a^{-} \, 
||\ell,  0 \rangle \, 
\langle \ell, 0 || \, \prod_{b \in \Sigma_{\ell} \setminus {\bm \alpha}^{+}} 
\sigma_b^{+} \, . \label{eq:eps-ab-2nd}
\eea
\end{lemma}

%%%%%%%%%%%%%%%%%%%%%%%%%%%%%%%%%%%%%%%%%%%%%%%%
%
% Sec 3.3 
%
\subsection{Reduction into the spin-1/2 elementary operators}

We shall express the spin-$\ell/2$ elementary operators 
$E^{i, \, j \, (\ell \, +)}$ 
for integers $i$ and $j$ satisfying $1 \le i , j \le \ell$ 
in terms of sums of products of the spin-1/2 elementary matrices. 
It follows from (\ref{eq:|ell,i>}) and (\ref{eq:<ell,n|}) that 
%for a pair of integers $i$ and $j$ satisfying $1 \le i, j \le \ell$ 
we have 
\be 
|| \ell, i \rangle  \langle \ell, j ||   
= \sum_{(\varepsilon_{\alpha}^{'}(i))_{\ell}} 
\sum_{(\varepsilon_{\beta}(j))_{\ell}} 
g_{ij}(\varepsilon_{\alpha}^{'}(i), \, \varepsilon_{\beta}(j)) 
e_1^{\varepsilon_1^{'}(i), \, \varepsilon_1(j)} \cdots 
e_{\ell}^{\varepsilon_{\ell}^{'}(i), \, \varepsilon_{\ell}(j)} \, . 
\ee
Here the sum is taken over all sequences 
$( \varepsilon_{\alpha}^{'}(i) )_{\ell}$ and 
$( \varepsilon_{\beta}(j) )_{\ell}$. 
%Here we recall that the number of integers $\alpha$ satisfying 
%$\varepsilon_{\alpha}^{'}(i)=1$ for $1 \le \alpha \le \ell$ is equal to $i$
%and the number of integers $\beta$ satisfying $\varepsilon_{\beta}=1$
%for $1 \le \beta \le \ell$ is equal to $j$, respectively.  
%
The coefficients $g_{ij}(\varepsilon_{\alpha}^{'}(i), \, 
\varepsilon_{\beta}(j)) $ 
are  given explicitly as follows.   
\be 
g_{ij}(\varepsilon_{\alpha}^{'}(i) , \,  \varepsilon_{\beta}(j)) 
= 
 \left[ 
\begin{array}{c}
\ell \\
j 
\end{array} \right]_{q}^{-1} \, 
q^{(a(1)+ \cdots a(i)) + (b(1) + \cdots + b(j)) - (i+j) \ell 
+i(i-1)/2 + j(j-1)/2} \, . 
\ee

The ket vectors $\langle \ell, i ||$ satisfy the following symmetry, 
which plays a central role in the fusion method for evaluating the 
spin-$\ell/2$ form factors.   
\begin{lemma}  
Let ${\bm \alpha}^{-}$ be a set of distinct integers 
$\{a(1), \ldots, a(i) \}$ satisfying  
$1 \le a(1) <  \ldots < a(i) \le \ell$, 
we have the following: 
\be
\langle \ell, i || \sigma_{a(1)}^{-} \cdots   \sigma_{a(i)}^{-} 
||\ell, 0 \rangle \, 
q^{-(a(1) + \cdots + a(i)) + i } 
= \left[ 
\begin{array}{c}
\ell \\
i 
\end{array} \right]_{q}^{-1} \, 
q^{-i(i-1)/2} \, , \label{eq:key}
\ee
which is independent of the set ${\bm \alpha}^{-} =
\{a(1), a(2), \ldots, a(i)\}$. 
\label{lem:independance}
\end{lemma}
\begin{proof}
 From the explicit expression (\ref{eq:<ell,n|}) of 
the conjugate vector $ \langle \ell, i ||$ we have 
(\ref{eq:key}).  
\end{proof}

Expressing the matrix elements of the matrix $\Phi(w)$ as 
$\left( \Phi(w)\right)_{a, \, b} = \delta(a, b) \exp( a w)$ for 
$a, b=0, 1$, 
we show the gauge transformation $\chi_{1 2 \ldots \ell}$ on 
the spin-1/2 elementary operators as follows.   
\begin{lemma} 
Recall that 
$\varepsilon_{\alpha}^{'}(i)$ and $\varepsilon_{\beta}(j)$ 
are related to $a(k)$ and $b(k)$ via (\ref{eq:eps-ab}).  
Every product of the spin-$1/2$ elementary operators 
is transformed with the gauge transformation as    
\be 
\chi_{1 2 \cdots \ell} \, \, 
e_1^{\varepsilon_1^{'}(i), \, \varepsilon_1(j)} \cdots 
e_{\ell}^{\varepsilon_{\ell}^{'}(i), \, \varepsilon_{\ell}(j)} \, \, 
\chi^{-1}_{1 2 \cdots \ell} 
=
e_1^{\varepsilon_1^{'}, \, \varepsilon_1} \cdots 
e_{\ell}^{\varepsilon_{\ell}^{'}, \, \varepsilon_{\ell} } \, 
q^{-(a(1) + \cdots + a(i)-i) + (b(1) + \cdots + b(j)-j)} e^{(i-j)\xi_1} .  
\ee
\label{lem:inv-gauge} 
\end{lemma}

It is useful to express Lemma \ref{lem:independance} in the following form. 
\begin{corollary} 
For a pair of integers $i$ and $j$ with $1 \le i, j \le \ell$, 
let us consider sequences $(\varepsilon_{\alpha}^{'}(i))_{\ell}$ and 
$(\varepsilon_{\beta}(j))_{\ell}$, which correspond to sets 
${\bm \alpha}^{-}$ and ${\bm \alpha}^{+}$, respectively, 
through (\ref{eq:eps-ab}) and (\ref{eq:eps-ab-2nd}). 
%(or (\ref{eq:def-sets-aa'}) and (\ref{eq:elements-aa'})). 
% such that the number of integers $\alpha$ satisfying 
The product of the spin-1/2 elementary operators multiplied by the projection 
operator from the left and  multiplied also by $q^{-(a(1)+ \cdots + a(i))+ i}$ 
does not depend on the set ${\bm \alpha}^{-}$   
\be 
P^{(\ell)} \, e_1^{\varepsilon_1^{'}(i), \, \varepsilon_1(j)} \cdots 
e_{\ell}^{\varepsilon_{\ell}^{'}(i), \, \varepsilon_{\ell}(j)}
\, q^{-(a(1)+ \cdots + a(i))+ i} 
=
\left[ 
\begin{array}{c}
\ell \\
i 
\end{array} \right]_{q}^{-1} \, 
q^{-i(i-1)/2} || \ell, i \rangle \langle \ell, 0 || 
 \prod_{\beta \in \Sigma_{\ell} \setminus {\bm \alpha}^{+}} 
\sigma_{\beta}^{+} \, . \label{eq:independ} 
\ee
In terms of the gauge transformation  we express relation 
(\ref{eq:independ}) as 
\be 
P^{(\ell)} \, \chi_{ 1 2 \cdots \ell} \, 
e_1^{\varepsilon_1^{'}(i), \, \varepsilon_1(j)} \cdots 
e_{\ell}^{\varepsilon_{\ell}^{'}(i), \, \varepsilon_{\ell}(j)}
\chi_{ 1 2 \cdots \ell}^{-1} = 
\left[ 
\begin{array}{c}
\ell \\
i 
\end{array} \right]_{q}^{-1} \, 
q^{-i(i-1)/2} || \ell, i \rangle \langle \ell, 0 || 
 \prod_{b \in \Sigma_{\ell} \setminus {\bm \alpha}^{+}} \sigma_b^{+}
\, q^{b(1) + \cdots + b(j)-j} \, . 
\ee
\end{corollary}

\begin{lemma}  
The sum of coefficients 
$g_{ij}(\varepsilon_{\alpha}^{'}(i), \, \varepsilon_{\beta}(j))$ 
over all sequences $( \varepsilon_{\alpha}^{'}(i))_{\ell}$ 
multiplied by $q^{a(1) + \cdots + a(i)-i}$ 
is given by the following:   
\be 
\sum_{ (\varepsilon_{\alpha}^{'}(i) )_{\ell} } 
g_{ij}(\varepsilon_{\alpha}^{'}(i), \, \varepsilon_{\beta}(j)) \, 
q^{a(1) + \cdots + a(i)-i} 
= 
\left[ \begin{array}{c}
\ell \\
i 
\end{array} \right]_{q}  
\, 
\left[ \begin{array}{c}
\ell \\
j 
\end{array} \right]_{q}^{-1} 
\, 
q^{b(1) + \cdots + b(j)-j} 
q^{i(i-1)/2-j(j-1)/2} \, . \label{eq:g-sum}
\ee 
\label{lem:g-sum} 
\end{lemma}
Here we remark that we take the sum over all sequences 
of the form of $(\varepsilon_{\alpha}^{'}(i))_{\ell}$.
\begin{proof}
Putting $n=i$ in (\ref{eq:sum-q-factors}) 
and observing that the sum over sequences 
$(\varepsilon_{\alpha}^{'}(i) )_{\ell}$ corresponds to the sum over 
integers $a(1), \ldots, a(i)$ satisfying 
$1 \le a(1)< \cdots < a(i) \le \ell$, we have (\ref{eq:g-sum}) 
from (\ref{eq:sum-q-factors}). 
\end{proof}
%%%%%%%%%%%HERE

We thus show the main formula for reducing the spin-$\ell/2$ operators into 
the spin-1/2 ones. 
\begin{proposition}
For every pair of integers $i$ and $j$ with $1 \le i, j \le \ell$  
the spin-$\ell/2$ elementary operator    
associated with grading $w$, 
$E_1^{i, \, j \, (\ell \, w)}$,  
is decomposed into a sum of products of the spin-1/2 elementary operators 
as follows. 
\be
E_1^{i, \, j \, (\ell \, w)}   
= 
\left[ 
\begin{array}{c} 
\ell \\ 
i 
\end{array}  
\right]_{q}
\left[ 
\begin{array}{c} 
\ell \\ 
 j 
\end{array}  
\right]_{q}^{-1}
q^{i(i-1)/2 - j(j-1)/2} e^{-(i- j) \xi_1}   
%\nonumber \\ & & \times 
P^{(\ell)}_{1 2 \ldots \ell}  \, \sum_{ ( \varepsilon_{\beta}(j) )_{\ell} } 
\chi_{1 2 \cdots \ell} \, 
e_1^{\varepsilon_1^{'}(i), \, \varepsilon_1(j)} \cdots 
e_{\ell}^{\varepsilon_{\ell}^{'}(i), \, \varepsilon_{\ell}(j) } \, 
\chi^{-1}_{1 2 \cdots \ell} \, . \label{eq:E-w}
\ee
Here, we fix a sequence $( \varepsilon_{\alpha}^{'}(i) )_{\ell}$ . 
Furthermore, the expression (\ref{eq:E-w}) does not depend on 
the order of $\varepsilon_{\alpha}^{'}(i)$ s with respect to $\alpha$s.  
\label{prop:E-w}
\end{proposition}

We shall show the derivation of Proposition \ref{prop:E-w} 
explicitly in Appendix A.

In terms of the string center: $\Lambda_1=\xi_1 - (\ell-1)\eta/2$,
the $q$ factors in eq. 
(\ref{eq:E-w}) are expressed as follows.    
\bea 
q^{i(i-1)/2 - j(j-1)/2} e^{-(i- j) \xi_1}   
& = & q^{-i(\ell-i)/2 + j(\ell-j)/2} e^{-(i- j) (\xi_1 - (\ell-1)\eta/2)}
\non \\ 
& = & q^{-i(\ell-i)/2 + j(\ell-j)/2} e^{-(i- j) \Lambda_1} \, . 
\eea
Thus, introducing the symbol 
\be 
N_{i, \, j}^{(\ell)} 
= 
\left[ 
\begin{array}{c} 
\ell \\ 
i 
\end{array}  
\right]_{q}
\left[ 
\begin{array}{c} 
\ell \\ 
 j 
\end{array}  
\right]_{q}^{-1}
q^{-i(\ell-i)/2 + j(\ell-j)/2} \, , 
\ee
we express  (\ref{eq:E-w}) compactly as follows 
\be 
E_1^{i, \, j \, (\ell \, w)}   
 =  N_{i, \, j}^{(\ell)} 
 e^{-(i- j) \Lambda_1}   
 P^{(\ell)}_{1 2 \ldots \ell}  \, \sum_{ ( \varepsilon_{\beta}(j) )_{\ell} } 
\chi_{1 2 \cdots \ell} \, 
e_1^{\varepsilon_1^{'}(i), \, \varepsilon_1(j)} \cdots 
e_{\ell}^{\varepsilon_{\ell}^{'}(i), \, \varepsilon_{\ell}(j) } \, 
\chi^{-1}_{1 2 \cdots \ell} \, . \label{eq:E-w-N}
\ee

In Ref. \cite{DM2} the Hermitian elementary operators 
${\widetilde E}^{i, \, j \, (\ell, \, +)}$ are introduced. 
The expectation values of the Hermitian elementary operators are the same 
as those of the standard elementary operators, 
${E}^{i, \, j \, (\ell, \, +)}$.  
We shall show the reduction formula for the Hermitian 
elementary operators in Appendix B.

%%%%%%%%%%%%%%%%%%%%%%%%%%%%%%%%%%%%%%%%%%%%%%%%%%%%%%%%%%%%%%
%
%  Sec 3.4
%
\subsection{General spin-$\ell/2$ elementary operators}

Let us consider a similarity transformation 
of the basis vectors as follows 
\be 
|| \ell, m \rangle \rightarrow || \ell, m \rangle / g(m) \, , \quad 
\langle \ell, n || \rightarrow  g(n) \, \langle \ell, n ||  \, , 
\quad \mbox{for} \, \, 
m, n= 0, 1, \ldots, \ell. 
\ee
In the spin-$\ell/2$ representation constructed in the 
$\ell$th tensor product space  $(V^{(1)})^{\otimes \ell}$, 
we define the general spin-$s$ elementary operators associated with 
principal grading, ${\hat E}^{i, \, j \, (\ell \, p)}$,  
by 
\be 
{\hat E}^{i, \, j \, (\ell \, p)} = || \ell, i \rangle \, \langle \ell, j || \,
{\frac {g(j)} {g(i)}}  \, , \quad \mbox{for} \, \, 
i, j = 0, 1, \ldots, \ell. 
\ee
Then, through the spin-$\ell/2$ gauge transformation 
we define the general spin-$s$ elementary operators 
associated with homogeneous grading by 
\be 
{\hat E}^{i, \, j \, (\ell \, +)} = 
\chi_{1 2 \ldots N_s}^{(\ell)} \, 
{\hat E}^{i, \, j \, (\ell \, p)} \, 
\left( \chi_{1 2 \ldots N_s}^{(\ell)} \right)^{-1} \, .  
\ee
We explicitly have 
\be 
{\hat E}^{i, \, j \, (\ell \, +)} 
= || \ell, i \rangle \, \langle \ell, j || \,
{\frac {g(j)} {g(i)}}  \, e^{(i-j)(\xi-(\ell-1)\eta/2)}   
, \quad \mbox{for} \, \, 
i, j = 0, 1, \ldots, \ell. \label{eq:generalE(+)}
\ee
Here we recall that the quantity $\xi-(\ell-1)\eta/2$ corresponds to 
the string center of the $\ell$-string: 
$\xi, \xi-\eta, \ldots, \xi-(\ell-1)\eta$.  
They are originally the evaluation parameters 
of the $\ell$th tensor product  of the spin-1/2 representations, 
$(V^{(1)})^{\otimes \ell}$. 

We remark that the definition of the 
general elementary operators $\widehat{E}^{i, \, j \, (\ell \, w)}$ 
associated with grading $w$ are covariant under the gauge transformations. 
We also remark that if we put $g(j)= \exp( j(\xi-(\ell-1)\eta/2) )$, 
then expression (\ref{eq:generalE(+)}) reduces to that of 
$E^{i, \, j \, (\ell \, +)}$.

We define the general spin-$\ell/2$ elementary operators 
associated with principal grading acting in the tenor product space 
$V_1^{(\ell)} \otimes \cdots V_{N_s}^{(\ell)}$ by 
\be 
\widehat{ E}_k^{i, \, j \, (\ell \, p)} = (I^{(\ell)})^{\otimes (k-1)} 
\otimes \widehat{ E}^{i, \, j \, (\ell \, p)} \otimes 
(I^{(\ell)})^{\otimes (N_s -k)} 
%|| \ell, i \rangle \, \langle \ell, j || \,
%{\frac {g(j)} {g(i)}}  
\, , \quad \mbox{for} \, \, 
i, j = 0, 1, \ldots, \ell. 
\ee
Similarly we define that of homogeneous grading, 
 $\widehat{ E}_{k}^{i, \, j \, (\ell, \, +)}$ 
for $i, j = 0, 1, \ldots, \ell$.

Let us introduce the normalization factor $\widehat{ N}_{i, \, j}^{(\ell)}$ 
by ${\widehat N}_{i, \, j}^{(\ell)} = {N}_{i, \, j}^{(\ell)} g(i)/g(j)$. 
We have 
\bea 
{\widehat N}_{i, \, j}^{(\ell)} = 
{\frac {g(j)} {g(i)}} \, 
{\frac {F(\ell, i)  } {F(\ell, j)}} 
\, q^{i(\ell-i)/2-j(\ell-j)/2} .   
\eea 
We define $\delta(w, p)$ 
for gradings $\pm$ and $p$ by 
\be 
\delta(w, p) = \left\{ 
\begin{array}{cc}
1 &  \mbox{if} \, \, w=p , \\  
0 &  \mbox{otherwise} . 
\end{array} 
\right.
\ee
With factor $\widehat{ N}_{i, \, j}^{(\ell)}$ 
and the string center: $\Lambda_1= \xi_1 -(\ell-1)\eta/2$, 
from Proposition \ref{prop:E-w}, we have  
\be
\widehat{ E}_1^{i, \, j \, (\ell \, w)} =
\widehat{ N}_{i, \, j}^{(\ell)}  \,   
e^{-(i- j) \Lambda_1 \, \delta(w, p)} \, 
P^{(\ell)}_{1 2 \ldots \ell}  \, \sum_{ ( \varepsilon_{\beta}(j) )_{\ell} } 
\chi_{1 2 \cdots \ell} \, 
e_1^{\varepsilon_1^{'}(i), \, \varepsilon_1(j)} \cdots 
e_{\ell}^{\varepsilon_{\ell}^{'}(i), \, \varepsilon_{\ell}(j) } \, 
\chi^{-1}_{1 2 \cdots \ell} \, . \label{eq:E-w2}
\ee
%Here we take the sum over all sets of $\varepsilon_{\beta}$   
% such that the number of integers $\beta$ satisfying $\varepsilon_{\beta}=1$ 
%for $1 \le \beta \le \ell$ is equal to $j_k$,  
Here, we recall that sequence 
$(\varepsilon_{\alpha}^{'}(i))_{\ell}$ is fixed.  
%such that the number of integers $\alpha$ with $\varepsilon_{\alpha}^{'}=1$ 
%for $1 \le \alpha \le \ell$ is equal to $i_k$. 
We also recall that $\Lambda_1$ denotes . 

Let us recall $\Lambda_k= \xi_k - (\ell-1) \eta/2$ for $k=1, 2, \ldots, N_s$. 
We have the following. 

\begin{proposition} 
In the $k$th component of the quantum space $(V^{(\ell)})^{\otimes N_s}$ 
for $1 \le k \le N_s$  
we take integers $i_k$ and $j_k$ 
satisfying $0 \le i_k, j_k \le \ell$.  
%With some normalization factor $N^{i, \, j \, (\ell \, w)}_k$ 
The general spin-$\ell/2$ elementary operator,  
$E_k^{i_k \, j_k \, (\ell \, w)}$,  
is decomposed into a sum of products of the spin-1/2 elementary operators  
as follows 
\bea 
\widehat{ E}^{i_k \, j_k \, (\ell \, w)}_k 
%
%\left( \left[ \begin{array}{c} 
%\ell \\
%i_k 
%\end{array} 
%\right]_q 
%\left[ \begin{array}{c} 
%\ell \\
%j_k 
%\end{array} 
%\right]_q^{-1}  \,  \right)^{1/2} \, 
%e^{-(i_k-j_k)(\xi_k - {\frac {(\ell-1) \eta} 2})} \non \\ 
% & & 
= \widehat{ N}_{i_k, \, j_k}^{(\ell)} 
\,   
e^{-(i- j) \Lambda_k \, \delta(w, p)} 
\,  
P_{\ell(k-1)+1}^{(\ell)} \, 
\sum_{(\epsilon_{\beta}(j_k))_{\ell} } \chi_{1 2 \cdots L} \, 
e_{\ell(k-1)+1}^{\epsilon_1^{'}(i_k), \epsilon_1(j_k)} \cdots 
e_{\ell(k-1)+\ell}^{\epsilon_{\ell}^{'}(i_k), \epsilon_{\ell}(j_k)} 
\, 
\chi_{1 2 \cdots L}^{-1} \, . \label{eq:gen-E(w)-e}
\eea 
%Here we take the sum over all sets of $\varepsilon_{\beta}$   
% such that the number of integers $\beta$ satisfying $\varepsilon_{\beta}=1$ 
%for $1 \le \beta \le \ell$ is equal to $j_k$,  
%and 
Here we fix a sequence $(\varepsilon_{\alpha}^{'}(i))_{\ell}$. 
%such that the number of integers $\alpha$ with $\varepsilon_{\alpha}^{'}=1$ 
%for $1 \le \alpha \le \ell$ is equal to $i_k$. 
Furthermore, expression (\ref{eq:gen-E(w)-e}) 
does not depend on the order of $\varepsilon_{\alpha}^{'}(i)$ 
with respect to $\alpha$. 
\label{prop:gen-E(w)-e}
\end{proposition}

Let us consider a product of the general 
spin-$\ell/2$ elementary operators, 
 ${\hat E}_1^{i_1 , \, j_1 \, (\ell \, w)} 
\cdots {\hat E}_m^{i_m, \, j_m \, (\ell \, w)}$, for which we shall  
evaluate the zero-temperature spin-$s$ XXZ correlation functions.    
We introduce variables $\varepsilon_{\alpha}^{[k] \, '}(i_k)$ 
and $\varepsilon_{\beta}^{[k]}(j_k)$ which take only two values 0 or 1 
for $k=1, 2, \ldots, m$ and $\alpha, \beta=0, 1, \ldots, \ell$. 
\begin{corollary}
Let us take integers $i_k$ and $j_k$ satisfying 
$1 \le i_k, j_k \le \ell$ for $k=1, 2, \ldots, m$.  
The product of the general spin-$\ell/2$ elementary operators, 
$\widehat{ E}_1^{i_1, \, j_1 \, (\ell \, w)}   
\cdots \widehat{ E}_m^{i_m, \, j_m (\ell \, w)}$, is expressed in terms of 
a sum of products of the spin-1/2 elementary operators as  
\bea
& & 
\prod_{k=1}^{m} \widehat{ E}_k^{i_k, \, j_k \, (\ell \, w)} 
%\non \\ & & 
= \prod_{k=1}^{m} 
\left( \widehat{ N}_{i_k, \, j_k}^{(\ell)} 
\,  e^{- (i_k-j_k) \Lambda_k \delta(w, p)} \, 
%\sqrt{\left[ \begin{array}{c} 
%\ell \\ 
%i_k 
%\end{array}  \right]
%\left[ \begin{array}{c} 
%\ell \\ 
% j_k 
%\end{array}  
%\right]^{-1}} \, q^{(i_k - j_k)(\ell-1)/2} e^{-(i_k- j_k)\xi_k} 
\right) 
\nonumber \\ 
& & \qquad \times \quad 
P^{(\ell)}_{1 \cdots L}  \, 
\sum_{( \varepsilon_{\beta}^{[1]}(j_1) )_{\ell}} \cdots 
\sum_{( \varepsilon_{\beta}^{[m]}(j_m) )_{\ell} }
\chi_{1 2 \cdots (\ell m)} \,  \prod_{k=1}^{m} \left( 
e_{\ell(k-1)+1}^{\varepsilon_1^{[k] \, '}(i_k), \, 
\varepsilon_1^{[k]}(j_k) } \cdots 
e_{\ell(k-1)+\ell}^{\varepsilon_{\ell}^{[k]\, '}(i_k), \, 
\varepsilon_{\ell}^{[k]}(j_k) } \right) \, 
\chi^{-1}_{1 2 \cdots (\ell m)} \, . \non \\ \label{eq:productE}
\eea
%Here we take the sums over all sets of $\varepsilon_{\beta}^{[k]}$s 
%for $1 \le k \le m$ such that the number of integers $\beta$ satisfying 
%$\varepsilon_{\beta}^{[k]}=1$ and $1 \le \beta \le \ell$ is given by $j_k$ 
%for each $k$. 
Here we fix $(\varepsilon_{\alpha}^{[k] \, '}(i_k))_{\ell}$ 
for each integer $k$ of $1 \le k \le m$ . 
%such that the number of $\alpha$ 
%satisfying $\varepsilon_{\alpha}^{[k] \, '}=1$ 
%and $1 \le \alpha \le \ell$ is given by $i_k$ for each $k$.  
\label{cor:productE}
\end{corollary}

The expression (\ref{eq:productE}) is 
useful for deriving the multiple-integral 
representations of correlation functions for the integrable 
higher-spin XXZ spin chain, as we shall see in \S 6.

%%%%%%%%%%%%%%%%%%%%%%%%%%%%%%%%%%%%%%%%%%%%%%%%%%%%
%
%   Sec 3.5 
%
%%%%%%%%%%%%%%%%%%%%%%%%%%%%%%%%%%%%%%%%%%%%%%%%%%%%%%
\subsection{Quantum inverse-scattering problem for the spin-$\ell/2$ operators}

Let us recall the formula of the quantum inverse-scattering problem 
(QISP) for the spin-1/2 XXZ spin chain \cite{KMT1999,KMT2000}. 
\be 
x_n = \prod_{k=1}^{n-1}\left(A^{(1 \, w)}+D^{(1 \, w)} \right)(w_k) \cdot 
{\rm tr}_0 \left(x_0 T_{0, \, 1 2 \cdots L}^{(1 \, w)}(w_n) \right) 
\cdot \prod_{k=1}^{n} \left(A^{(1 \, w)}+D^{(1 \, w)} \right)^{-1}(w_k) \, . 
\label{eq:spin-1/2-QISP}
\ee
Here we assume that inhomogeneity parameters $w_j$ are given by 
generic values so that the transfer matrices 
$\left( A^{(1 \, w)}+D^{(1 \, w)} \right)(w_k)$ are regular 
for $k=1, 2, \ldots, n$. 
   
Making use of the QISP formula (\ref{eq:spin-1/2-QISP}) 
we have the following expressions for $b=1, 2, \ldots, N_s$: 
\bea
& & 
e_{\ell(b-1) +1}^{\varepsilon_1^{'}, \varepsilon_1} \cdots 
e_{\ell(b-1)+\ell}^{\varepsilon_{\ell}^{'}, \varepsilon_{\ell}}
 =   
\prod_{k=1}^{\ell (b-1)} 
\left(A^{(1 \, w)}(w_{k}) + D^{(1 \, w)}(w_{k}) \right) 
\non \\ 
& & \quad \times 
T^{(1 \, w)}_{\varepsilon_1, \, \varepsilon_1^{'}}
(w_{\ell(b-1)+1}) \cdots 
T^{(1 \, w)}_{\varepsilon_{\ell}, \, \varepsilon_{\ell}^{'}}
(w_{\ell(b-1)+\ell}) 
\prod_{k=1}^{\ell b } 
\left(A^{(1 \, w)}(w_{k}) 
+ D^{(1 \, w)}(w_{k}) \right)^{-1} \, . \label{eq:QISP-prod-e}
\eea
Here we have denoted by $T_{\alpha, \beta}(\lambda)$ the $(\alpha, \beta)$ 
element of the monodromy matrix $T(\lambda)$.

Applying (\ref{eq:QISP-prod-e}) to reduction formula (\ref{eq:E-w}) 
(or (\ref{eq:gen-E(w)-e})) 
we obtain the QISP formula for the spin-$\ell/2$ local 
operators. For an illustration, we show the case of $b=1$ as follows   
\bea 
& & \widehat{ E}^{i \, j \, (\ell \, w)}_1 = 
%\left(  \left[ \begin{array}{c} 
%\ell \\
%i 
%\end{array} 
%\right]_q 
%\left[ \begin{array}{c} 
%\ell \\
%j 
%\end{array} 
%\right]_q^{-1} \right)^{1/2} \, 
%e^{-(i-j)\xi_1} \, q^{(i-j)(\ell-1)/2} 
\widehat{ N}_{i, \, j}^{(\ell)} \, e^{- (i-j) \Lambda_1 \delta(w, p)} 
\, \times \non \\ 
&  \times &  P_{1 \cdots \ell}^{(\ell)} \,  \cdot \,  \chi_{1 2 \cdots \ell}  
\,  
\sum_{(\varepsilon_{\beta}(j))_{\ell}} 
T^{(1 \, w)}_{\varepsilon_1(j), \, \varepsilon_1^{'}(i)}
(w_1) \cdots 
T^{(1 \, w)}_{\varepsilon_{\ell}(j), \, \varepsilon_{\ell}^{'}(i)}
(w_{\ell}) 
\prod_{k=1}^{\ell} 
\left(A^{(1 \, w)}(w_{k}) 
+ D^{(1 \, w)}(w_{k}) 
\right)^{-1} 
\chi_{1 2 \cdots \ell}^{-1} 
\, . \non \\ 
\eea
%Here the sum over $\{ \varepsilon_{\beta} \}$ is taken over all sets 
%of $\varepsilon_{\beta}$ such that the number of $\varepsilon_{\beta}=1$ 
%for $1 \le \beta \le \ell$ is given by $j$. 
Here, we fix a sequence $(\varepsilon_{\alpha}^{'}(i))_{\ell}$.  
%is given by such a set that the number of $\varepsilon_{\alpha}^{'}=1$ 
%for $1 \le \alpha \le \ell$ is given by $i$.  

%%%%%%%%%%%%%%%%%%%%%%%%%%%%%%%%%%%%%%%%%%%%%%%%%%%%
%
%   Sec 3.6 
%
%%%%%%%%%%%%%%%%%%%%%%%%%%%%%%%%%%%%%%%%%%%%%%%%%%%%%%
\subsection{Non-regularity of the transfer matrix at special points}

Let us consider the sector of $M$ down-spins 
on the spin-1/2 chain with $L$ sites.  
\begin{proposition} 
In the sector of $M$ down-spins with $1 \le M \le L-1$, 
the spin-1/2 transfer matrix $A^{(\ell \, w; \ 0)}(\lambda) + 
D^{(\ell \, w; \ 0)}(\lambda)$ is non-regular 
at $\lambda=w_{\ell (k-1) + 1}^{(\ell)}+ n \pi \sqrt{-1}$ 
for $k=1, 2, \ldots, N_s$  and $n \in {\bf Z}$.  
Here, $w_{\ell (k-1) + 1}^{(\ell)}=\xi_k$ is the first rapidity 
of the $k$th complete $\ell$-string.  
\end{proposition} 
\begin{proof} 
Calculating the matrix elements of the transfer matrix  
$A^{(\ell \, w; \ 0)}(\lambda) + D^{(\ell \, w; \ 0)}(\lambda)$ 
in the sector of $M$ down-spins 
we show that there exists a pair of column vectors that are parallel 
to each other if $\lambda=\xi_k$.  
\end{proof} 

Thus, the inverse matrix of the spin-1/2 transfer matrix 
$A^{(\ell \, w; \ 0)}(\lambda) + D^{(\ell \, w; \ 0)}(\lambda)$ 
does not exist at the special points. 
We remark that in the sector of $M=0$ (and $M=L$), 
it is regular at $\lambda=\xi_k + n \pi \sqrt{-1}$ 
for $k=1, 2, \ldots, N_s$ and $n \in {\bf Z}$.

For an illustration, let us consider the case of $L=2$ 
with $\ell=1$ and $N_s=1$. 
The operators $A$ and $D$ are explicitly given by  
\be  
A_{12}^{(2 \, +; \, 0)}(\lambda)= 
\left(
\begin{array}{cccc} 
1 & 0 & 0 & 0 \\ 
0 & b_{02} & c_{01}^{+} c_{02}^{-}  & 0 \\ 
0 & 0 & b_{01} & 0 \\ 
0 & 0 & 0 & b_{01} b_{02}  
\end{array} 
\right)_{[1, 2]}\, , \quad 
D_{12}^{(2 \, +; \, 0)}(\lambda)= 
\left(
\begin{array}{cccc} 
b_{01} b_{02} & 0 & 0 & 0 \\ 
0 & b_{01} & 0  & 0 \\ 
0 & c_{01}^{-} c_{02}^{+} & b_{02} & 0 \\ 
0 & 0 & 0 & 1  
\end{array} 
\right)_{[1, 2]}\, . 
\ee
Here we have introduced  $b_{0j}$ and $c_{0j}^{\pm}$ for $j=1, 2$ by  
 $b_{0j}=b(\lambda-w_j^{(2)})$ and $c_{0j}^{\pm} = 
\exp(\pm (\lambda - w_j^{(2)})) c(\lambda - w_j^{(2)})$ for $j=1, 2$, 
respectively. Putting $\lambda=w_1^{(2)}= \xi_1$ we have  
\be 
 A_{12}^{(2 \, +; \, 0)}(\xi_1) + D_{12}^{(2 \, +; \, 0)}(\xi_1) = 
\left(
\begin{array}{cccc} 
1 & 0 & 0 & 0 \\ 
0 & {\frac 1 {[2]_q}} & {\frac {q^{-1}} {[2]_q}}   & 0 \\ 
0 & {\frac {q} {[2]_q}}  & {\frac 1 {[2]_q}}  & 0 \\ 
0 & 0 & 0 & 1 
\end{array} 
\right)_{[1, 2]} \, . 
\ee
We thus show that the transfer matrix  is non-regular 
at $\lambda=w_1^{(2)}=\xi_1$: 
\be  
\det \left(
A_{12}^{(2 \, +; \, 0)}(\xi_1) + D_{12}^{(2 \, +; \, 0)}(\xi_1) \right) = 0.  
\ee
In the sector of $M=1$ the determinant is given by 
\be 
\left. \det \left(
A_{12}^{(2 \, +; \, 0)}(\lambda) + D_{12}^{(2 \, +; \, 0)}(\lambda) \right) 
\right|_{M=1} 
= {\frac {4 \sinh(\lambda- \xi_1)} 
{\sinh(\lambda- \xi_1 + 2 \eta)}} 
\ee 

For an illustration, we shall show in Appendix C that there exists 
a pair of column vectors that are parallel to each other 
if we set $\lambda=\xi_k$, in the sector of $M=1$, 
 for the case of $L=3$ with $w_1=w_1^{(2)}$, $w_1=w_2^{(2)}$ and $w_3=\xi_2$. 

Consequently, the QISP formula does not hold in the straightforward form 
for the operator-valued matrix elements of the monodromy matrix 
$T^{(\ell \, w; \, 0)}(\lambda)$ for $w=\pm, p$ at 
$\lambda=w_{\ell(k-1) +1}^{(\ell)}$ for $k=1, 2, \ldots, N_s$. 
Here we recall that the monodromy matrix 
$T^{(\ell \, w; \, 0)}(\lambda)$ is given by the spin-1/2 monodromy 
matrix $T^{(\ell \, w; \, \epsilon)}(\lambda)$ 
by putting $\epsilon=0$.

%\newpage 
%%%%%%%%%%%%%%%%%%%%%%%%%%%%%%%%%%%%%%%%%%%%%%%%%%%%%%%%%%%%%%
%
%
%    Sec 4
%
%
\setcounter{equation}{0} 
\renewcommand{\theequation}{4.\arabic{equation}}
\section{Reduction of the matrix elements of spin-$\ell/2$ operators}

%%%%%%%%%%%%%%%%%%%%%%%%%%%%%%%%%%%%%%%%%%%%%%%%%%%%%%%%%%%%%%
%
%  Sec 4.1
%
%
\subsection{Definition of the spin-$\ell/2$ off-shell matrix elements 
and the spin-$\ell/2$ form factors}

Let $|0 \rangle$ be the vacuum vector of the spin-1/2 chain of $L$ sites; 
i.e., $|0 \rangle = | \uparrow \rangle_1 \otimes \cdots 
\otimes | \uparrow \rangle_L$. 
Here we recall that the symbol $\{ \lambda_{\alpha} \}_M$ 
denotes a set of $M$ parameters 
$\lambda_{\alpha}$ for $\alpha=1 , 2, \ldots, M$.

For given sets of parameters $\{ \mu_{\alpha} \}_N$ 
and $\{ \lambda_{\beta} \}_M$ we define the off-shell Bethe 
 covectors and vectors 
$\langle \{ \mu_{\alpha} \}_N^{(\ell \, w)} |$ and 
$| \{ \lambda_{\beta} \}_M^{(\ell \, w)} \rangle$, respectively,  
for $w=\pm, p$ as follows: 
\be 
\langle \{ \mu_{\alpha} \}_N^{(\ell \, w)} | = 
\langle 0 | \prod_{\alpha=1}^{N} C^{(\ell \, w )}(\mu_{\alpha}) \, , 
\quad  
| \{ \lambda_{\beta} \}_M^{(\ell \, w)} \rangle =
\prod_{\beta=1}^{M} B^{(\ell \, w)}(\lambda_{\beta}) | 0 \rangle \, .  
\ee
Here, parameters $\{ \mu_{\alpha} \}_N$ 
and $\{ \lambda_{\beta} \}_M$ do not necessarily satisfy the 
Bethe-ansatz equations. 
We define the spin-$\ell/2$ off-shell matrix elements   
of $E_k^{i_k, \, j_k \, (\ell \, w)}$ for $w=\pm, p$ by   
\be 
M_k^{i_k, \, j_k \, (\ell \, w)} 
(\{ \mu_{\alpha} \}_N, \{ \lambda_{\beta} \}_M)  = 
\langle \{ \mu_{\alpha} \}_N^{(\ell \, w)} |
\, 
E_k^{i_k, \, j_k \, (\ell \, w)} \, 
| \{ \lambda_{\beta} \}_M^{(\ell \, w)} \rangle \, . 
\ee
We also define the spin-$\ell/2$ off-shell matrix elements  
of the general elementary operators 
$\widehat{ E}_k^{i_k, \, j_1 \, (\ell \, w)}$ for $w=\pm, p$ by   
\be 
\widehat{ M}_k^{i_k, \, j_k \, (\ell \, w)} 
(\{ \mu_{\alpha} \}_N, \{ \lambda_{\beta} \}_M)    
= 
\langle \{ \mu_{\alpha} \}_N^{(\ell \, w)} |
\, 
\widehat{ E}_k^{i_k, \, j_k \, (\ell \, w)} \, 
| \{ \lambda_{\beta} \}_M^{(\ell \, w)} \rangle \, . 
\ee

Let us assume that  $\{ \mu_{\alpha} \}_N$ and $\{ \lambda_{\beta} \}_M$ satisfy the Bethe ansatz equations. 
We call $\langle \{ \mu_{\alpha} \}_N^{(\ell \, w)} |$ and 
$| \{ \lambda_{\beta} \}_M^{(\ell \, w)} \rangle$ the on-shell Bethe covectors and vectors, respectively. 
We define the spin-$\ell/2$ form factors   
of $E_k^{i_k, \, j_k \, (\ell \, w)}$ for $w=\pm, p$ by   
\be 
F_k^{i_k, \, j_k \, (\ell \, w)} 
(\{ \mu_{\alpha} \}_N, \{ \lambda_{\beta} \}_M)    
= 
\langle \{ \mu_{\alpha} \}_N^{(\ell \, w)} |
\, 
E_k^{i_k, \, j_k \, (\ell \, w)} \, 
| \{ \lambda_{\beta} \}_M^{(\ell \, w)} \rangle \, . 
\ee
We also define the spin-$\ell/2$ form factors   
of the general elementary operators 
$\widehat{ E}_k^{i_k, \, j_1 \, (\ell \, w)}$ for $w=\pm, p$ by   
\be 
\widehat{ F}_k^{i_k, \, j_k \, (\ell \, w)} 
(\{ \mu_{\alpha} \}_N, \{ \lambda_{\beta} \}_M)    
= 
\langle \{ \mu_{\alpha} \}_N^{(\ell \, w)} |
\, 
\widehat{ E}_k^{i_k, \, j_k \, (\ell \, w)} \, 
| \{ \lambda_{\beta} \}_M^{(\ell \, w)} \rangle \, . 
\ee

We have defined the form factors of a local operator 
by the matrix elements of the operator between all pairs of 
the Bethe eigenvectors, in the paper. However, it is often the case that 
only the matrix elements between the ground state and excited states 
are called form factors.

%%%%%%%%%%%%%%%%%%%%%%%%%%%%%%%%%%
%
% Sec 4.2
%
\subsection{Commutation relation with the projection operator}

\begin{lemma} 
If spectral parameter $\lambda$ is distinct from  
discrete values such as 
%$w_{\ell(k-1)+1}^{(\ell)}+ n \pi \sqrt{-1} $
%for $k=1, 2, \ldots, N_s$ and $n \in {\bf Z}$ and 
$w_j^{(\ell)} - \eta + n \pi \sqrt{-1}$ 
for $j=1, 2, \ldots, L$ and $n \in {\bf Z}$,  
the projection operator $P^{(\ell)}_{12 \cdots L}$ 
commutes with the matrix elements of the monodromy matrix 
$T^{(\ell \, +; \, 0)}_{0, 1 2 \cdots L}(\lambda)= 
T^{(1 \, +)}_{0, 1 2 \cdots L}(\lambda; \{ w_j^{(\ell)} \}_L)$ 
as follows. 
\bea  
 P_{12 \cdots L}^{(\ell)}  
T^{(1 \, +)}_{0, 12 \cdots L}(\lambda; 
\{ w_j^{(\ell)} \}_L) \, P_{12 \cdots L}^{(\ell)}  
& = & 
P_{12 \cdots L}^{(\ell)}
 \, 
T^{(1 \, +)}_{0, 1 2 \cdots L}(\lambda; \{ w_j^{(\ell)} \}_L) 
\, . \label{eq:commute}  
\eea
\label{lem:commute}  
\end{lemma}

Let us assume that 
all the  parameters in  $\{ \mu_{\alpha}\}_N$ and $\{ \lambda_{\beta} \}_M$ 
are different from 
the discrete values given by $w_j^{(\ell)}- \eta + n \pi \sqrt{-1}$ 
for $j=1, 2, \ldots, L$ and $n \in {\bf Z}$. 
Here we recall that they correspond to 
$N_s$ pieces of complete $\ell$-strings minus $\eta$ 
modulo $\pi \sqrt{-1}$, and 
the transfer matrix is singular at these points. 
%Then, the monodromy matrix $T^{(1, \, \ell ; \, 0)}(\lambda)$ 
%are regular at $\lambda=\mu_{\alpha}$ and also 
%at $\lambda=\lambda_{\beta}$.     
Applying Lemma \ref{lem:commute} we have 
\be | \{ \lambda_{\beta} \}_M^{(\ell \, +)} \rangle = 
\prod_{\beta=1}^{M} \left(  P^{(\ell)}_{1 2 \ldots L} 
B^{(\ell \, + \, ; \, 0)}(\lambda_{\beta})  
P^{(\ell)}_{1 2 \ldots L} \right) \, | 0 \rangle 
= 
P^{(\ell)}_{1 2 \ldots L} 
\prod_{\beta=1}^{M} B^{(\ell \, + ; \, 0)}(\lambda_{\beta}) 
\, | 0 \rangle \label{eq:|lambda>}    
\ee
and 
\be  
\langle \{ \mu_{\alpha} \}_N^{(\ell \, +)} | = 
\langle 0 | \prod_{\alpha=1}^{N} 
\left( P^{(\ell)}_{1 2 \ldots L}  C^{(\ell \, +; \, 0)}(\mu_{\alpha}) 
 P^{(\ell)}_{1 2 \ldots L} \right)
%
%= \langle 0 | \prod_{\alpha=1}^{N} 
% C^{(\ell \, +; \, 0)}(\mu_{\alpha}) | \, P^{(\ell)}_{1 \ldots L}
%
= \langle 0 | \prod_{\alpha=1}^{N} 
 C^{(\ell \, +; \, 0)}(\mu_{\alpha}) \, . \label{eq:<mu|}  
\ee
Here we remark that for the off-shell Bethe covectors 
we can absorb the projection operator acting to the left as follows.   
\be 
\langle 0 | \prod_{\alpha=1}^{N} 
 C^{(\ell \, +; \, 0)}(\mu_{\alpha}) \, \cdot \, P_{12 \cdots L}^{(\ell)} 
= \langle 0 | \prod_{\alpha=1}^{N} 
 C^{(\ell \, +; \, 0)}(\mu_{\alpha}) \, .    
\ee
However, in eq. (\ref{eq:|lambda>}) we can not remove 
the projection operator acting to the right. 

It follows from (\ref{eq:|lambda>}) and (\ref{eq:<mu|}) that we can evaluate 
the spin-$\ell/2$ off-shell matrix elements  
by calculating  the spin-1/2 ones. 
For instance, applying (\ref{eq:|lambda>}) and (\ref{eq:<mu|})
we reduce every spin-$\ell/2$ 
off-shell matrix element into a  
spin-$1/2$ off-shell matrix element as follows. 
\bea 
\widehat{ M}^{i, \, j \, (\ell \, +)}_{k}
(\{ \mu_{\alpha} \}_N, \{ \lambda_{\beta} \}_M)   
& = & \langle 0 | \prod_{\alpha=1}^{N} 
 C^{(\ell \, +; 0)}(\mu_{\alpha}) \, 
\widehat{ E}^{i, \, j \, (\ell \, +)}_k \, 
P^{(\ell)}_{1 2 \ldots L} 
\prod_{\beta=1}^{M} B^{(\ell \, + ; \, 0)}(\lambda_{\beta}) | 0 \rangle 
\non \\ 
& = & \langle 0 | \prod_{\alpha=1}^{N} 
 C^{(\ell \, +; 0)}(\mu_{\alpha}) \, 
\widehat{ E}^{i, \, j \, (\ell \, +)}_k \,  
\prod_{\beta=1}^{M} B^{(\ell \, + ; \, 0)}(\lambda_{\beta}) | 0 \rangle \, . 
\label{eq:gen-F+}
\eea
Here we have made use of Lemma \ref{lem:PE=EP=E} in order to 
delete the projection operator.

We remark that in Refs. \cite{DM1,DM2,DM3} 
there was a nontrivial assumption that the projection operator should commute 
with the operator-valued matrix elements of the spin-1/2 monodromy matrix 
$T^{(\ell, \, w; \, 0)}(\lambda)$ 
at an {\it arbitrary} value of the spectral parameter $\lambda$.  
In fact, the spin-1/2 monodromy matrix becomes singular 
if $\lambda$ is equal to some discrete values such as 
$w_j-\eta$. At $\lambda=w_j-\eta$ 
the commutation relation of the monodromy matrix  with the projection 
operator becomes non-trivial.    
If we multiply it with normalization factor $\sinh(\lambda-w_j+\eta)$ and 
define the normalized monodromy matrix, then its commutation relation 
with the projection operator becomes valid at $\lambda=w_j-\eta$.

%%%%%%%%%%%%%%%%%%%%%%%%%%%%%%%%%%%%%%%%%%%%%%%%%%%%%%%%%%%%%%
%
%  Sec 4.3
%
\subsection{Reduction of the spin-$\ell/2$ off-shell matrix elements 
into the spin-1/2 ones} 

For homogeneous gradings with $w=\pm$ and principal grading with $w=p$, 
we define $\sigma(w)$ by 
\be 
\sigma(w) = \left\{ 
\begin{array}{ccc} 
\pm 1 & \mbox{for} & w=\pm \, , \\ 
0 & \mbox{for} & w=p \, . 
\end{array} 
\right.  
\ee
We denote by ${\cal S}_n$ the symmetric group of $n$ elements. 
\begin{proposition} 
Let $i_1$ and $j_1$ be integers satisfying $1 \le i_1, j_1 \le \ell$. 
For arbitrary parameters $\{\mu_{\alpha} \}_N$ and   
$\{\lambda_{\beta} \}_M$ with $i_1-j_1=N-M$ 
we have 
\bea 
& & \widehat{ M}^{i_1, \, j_1 \, (\ell \, w)}_1(\{ \mu_k \}_N, 
\{ \lambda_{\beta} \}_M)  =  
\langle 0 | \prod_{k=1}^{N} C^{(\ell \, w)}(\mu_k) \cdot 
\widehat{ E}^{i_1, \, j_1 \, (\ell \, w)}_1 \cdot \prod_{\beta=1}^{M} 
B^{(\ell \, w)}(\lambda_{\beta}) | 0 \rangle 
\non \\ 
& = & 
\widehat{ N}_{i_1, \, j_1}^{(\ell)} \, 
e^{\sigma(w) (\sum_k \mu_k-\sum_{\gamma} \lambda_{\gamma})} 
\, 
\sum_{(\varepsilon_{\beta}(j_1) )_{\ell} }  
\langle 0 | \prod_{\alpha=1}^{N} C^{(\ell \, p; \, 0)}(\mu_a) 
\, \cdot \, 
e_1^{\varepsilon_1^{'}(i_1), \varepsilon_1(j_1)} \cdots 
e_{\ell}^{\varepsilon_{\ell}^{'}(i_1), \varepsilon_{\ell}(j_1)} 
\, \cdot \, \prod_{\beta=1}^{M} 
B^{(\ell \, p; \, 0)}(\lambda_{\beta}) | 0 \rangle \, .  
\non \\ 
\label{eq:<gen-E(w)>-e}
\eea
Each summand is symmetric with respect to 
exchange of $\varepsilon_{\alpha}^{'}(i_1)$; i.e.,   
the following expression is independent of any 
permutation $\pi \in {\cal S}_{\ell}$:  
\be 
\langle 0 | \prod_{\alpha=1}^{N} C^{(\ell \, p; \, 0)}(\mu_{\alpha}) 
\, \cdot \, 
e_1^{\varepsilon_{\pi 1}^{'}(i_1), \, \varepsilon_1(j_1)} \cdots 
e_{\ell}^{\varepsilon_{\pi \ell}^{'}(i_1), \, \varepsilon_{\ell}(j_1)} 
\, \cdot \, \prod_{\beta=1}^{M} 
B^{(\ell \, p; \, 0)}(\lambda_{\beta}) | 0 \rangle \, . 
\label{eq:sym-epsilon_a'}
\ee
\label{prop:<gen-E(w)>-e} 
\end{proposition} 
\begin{proof}
In the case of homogeneous grading with $w=+$, we put 
(\ref{eq:gen-E(w)-e}) into (\ref{eq:gen-F+}), and we have 
(\ref{eq:<gen-E(w)>-e}) through the gauge transformation:
\bea 
\langle 0 | \prod_{\alpha=1}^{N}C^{(\ell \, + ; \, 0)}(\mu_{\alpha})  
& = & e^{\sum_{\alpha} \mu_{\alpha}} \, 
\langle 0 | \prod_{\alpha=1}^{N}C^{(\ell \, p ; \, 0)}(\mu_{\alpha}) 
 \, \cdot \, 
\chi_{1 2 \cdots L}^{-1} \, , \non \\ 
\prod_{\beta=1}^{M} B^{(\ell \, + ; \, 0)}(\lambda_{\beta}) | 0 \rangle  
& = &  \chi_{1 2 \cdots L} \, \cdot \, 
\prod_{\beta=1}^{M} B^{(\ell \, p ; \, 0)}(\lambda_{\beta}) | 0 \rangle 
\, e^{-\sum_{\beta} \lambda_{\beta}}  \, .  
\eea
In the case of principal grading with $w=p$ we shall show 
explicitly in Appendix D the following:  
\bea
 \langle 0 | \prod_{\alpha=1}^{N} C^{(\ell \, p)}(\mu_{\alpha}) 
&  = & \langle 0 | \prod_{k=1}^{N}  C^{(\ell \, p; \, 0)}(\mu_{k}) 
\, \cdot \,  \chi_{ 1 \cdots L}^{-1}  
P^{(\ell)}_{1 \cdots L} \chi_{ 1 \cdots N_s}^{(\ell)}  \, ,  
\non \\ 
\prod_{\alpha=1}^{N} B^{(\ell \, p)}(\lambda_{\alpha}) | 0 \rangle 
& = & 
\left( \chi_{ 1 \cdots N_s}^{(\ell)} \right)^{-1} 
P^{(\ell)}_{1 \cdots L} \,  \chi_{ 1 \cdots L} \, \cdot \, 
 \prod_{\alpha=1}^{M}  
B^{(\ell \, p; \, 0)}(\lambda_{\alpha}) | 0 \rangle \,
\eea
We thus evaluate the form factor as follows.  
\bea 
& & F_1^{i_1, \, j_1 \, (\ell \, p)}
\{ \mu_{\alpha} \}_N, \{ \lambda_{\beta} \}_M)  
= \langle 0 | \prod_{\alpha=1}^{N} 
C^{(\ell \, p)}(\mu_{\alpha}) \, \cdot 
\, E_1^{i_1, \, j_1 \, (\ell \, p)} \, \cdot \, 
\prod_{\beta=1}^{M} B^{(\ell \, p)}(\lambda_{\beta}) | 0 \rangle     
\non \\ 
& = & 
\langle 0 | \prod_{\alpha=1}^{N} 
C^{(\ell \, p; \, 0)}(\mu_{\alpha}) \, \cdot 
\, 
\chi_{ 1 \cdots L}^{-1} \,  
P^{(\ell)}_{1 \cdots L} 
\chi_{ 1 \cdots N_s}^{(\ell)} \, \cdot \,  
E_1^{i_1, \, j_1 \, (\ell \, p)} \, \cdot \, 
\left( \chi_{ 1 \cdots N_s}^{(\ell)} \right)^{-1} 
P^{(\ell)}_{1 \cdots L} \, \chi_{ 1 \cdots L}  
\, \cdot \, 
\prod_{\beta=1}^{M} B^{(\ell \, p; \, 0)}(\lambda_{\beta}) | 0 \rangle   
\non \\ 
& = & 
\langle 0 | \prod_{\alpha=1}^{N} 
C^{(\ell \, p; \, 0)}(\mu_{\alpha}) 
\, \cdot \, 
\chi_{ 1 \cdots L}^{-1}  E_1^{i_1, \, j_1 \, (\ell \, p)} 
\chi_{ 1 \cdots L}  
\, \cdot \, 
\prod_{\beta=1}^{M} B^{(\ell \, p; \, 0)}(\lambda_{\beta}) | 0 \rangle  \, 
e^{(i_1-j_1)\left( \xi_1- (\ell-1) {\eta}/2 \right) }  \, . 
\eea
Here we remark that 
\be 
\chi_{ 1 \cdots N_s}^{(\ell)} \, \cdot \,  
E_1^{i_1, \, j_1 \, (\ell \, p)} \, \cdot \, 
\left( \chi_{ 1 \cdots N_s}^{(\ell)} \right)^{-1} 
= E_1^{i_1, \, j_1 \, (\ell \, p)} \, e^{(i_1-j_1)(\xi_1-(\ell-1)\eta/2)} \, . 
\ee
Applying Proposition \ref{prop:E-w} we obtain eq.  
(\ref{eq:<gen-E(w)>-e}) for the case of $w=p$.  
\end{proof}

\begin{corollary} 
Let us take integers $i_k$ and $j_k$ satisfying 
$1 \le i_k, j_k \le \ell$ for $k=1, 2, \ldots, m$.
For arbitrary sets of parameters $\{\mu_{\alpha} \}_N$ and   
$\{\lambda_{\beta} \}_M$ with $\sum_{k=1}^{m} i_k- 
\sum_{k=1}^{m} j_k = N-M$, we have the matrix element for 
the $m$th product of the spin-$\ell/2$ elementary operators with 
entries $(i_k, j_k)$ 
($1 \le k \le m$) associated with principal grading as follows: 
\bea 
& & \langle 0 | \prod_{\alpha=1}^{N} C^{(\ell \, w)}(\mu_a) \cdot 
\prod_{k=1}^{m} \widehat{ E}^{i_k, \, j_k \, (\ell \, w)}_k 
\cdot \prod_{\beta=1}^{M} 
B^{(\ell \, w)}(\lambda_{\beta}) | 0 \rangle  
= 
\prod_{k=1}^{m} 
%\sqrt{\left[ \begin{array}{c} \ell \\ i_k \end{array} 
%\right] \left[ \begin{array}{c} \ell \\ j_k \end{array} \right]^{-1}} 
%
N^{(\ell)}_{i_k, \, j_k} \, \cdot \, 
e^{\sigma(w) (\sum_k \mu_k-\sum_{\gamma} \lambda_{\gamma})}
\, 
\non \\  
&  \times & 
\sum_{ (\varepsilon_{\beta}^{[1]}(j_1) )_{\ell} } \cdots 
\sum_{ ( \varepsilon_{\beta}^{[m]}(j_m) )_{\ell} } 
\langle 0 | \prod_{\alpha=1}^{N} C^{(\ell \, p; \, 0)}(\mu_a) 
\, \cdot \, \prod_{k=1}^{m} \left( 
e_{\ell(k-1)+1}^{\varepsilon_1^{[k] \, '}(i_k), \varepsilon_1^{[k]}(j_k) } 
\cdots 
e_{\ell(k-1)+\ell}^{\varepsilon_{\ell}^{[k] \, '}(i_k), 
\varepsilon_{\ell}^{[k]}(j_k)} \right)
\, \cdot \, \prod_{\beta=1}^{M} 
B^{(\ell \, p; \, 0)}(\lambda_{\beta}) | 0 \rangle \, .   \non \\ 
\label{eq:E(p)k-e[k]}
\eea
%
%Here we take the sum over all $\varepsilon_{\beta}^{[k]}$ 
%such that the number of $\beta$ with $\varepsilon_{\beta}=1^{[k]}$ 
%for $1 \le \beta \le \ell$ is given by $j_k$ 
%for each integer $k$ with $1 \le k \le m$,  
%while for each $k$ we  take any set of $\varepsilon_{\alpha}^{[k] \, '}$s   
%such that the number of $\alpha$ with $\varepsilon_{\alpha}^{[k] \, '}=1$ 
%for $1 \le \alpha \le \ell$ is given by $i_k$. 
The summand is symmetric with respect to 
exchange of $\varepsilon_{\alpha}^{[k] \, '}(i_k)$s for each $k$ 
of $1 \le k \le m$; i.e.,   
the following expression is independent of any 
permutation $\pi^{[k]} \in {\cal S}_{\ell}$ for $k=1, 2, \ldots, m$:  
\be 
\langle 0 | \prod_{\alpha=1}^{N} C^{(\ell \, p; \, 0)}(\mu_{\alpha}) 
\, \cdot \, 
\prod_{k=1}^{m} 
e_{\ell(k-1)+1}^{\varepsilon_{\pi^{[k]} 1}^{[k] \, '}(i_k), \, 
\varepsilon_1^{[k]}(j_k) } \cdots 
e_{\ell(k-1)+\ell}^{\varepsilon_{\pi^{[k]} \ell}^{[k] \, '}(i_k), \, 
\varepsilon_{\ell}^{[k]}(j_k) } 
\, \cdot \, \prod_{\beta=1}^{M} 
B^{(\ell \, p; \, 0)}(\lambda_{\beta}) | 0 \rangle  \, . 
\label{eq:ff-e[k]}
\ee
\label{cor:E(p)k-e[k]}
\end{corollary}

%%%%%%%%%%%%%%%%%%%%%%%%%%%%%
%
%  Sec 4.4 
%
\subsection{Consequence of the continuity assumption of the Bethe roots}

We now consider the Bethe-ansatz equations 
for the integrable spin-$\ell/2$ XXZ spin chain with inhomogeneity parameters 
$\xi_b$ for $b=1, 2, \ldots, N_s$:
\be 
{\frac {a^{(\ell)}(\lambda_{\alpha})} {d^{(\ell)}
(\lambda_{\alpha}; \{ \xi_k \}_{N_s}) } }  
= \prod_{\beta=1; \beta \ne \alpha}^{M} 
{\frac {\sinh(\lambda_{\alpha} - \lambda_{\beta} + \eta)}  
{\sinh(\lambda_{\alpha} - \lambda_{\beta} - \eta)} } \quad 
(\alpha=1, 2, \ldots, M).  \label{eq:BAE-spin-ell/2}
\ee
Here we recall $L=\ell N_s$. For intgers $\ell$ we have set 
$a^{(\ell)}(\lambda_{\alpha})=1$ and defined 
$d^{(\ell)}(\mu; \{\xi_k \}_{N_s})$ by 
\be 
 d^{(\ell)}(\mu; \{\xi_k \}_{N_s})= 
\prod_{k=1}^{N_s} 
{\frac {\sinh(\mu - \xi_k)} {\sinh(\mu - \xi_k + \ell \eta)}}.  
\ee%

Let $\{ \lambda_{\gamma} \}_M$ be a solution of the Bethe-ansatz equations 
of the spin-$\ell/2$ chain with inhomogeneity parameters $\{ \xi_b \}_{N_s}$.  
Suppose that  $\{ \lambda_{\beta}(\epsilon) \}_M$ denotes a solution 
of the spin-1/2 Bethe-ansatz equations 
with inhomogeneity parameters $w_j$ being given by the $N_s$ sets 
of the almost complete $\ell$-strings, 
$w_j^{(\ell; \, \epsilon)}$ for $j=1, 2, \ldots, L$. 
They satisfy the Bethe-ansatz equations for the spin-$1/2$ XXZ spin chain:
\be 
{\frac {a(\lambda_{\alpha}(\epsilon))} 
{d(\lambda_{\alpha}(\epsilon) ; \{ w_j^{(\ell; \, \epsilon)} \}_{L}) } }  
= \prod_{\beta=1; \beta \ne \alpha}^{M} 
{\frac {\sinh(\lambda_{\alpha}(\epsilon) 
- \lambda_{\beta}(\epsilon) + \eta)}  
{\sinh(\lambda_{\alpha}(\epsilon) 
- \lambda_{\beta}(\epsilon) - \eta)} } \quad 
(\alpha=1, 2, \ldots, M).  \label{eq:BAE-spin-1/2}
\ee
Here we have defined $d(\mu; \{w_j \}_{L})$ by 
\be 
 d(\mu; \{w_j \}_{L})= \prod_{j=1}^{L} b(\mu-w_j) =  
\prod_{j=1}^{L} 
{\frac {\sinh(\mu - w_j)} {\sinh(\mu - w_j + \eta)}}.  
\ee
Then, the 
Bethe-ansatz equations (\ref{eq:BAE-spin-1/2}) 
for the spin-1/2 XXZ chain with $w_j=w_j^{(\ell; \, \epsilon)}$ 
($1 \le j \le L$) become those of the spin-$\ell/2$ XXZ chain 
by sending $\epsilon$ to zero. Here we remark the followng limit: 
\be 
\lim_{\epsilon \rightarrow 0}  d(\mu; \, \{ w_j^{(\ell; \, \epsilon)} \})
=d^{(\ell)}(\mu; \, \{ \xi_k \}_{N_s}).  \label{eq:limit-d}
\ee

Let us now assume that the Bethe roots 
$\{ \lambda_{\beta}(\epsilon) \}_M $ approach the Bethe roots 
$\{ \lambda_{\beta} \}_M$ continuously 
in the limit of sending $\epsilon$ to $0$.  
It follows that each entry of 
the Bethe-ansatz eigenstate of the Bethe roots 
$\{ \lambda_{\beta}(\epsilon) \}_M$ 
is continuous with respect to $\epsilon$.  
For a set of arbitrary parameters $\{ \mu_k \}_N$ we therefore have  
\bea 
& & \langle 0 | \prod_{\alpha=1}^{N} C^{(\ell \, p; \, 0)}(\mu_a) 
\, \cdot \, 
e_1^{\varepsilon_{1}^{'}, \, \varepsilon_1} \cdots 
e_{\ell}^{\varepsilon_{\ell}^{'}, \, \varepsilon_{\ell}} 
\, \cdot \, \prod_{\beta=1}^{M} 
B^{(\ell \, p; \, 0)}(\lambda_{\beta}) | 0 \rangle 
\non \\ 
& & = \lim_{\epsilon \rightarrow 0} 
 \langle 0 | \prod_{\alpha=1}^{N} C^{(\ell \, p; \, \epsilon)}(\mu_a) 
\, \cdot \, 
e_1^{\varepsilon_{1}^{'}, \, \varepsilon_1} \cdots 
e_{\ell}^{\varepsilon_{\ell}^{'}, \, \varepsilon_{\ell}} 
\, \cdot \, \prod_{\beta=1}^{M} 
B^{(\ell \, p; \, \epsilon)}(\lambda_{\beta}(\epsilon)) 
| 0 \rangle \, . 
\eea

%%%%%%%%%%%%%%%%%%%%%%%%%%%%%%%%%%%%%%%%%%%%%%%%%%%%%%%%%%%%%%
%
% Sec 4.*
%
%\subsection{QISP with almost complete $\ell$-strings}

%Let us recall that $w_j^{(\ell; \, \epsilon)}$ forming 
%$N_s$ pieces of ``almost complete $\ell$-strings''. 
The inhomogeneity parameters $w_j=w_j^{(\ell; \, \epsilon)}$ 
for $j=1, 2, \ldots, L$  are generic 
since the small number $\epsilon$ takes generic values and  
parameters $r_{b}^{\beta}$ are also generic.    
Putting $w_j=w_j^{(\ell; \, \epsilon)}$ in (\ref{eq:spin-1/2-QISP}) 
we have the QISP formula with suffix $(1 \, w)$ replaced by 
$(\ell \, w; \, \epsilon)$ for local operator $x_n$. 
We thus have the following expressions for $b=1, 2, \ldots, N_s$: 
\bea
& & 
e_{\ell(b-1) +1}^{\varepsilon_1^{'}, \varepsilon_1} \cdots 
e_{\ell(b-1)+\ell}^{\varepsilon_{\ell}^{'}, \varepsilon_{\ell}}
 =   
\prod_{k=1}^{\ell (b-1)} 
\left(A^{(\ell \, w; \, \epsilon)}(w_{k}^{(\ell; \, \epsilon)}) 
+ D^{(\ell \, w; \, \epsilon)}(w_{k}^{(\ell; \, \epsilon)}) 
\right) 
\non \\ 
& & \quad \times 
T^{(\ell \, w; \, \epsilon)}_{\varepsilon_1, \varepsilon_1^{'}}
(w_{\ell(b-1)+1}^{(\ell; \, \epsilon)}) \cdots 
T^{(\ell \, w; \, \epsilon)}_{\varepsilon_{\ell}, \varepsilon_{\ell}^{'}}
(w_{\ell(b-1)+\ell}^{(\ell; \, \epsilon)}) 
\prod_{k=1}^{\ell b } 
\left(A^{(\ell \, w; \, \epsilon)}(w_{k}^{(\ell; \, \epsilon)}) 
+ D^{(\ell \, w; \, \epsilon)}(w_{k}^{(\ell; \, \epsilon)}) 
\right)^{-1} \, . \label{eq:almost-QISP-prod-e}
\eea
For instance in the case of $b=1$, 
applying formula (\ref{eq:almost-QISP-prod-e}) of $w=p$ we have 
\bea 
& & \langle 0 | \prod_{\alpha=1}^{N} C^{(\ell \, p; \, \epsilon)}(\mu_a) 
\, \cdot \, 
e_1^{\varepsilon_{1}^{'}, \, \varepsilon_1} \cdots 
e_{\ell}^{\varepsilon_{\ell}^{'}, \, \varepsilon_{\ell}} 
\, \cdot \, \prod_{\beta=1}^{M} 
B^{(\ell \, p; \, \epsilon)}(\lambda_{\beta}(\epsilon)) 
| 0 \rangle \non \\
& = & \phi_{\ell}(\{ \lambda_{\beta} \}; \{ w_j^{(\ell)} \} ) \,   
\langle 0 | \prod_{\alpha=1}^{N} C^{(\ell \, p; \, \epsilon)}(\mu_a) 
\, \cdot \, 
T^{(\ell \, p ; \, \epsilon)}_{\varepsilon_{1}, \, \varepsilon_1^{'}}
(w_1^{(\ell; \, \epsilon)}) \cdots 
T^{(\ell \, p ; \, \epsilon)}_{\varepsilon_{\ell}, \, \varepsilon_{\ell}^{'}}
(w_{\ell}^{(\ell; \epsilon)}) 
\, \cdot \, \prod_{\beta=1}^{M} 
B^{(\ell \, p; \, \epsilon)}(\lambda_{\beta}(\epsilon)) 
| 0 \rangle \, . \non \\ 
\eea

\begin{proposition} 
Let $\{ \mu_k \}_N$ be a set of arbitrary parameters and   
$\{ \lambda_{\alpha} \}_M$ a solution 
of the spin-$\ell/2$ Bethe-ansatz equations. 
We denote by $\{ \lambda_{\alpha}(\epsilon) \}_M$  a solution of 
the Bethe-ansatz equations for the spin-1/2 XXZ chain whose inhomogeneity 
parameters $w_j$ are given by the $N_s$ pieces 
of the almost complete $\ell$-strings: 
$w_j=w_j^{(\ell; \, \epsilon)}$ for $1 \le j \le L$.  
We assume that the set $\{ \lambda_{\alpha}(\epsilon) \}_M$ approaches 
$\{ \lambda_{\alpha} \}_M$ continuously when we send $\epsilon$ to zero. 
For the Bethe states 
$\langle \{ \mu_k \}_N |$ and $| \{ \lambda_{\alpha} \}_M \rangle$, 
which are off-shell and on-shell, respectively, 
we evaluate the matrix elements of a given product of 
elementary operators 
$e_1^{\varepsilon_1^{'}, \varepsilon_1} \cdots 
e_{\ell}^{\varepsilon_{\ell}^{'}, \varepsilon_{\ell}}$ as follows. 
\bea  
& &  \langle 0 | \prod_{\alpha=1}^{N} C^{(\ell \, p; \, 0)}(\mu_a) \, 
e_1^{\varepsilon_1^{'}, \varepsilon_1} \cdots 
e_{\ell}^{\varepsilon_{\ell}^{'}, \varepsilon_{\ell}} 
\, \prod_{\beta=1}^{M} 
B^{(\ell \, p; \, 0)}(\lambda_{\beta}) | 0 \rangle  \nonumber \\ 
& & = \phi_{\ell}(\{ \lambda_{\beta} \}; \{w_j^{(\ell)} \}) \,   
\lim_{\epsilon \rightarrow 0} 
\langle 0 | \prod_{\alpha=1}^{N} C^{(\ell \, p; \, \epsilon)}(\mu_a) 
\, 
T^{(\ell \, p; \, \epsilon)}_{\varepsilon_1, \, \varepsilon_1^{'}}(w_1^{(\ell; \, \epsilon)}) \cdots 
T^{(\ell  \, p; \, \epsilon)}_{\varepsilon_{\ell}, \, \varepsilon_{\ell}^{'}}
(w_{\ell}^{(\ell; \, \epsilon)}) 
\, \prod_{\beta=1}^{M} 
B^{(\ell \, p; \, \epsilon)}(\lambda_{\beta}(\epsilon)) | 0 \rangle \, , 
\non \\ \label{eq:CeeeB}
\eea
where $\phi_{m}(\{ \lambda_{\beta} \})$ 
has been defined by  
$\phi_{m} (\{ \lambda_{\beta} \}; \{ w_j \} ) 
= \prod_{j=1}^{m} \prod_{\alpha=1}^M 
b(\lambda_{\alpha}-w_j)$ with $b(u)=\sinh(u)/\sinh(u+\eta)$. 
\label{prop:FF-principal} 
\end{proposition}

%%%%%%%%%%%%%%%%%%%%%%%%%%%%%
%
%  Sec 4.5 
%
\subsection{Spin-$\ell/2$ form factors reduced into the spin-1/2 ones}

Combining Propositions \ref{prop:<gen-E(w)>-e},  and \ref{prop:FF-principal} 
we have the following: 
\begin{proposition} 
Let  $i_1$ and $j_1$ be integers satisfying 
$1 \le i_1, j_1 \le \ell$. We set $i_1-j_1=N-M$.  
Let $\{ \mu_k \}_N$ be a set of arbitrary $N$ parameters. 
For a set of Bethe roots $\{ \lambda_{\beta}(\epsilon) \}_M$ 
which approaches $\{ \lambda_{\beta} \}_M$ continuously at $\epsilon=0$ 
we have the following:  
\bea 
& & \langle 0 | \prod_{\alpha=1}^{N} C^{(\ell \, w)}(\mu_a) \cdot 
\widehat{ E}^{i_1, \, j_1 \, (\ell \, w)}_1 \cdot \prod_{\beta=1}^{M} 
B^{(\ell \, w)}(\lambda_{\beta}) | 0 \rangle \non \\ 
%
%\sqrt{\left[ \begin{array}{c} \ell \\ i_1 \end{array} \right] 
%\left[ \begin{array}{c} \ell \\ j_1 \end{array} \right]^{-1}}  
%
& & = \widehat{ N}_{i_1, \, j_1}^{(\ell)} \,  
e^{\sigma(w) (\sum_k \mu_k-\sum_{\gamma} \lambda_{\gamma})} \, 
\phi_{\ell}(\{ \lambda_{\beta} \}; \{w_j^{(\ell)} \}) \,   
\non \\ 
&  & \times \, \sum_{ (\varepsilon_{\beta}(j_1) )_{\ell} } \,   
\lim_{\epsilon \rightarrow 0} 
\langle 0 | \prod_{\alpha=1}^{N} C^{(\ell \, p; \, \epsilon)}(\mu_a) 
\, 
T^{(\ell \, p; \, \epsilon)}_{\varepsilon_1(j_1), \, \varepsilon_1^{'}(i_1)}
(w_1^{(\ell; \, \epsilon)}) \cdots 
T^{(\ell  \, p; \, \epsilon)}_{\varepsilon_{\ell}(j_1), \, 
\varepsilon_{\ell}^{'}(j_{\ell}) }(w_{\ell}^{(\ell; \, \epsilon)}) 
\, \prod_{\beta=1}^{M} 
B^{(\ell \, p; \, \epsilon)}(\lambda_{\beta}(\epsilon)) | 0 \rangle \, . 
\label{eq:reduction-formula-E1} \non \\
\eea
\end{proposition}

Let us recall a product of the general 
spin-$\ell/2$ elementary operators, 
 ${\hat E}_1^{i_1 , \, j_1 \, (\ell \, w)} 
\cdots {\hat E}_m^{i_m, \, j_m \, (\ell \, w)}$, 
which we have introduced in Corollary \ref{cor:productE}.    
We also recall variables $\varepsilon_{\alpha}^{[k] \, '}(i_k)$ 
and $\varepsilon_{\beta}^{[k]}(j_k)$ which take only two values 0 or 1 
for $k=1, 2, \ldots, m$ and $\alpha, \beta=0, 1, \ldots, \ell$. 
We have the following:
\begin{corollary}
Let us take integers $i_k$ and $j_k$ satisfying 
$1 \le i_k, j_k \le \ell$ for $k=1, 2, \ldots, m$. 
 We set $\sum_k i_k-\sum_{k} j_k=N-M$.  
Let $\{ \mu_k \}_N$ be a set of arbitrary $N$ parameters. 
If the set of the Bethe roots $\{ \lambda_{\beta}(\epsilon) \}_M$ 
approaches the set of the Bethe roots 
$\{ \lambda_{\beta} \}_M$ continuously at $\epsilon=0$, 
we have the following:  
\bea 
& & \langle 0 | \prod_{\alpha=1}^{N} C^{(\ell \, w)}(\mu_a) \cdot 
\prod_k \widehat{ E}_k^{i_k, \, j_k \, (\ell \, w)} 
\cdot \prod_{\beta=1}^{M} 
B^{(\ell \, w)}(\lambda_{\beta}) | 0 \rangle \non \\ 
%
%\sqrt{\left[ \begin{array}{c} \ell \\ i_1 \end{array} \right] 
%\left[ \begin{array}{c} \ell \\ j_1 \end{array} \right]^{-1}}  
%
& & = \left( \prod_{k=1}^{m} \widehat{ N}_{i_k, \, j_k}^{(\ell_k)} 
\right) \, \cdot \, 
e^{\sigma(w) (\sum_{k=1}^{N} \mu_k-\sum_{\gamma=1}^{M} \lambda_{\gamma})} \, 
\phi_{\ell}(\{ \lambda_{\beta} \}; \{w_j^{(\ell)} \}) \,   
\non \\ 
&  & \times \, \sum_{ (\varepsilon_{\beta}^{[1]}(j_1) )_{\ell} } 
\cdots \sum_{ (\varepsilon_{\beta}^{[m]}(j_m) )_{\ell} } \,   
\lim_{\epsilon \rightarrow 0} 
\langle 0 | \prod_{\alpha=1}^{N} C^{(\ell \, p; \, \epsilon)}(\mu_a) \, 
\prod_{k=1}^{m} 
\left( T^{(\ell \, p; \, \epsilon)}_{\varepsilon_1^{[1]}(j_k), 
\, \varepsilon_1^{[k]'}(i_k)}(w_1^{(\ell; \, \epsilon)}) \cdots 
T^{(\ell  \, p; \, \epsilon)}_{\varepsilon_{\ell}(j_k), \, 
\varepsilon_{\ell}^{'}(j_{k}) }(w_{\ell}^{(\ell; \, \epsilon)})
 \right) \non \\ 
& & \quad \times 
\, \prod_{\beta=1}^{M} 
B^{(\ell \, p; \, \epsilon)}(\lambda_{\beta}(\epsilon)) | 0 \rangle \, . 
\label{eq:reduction-formula-product} 
\eea
Here we have chosen sequences $\varepsilon_{\alpha}^{[k] \, '}(j_k)$
for each integer $k$ of $1 \le k \le m$. 
\end{corollary}

%\newpage
%%%%%%%%%%%%%%%%%%%%%%%%%%%%%%%%%%%
%
%         Sec 5
%
\setcounter{equation}{0} 
\renewcommand{\theequation}{5.\arabic{equation}}
\section{Spin-$\ell/2$ form factors via the spin-1/2 scalar products}

%%%%%%%%%%%%%%%%%%%%%%%%%%%%%%%%%%%
%
%         Sec 5.1 
%
\subsection{Fundamental commutation relations}

For given sequences  
$(\varepsilon_{\alpha}^{'} )_m$ and $( \varepsilon_{\beta} )_m$  
we consider sets $\mbox{\boldmath$\alpha$}^{\pm}$ defined by 
eqs. (\ref{eq:def-sets-aa'}). We also denote 
the sets $\mbox{\boldmath$\alpha$}^{\pm}$ by 
${\bm \alpha}^{-}(\{ \varepsilon_{\alpha}^{'}\})$ and 
${\bm \alpha}^{+}(\{ \varepsilon_{\alpha}^{'}\})$, respectively, 
in order to show their dependence on the sequences
$(\varepsilon_{\alpha}^{'} )_m$ and $( \varepsilon_{\beta} )_m$ explicitly.  
We take a set of distinct integers $a_j$ for $j \in {\bm \alpha}^{-}$ 
and $a_k^{'}$ for $k \in {\bm \alpha}^{+}$
such that they satisfy $1 \le a_j \le N$ for $j \in {\bm \alpha}^{-}$ and 
$1 \le a_k^{'} \le N+k$ for $k \in {\bm \alpha}^{+}$.  
For the given set of $a_j, a_j^{'}$, 
we introduce $\mbox{\boldmath$A$}_j$ and $\mbox{\boldmath$A$}_j^{'}$ by 
\begin{eqnarray} 
\mbox{\boldmath$A$}_j & = & 
\{ b; 1 \le b \le N + m, \, b \ne a_k, a_k^{'} \, \, 
\mbox{for} \, k < j \} \, , 
\nonumber \\ 
\mbox{\boldmath$A$}_j^{'} 
& = & \{ b; 1 \le b \le N+m, \, b \ne a_k \, \, \mbox{for} \, k \le j, 
b \ne a_k^{'} \, \mbox{for} \, k < j \}.          
\end{eqnarray}

Setting rapidities $\mu_{N+j}$ by 
\begin{equation} 
\mu_{N+j}= w_j \, , \quad \mbox{for}
 \, \,  j=1, 2, \ldots, m \, ,   
\end{equation} 
we can show the fundamental commutation relations as follows \cite{KMT2000}.   
\begin{eqnarray}  
& & \langle 0 | \left( \prod_{\alpha=1}^{N} C^{(1 \, p)}
(\mu_{\alpha}) \right) \, 
T^{(1 \, p)}_{\varepsilon_1, \varepsilon_1^{'}}(\mu_{N+1}) \cdots  
T^{(1 \, p)}_{\varepsilon_{2sm}, \varepsilon_{m}^{'}}(\mu_{N+m})  
\nonumber \\ 
& = & 
\left( 
\prod_{j \in \mbox{\boldmath$\alpha$}^{-}(\{ \varepsilon_{\alpha}^{'} \})} 
\sum_{a_j=1}^{N}    
\prod_{j \in \mbox{\boldmath$\alpha$}^{+}(\{ \varepsilon_{\beta} \})} 
\sum_{a_j^{'}=1}^{N+j} 
\right) 
\, \, 
G_{\{ a_j, \,  a_j^{'} \}}
^{ ( \varepsilon_{\alpha}^{'} )_m, \, ( \varepsilon_{\beta} )_m }
( ( \mu_k )_{N+m} )      
\langle 0 | 
\prod_{k \in \mbox{\boldmath$A$}_{m+1}(\{a_j, \, a_j^{'} \}) } 
C^{(1 \, p)}(\mu_{k}) \, ,  \nonumber   
\end{eqnarray}
where coefficients $G_{\{ a_j, \, a_j^{'} \}}
^{ ( \varepsilon_{\alpha}^{'} )_m, \, ( \varepsilon_{\beta} )_m }
((\mu_{\alpha})_{N+m})$ are given by   
\begin{eqnarray} 
& & G_{\{ a_j, \, a_j^{'} \}}
^{ ( \varepsilon_{\alpha}^{'} )_m, \, ( \varepsilon_{\beta} )_m }
( (\mu_k)_{N+m} )  
= \prod_{j \in \mbox{\boldmath$\alpha$}^{+}(\{ \varepsilon_{\beta} \}) }
 \left( 
{\frac {\prod_{b=1; b \in \mbox{\boldmath$A$}^{'}_j }^{N+j-1} 
\sinh(\mu_{b}-\mu_{a_j^{'}} + \eta) }  
 {\prod_{b=1, b \in \mbox{\boldmath$A$}_{j+1} }^{N+j} 
 \sinh(\mu_{b} - \mu_{a_j^{'}})}}  
\right) 
\nonumber \\ 
& & \qquad \qquad  
\times  
 \prod_{j \in \mbox{\boldmath$\alpha$}^{-}(\{ \varepsilon_{\alpha}^{'}\}) }
 \left( d(\mu_{a_j}; \{ w_k \}_L) 
{\frac {\prod_{b=1; b \in \mbox{\boldmath$A$}_j}^{N+j-1} 
\sinh(\mu_{a_j}-\mu_b + \eta) }  
 {\prod_{b=1, b \in \mbox{\boldmath$A$}^{'}_j }^{N+j} 
 \sinh(\mu_{a_j} - \mu_b)}}  
     \right) \, . \label{eq:FCR}
\eea
Here we recall 
\be 
  d(\mu; \{ w_k \}_L) =  \prod_{k=1}^{L} b(\mu - w_k) \, . 
\ee

We consider  
the sums over integers $a_j$ and $a_k^{'}$ such that they satisfy  
$1 \le a_j \le N$ and $1 \le a_k^{'} \le N+j$ 
for $j \in {\bm \alpha}^{-}$ and $k \in {\bm \alpha}^{+}$, 
respectively. Hereafter, we express  
the products of the sums over $a_j$ and $a_k^{'}$ by 
the symbol $\sum_{\{ a_j, a_j^{'} \}}$, as follows.  
\begin{equation} 
\sum_{\{ a_j, \, a_k^{'} \}} =
 \prod_{j \in \mbox{\boldmath$\alpha$}^{-}} 
\left( \sum_{a_j=1}^{N} \right)   
\prod_{j \in \mbox{\boldmath$\alpha$}^{+}} 
\left( \sum_{a_j^{'}=1}^{N+j} \right) \, . 
\end{equation}

%%%%%%%%%%%%%%%%%%%%%%%%%%%%%%%%%%%%%%%%%%%%%%%%%%%%%%
%
%  Sec 5.2
%
\subsection{Form factors as a sum of the spin-1/2 scalar products }

We shall evaluate the spin-$\ell/2$ form factor 
${\widehat F}^{i, \, j \, (\ell \, w)}_k
(\{ \mu_{\alpha} \}_N, \, \{ \lambda_{\beta} \}_M)$. 

We first define the scalar product in the spin-$1/2$ case for two sets of $M$ parameters $\{ \mu_k \}_M$ and $\{ \lambda_{\gamma} \}_M$ by 
\be 
S^{(1)}_M(\{ \mu_1, \ldots, \mu_M\}, \{ \lambda_1, \ldots, \lambda_M \}; \{ w_j \}_L) = 
\langle 0 | \prod_{k=1}^{M} C^{(1 \, p)}(\mu_k) 
\prod_{\gamma=1}^{M} B^{(1 \, p)}(\lambda_{\gamma}) | 0 \rangle \, .  
\ee
Here $\{ \mu_k \}_M$ and $\{ \lambda_{\gamma} \}_M$ are not necessarily 
solutions of the Bethe ansatz equations. 
We define the scalar product for the spin-$\ell/2$ operators  
$B^{(\ell \, p)}(\mu_k)$ and $C^{(\ell \, p)}(\lambda_{k})$ 
for $k=1, 2, \ldots, M$, by    
\be 
S_M^{(\ell)}(\{\mu_{\alpha} \}, \{\lambda_{\beta} \}; \{ \xi_k \}_{N_s} ) 
= \langle 0 | \prod_{\alpha=1}^{M} C^{(\ell \, p)}(\mu_{\alpha})  
\prod_{\beta=1}^{M} B^{(\ell \, p)}(\lambda_{\beta}) | 0 \rangle \, . 
\ee
Here we also recall that 
$\{ \mu_k \}_M$ and $\{ \lambda_{\gamma} \}_M$ are not necessarily Bethe roots. 

Let us first review Slavnov's formula of scalar products 
in the spin-1/2 case \cite{Slavnov}: 
if $\lambda_1, \lambda_2, \ldots, \lambda_M$ 
satisfy the spin-1/2 Bethe-ansatz equations 
for the spin-$1/2$ XXZ spin chain with inhomogeneity parameters $w_j$, 
the scalar product is expressed in terms of the determinant:  
\be 
S_M^{(1)}(\{\mu_{\alpha} \}, \{\lambda_{\beta} \}; \{ w_j \}_{L} ) 
= {\frac 
{\det \widehat{H}(\{ \lambda_{\alpha} \}_M, \{ \mu_k \}_M; \{w_j \}_{L})} 
 {\prod_{1 \le j < k \le M} \sinh(\mu_j-\mu_k)  
\prod_{1 \le \alpha < \beta \le M} \sinh(\lambda_{\beta}-\lambda_{\alpha})}} 
\, . 
\ee
Here, the matrix elements of $\widehat{ H }$ 
with entry $(a, b)$ for $a, b= 1, 2, \ldots, M$ are given by 
\bea 
\widehat{ H}_{a, \, b}(\{ \lambda \}_M, \mu_b; \, \{ w_j \}_{L}) 
& = & {\frac {\sinh \eta} {\sinh(\lambda_a- \mu_b)}}    
\Big( a(\mu_b) \prod_{k=1; k \ne a}^{M} \sinh(\lambda_k - \mu_b + \eta)  
\non \\ 
& & \qquad - d(\mu_b; \{ w_j \}_{L}) 
\prod_{k=1; k \ne a}^{M} \sinh(\lambda_k - \mu_b - \eta)  
\Big) \, . 
\eea
Here we recall that 
$d(\mu; \{ w_j \}_{L}) 
= \prod_{j=1}^{L} \sinh(\mu- w_j)/\sinh(\mu- w_j + \eta)$.

%%%%%%%%%%%%%%%%%%%%%%%%%%%
We remark that 
it is sometimes useful to make use of the following relations: 
\begin{lemma}
For two sets of arbitrary parameters $\{ \mu_k \}_{N+m}$ 
and $\{ \lambda_{\gamma} \}_M$ we have 
\begin{eqnarray}  
& & \langle 0 | \prod_{\alpha=1}^{N} 
C^{(1 \, p)}(\mu_{\alpha}) \, \cdot \, 
T_{\varepsilon_1, \varepsilon_1^{'}}^{(1 \, p)}(\mu_{N+1}) \cdots  
T_{\varepsilon_{m}, \varepsilon_{m}^{'}}^{(1 \, p)}(\mu_{N+m}) 
\, \cdot \, \prod_{\gamma=1}^{M} 
B^{(1 \, p)}(\lambda_{\gamma}) | 0 \rangle  
\non \\ 
& & = 
\langle 0 | \prod_{\beta=1}^{M} 
C^{(1 \, p)}(\lambda_{\beta}) \, \cdot \, 
T_{\varepsilon_{m}^{'}, \varepsilon_{m}}^{(1 \, p)}(\mu_{N+m}) 
 \cdots 
T_{\varepsilon_1^{'}, \varepsilon_1}^{(1 \, p)}(\mu_{N+1})
\, \cdot \, \prod_{\gamma=1}^{M} 
B^{(1 \, p)}(\mu_{\gamma}) | 0 \rangle  \, . 
\eea
\end{lemma}

We now consider the matrix element of an $m$th product 
of the spin-1/2 operators with respect to given bra and ket vectors,  
 $\langle \{ \mu_{\alpha} \}_N^{(\ell \, p; \, 0)} |$ and     
$| \{ \lambda_{\alpha} \}_N^{(\ell \, p; \, 0)} \rangle$, respectively.  
We define $P$ by $P=N-M$, where $P$ can be negative. 
Setting $\mu_{N+j}=w_j$ for $j=1, 2, \ldots, m$, we have 
\begin{eqnarray}  
& & \langle 0 | \left( \prod_{\alpha=1}^{N} 
C^{(1 \, p)}(\mu_{\alpha}) \right) \, 
T_{\varepsilon_1, \varepsilon_1^{'}}^{(1 \, p)}(\mu_{N+1}) \cdots  
T_{\varepsilon_{m}, \varepsilon_{m}^{'}}^{(1 \, p)}(\mu_{N+m}) 
\, \cdot \, \prod_{\gamma=1}^{M}
B^{(1 \, p)}(\lambda_{\gamma}) | 0 \rangle  
\nonumber \\ 
& = & 
 \sum_{\{ a_j, \, a_k^{'} \} }  
G_{\{ a_j, \,  a_k^{'} \}}^{(\varepsilon_{\alpha}^{'})_m, \, 
(\varepsilon_{\beta})_m}( (\mu_k)_{N+m} ) \quad      
\langle 0 | 
\prod_{k \in \mbox{\boldmath$A$}_{m+1}(\{a_j, a_j^{'} \}) } 
C(\mu_{k}) \, \cdot \,  \prod_{\gamma=1}^{M}
B(\lambda_{\gamma}) | 0 \rangle  
  \nonumber  \\  
& = & 
 \sum_{\{ a_j, \, a_k^{'} \} } 
G_{\{ a_j, \,  a_k^{'} \}}
^{(\varepsilon_{\alpha}^{'})_m, \, (\varepsilon_{\beta})_m}
( (\mu_k)_{N+m} )    \non \\ 
& & \times \, 
S_M^{(1)}(\{ \mu_1, \ldots, \mu_{N}, w_1, \ldots, 
w_{m} \} \setminus \{\mu_{a_j}, \mu_{a_k^{'}} \}_{m+P}, 
\{ \lambda_{\gamma} \}_M; \{ w_j \}_L)  \, .  
\label{eq:scalar-product}  
\end{eqnarray} 
Here we remark that the number of elements in the set $\{ a_j, \, a_k^{'} \}$ 
is given by $m+P$.

Let us now consider the case of $m=\ell$ for 
the form factor of the spin-$\ell/2$ operators. 
For a given sequence of parameters $\mu_k$ for $1 \le k \le N$ 
we extend it into a sequence of length $N+\ell$ by setting 
$\mu_k = w_{k-N}^{(\ell)}$ for $k=N+1, \ldots, N+\ell$.  
We define another sequence $\mu_k(\epsilon)$ for $k=1, 2, \ldots, N+\ell$ by 
\be  
\mu_k(\epsilon) = 
\left\{ 
\begin{array}{cc} 
\mu_k & \mbox{for} \, \, 1 \le k \le N , \\
w_{k-N}^{(\ell; \, \epsilon)} & \mbox{for} \, \, N < k \le N+\ell . 
\end{array}
\right.
\ee
Substituting (\ref{eq:scalar-product}) into (\ref{eq:CeeeB}) 
we have 
\bea  
& &  \langle 0 | \prod_{\alpha=1}^{N} C^{(\ell \, p; \, 0)}(\mu_a) \, 
e_1^{\varepsilon_1^{'}, \varepsilon_1} \cdots 
e_{\ell}^{\varepsilon_{\ell}^{'}, \varepsilon_{\ell}} 
\, \prod_{\beta=1}^{M} 
B^{(\ell \, p; \, 0)}(\lambda_{\beta}) | 0 \rangle  \nonumber \\ 
%& & = \phi_{\ell}(\{ \lambda_{\beta} \}; \{w_j^{(\ell)} \}) \,   
%\lim_{\epsilon \rightarrow 0} 
%\langle 0 | \prod_{\alpha=1}^{N} C^{(\ell \, p; \, \epsilon)}(\mu_a) \, 
%T^{(\ell \, p; \, \epsilon)}_{\varepsilon_1, \, \varepsilon_1^{'}}
%(w_1^{(\ell; \, \epsilon)}) \cdots 
%T^{(\ell  \, p; \, \epsilon)}_{\varepsilon_{\ell}, \, \varepsilon_{\ell}^{'}} 
%(w_{\ell}^{(\ell; \, \epsilon)}) \, \prod_{\beta=1}^{M} 
%B^{(\ell \, p; \, \epsilon)}(\lambda_{\beta}(\epsilon)) | 0 \rangle \non \\ 
%
& & = \phi_{\ell}(\{ \lambda_{\beta} \}; \{w_j^{(\ell)} \}) \,    
 \sum_{\{ a_j, \, a_k^{'} \} } 
G_{\{ a_j, \,  a_k^{'} \}}^{(\varepsilon_{\alpha}^{'})_{\ell}, \, 
(\varepsilon_{\beta})_{\ell}}((\mu_k)_{N+\ell})   \nonumber \\ 
& & \times \, \lim_{\epsilon \rightarrow 0}     
S_M^{(1)}(\{ \mu_1(\epsilon), \ldots, \mu_{N}(\epsilon), 
w_1^{(\ell; \, \epsilon)}, \ldots, 
w_{\ell}^{(\ell; \, \epsilon)} \} \setminus \{\mu_{a_j}(\epsilon), 
\mu_{a_k^{'}}(\epsilon) \}_{\ell+P}, 
\{ \lambda_{\gamma}(\epsilon) \}_M; 
\{ w_j^{(\ell; \, \epsilon)} \}_L)  .  \non \\ 
\label{eq:limit-scalar-product}
\eea

%\bea & & \langle 0 | \prod_{\alpha=1}^{N} C^{(\ell \, w)}(\mu_a) \cdot 
%\widehat{ E}^{i_1, \, j_1 \, (\ell \, w)}_1 \cdot \prod_{\beta=1}^{M} 
%B^{(\ell \, w)}(\lambda_{\beta}) | 0 \rangle \non \\ 
%
%& & = \widehat{ N}_{i_1, \, j_1}^{(\ell)} \,  
%e^{\sigma(w) (\sum_k \mu_k-\sum_{\gamma} \lambda_{\gamma})} \, 
% \phi_{\ell}(\{ \lambda_{\beta} \}; \{w_j^{(\ell)} \}) \,   
% \sum_{ (\varepsilon_{\beta}(j_1) )_{\ell} } \, \sum_{\{ a_j, \, a_k^{'} \} } 
%G_{\{ a_j, \,  a_k^{'} \}}^{(\varepsilon_{\alpha}^{'}(i_1))_{\ell}, \, 
%(\varepsilon_{\beta}(j_1))_{\ell}}((\mu_k)_{N+\ell})   \nonumber \\ 
% & & \times \, \lim_{\epsilon \rightarrow 0}     
% S_M^{(1)}(\{ \mu_1(\epsilon), \ldots, \mu_{N}(\epsilon), 
% w_1^{(\ell; \, \epsilon)}, \ldots, 
% w_{\ell}^{(\ell; \, \epsilon)} \} \setminus \{\mu_{a_j}(\epsilon), 
% \mu_{a_k^{'}}(\epsilon) \}_{\ell+P}, \{ \lambda_{\gamma}(\epsilon) \}_M; 
%\{ w_j^{(\ell; \, \epsilon)} \}_L) . \non \\ \label{eq:red-FF} \eea
%

We shall explicitly express the limiting procedure 
in the last line of (\ref{eq:limit-scalar-product}). 
Here we recall the Bethe-ansatz equations (BAE) for the spin-$\ell/2$ case (\ref{eq:BAE-spin-ell/2}) and the limiting procedure (\ref{eq:limit-d}), where the 
spin-$\ell/2$ BAE is derived from the spin-1/2 BAE with 
almost complete strings $w_j^{(\ell; \, \epsilon)}$ by sending $\epsilon$ to 0.
Let us now consider the case where some of $\mu_k$ 
are given by inhomogeneity parameters $w_j$. 
We first define the following function: 
\be 
  d^{(\ell)'}(\mu; \{ \xi_k \}_{N_s}) = 
\left\{ \begin{array}{cc}
 \prod_{k=1}^{N_s} {\frac 
{\sinh(\mu - \xi_k)} {\sinh(\mu - \xi_k+ \ell \eta)}}
 & \mbox{for} \, \, \mu \ne w_j^{(\ell)} , \\  
  0 & \mbox{for} \,\,  \mu=w_j^{(\ell)} . 
\end{array}
\right. \, . 
\ee
We next introduce the matrix ${\widehat H}^{(\ell)'}$. 
We define the matrix elements of the matrix ${\widehat H}^{(\ell)'}$ 
with entry $(a, b)$ for $a, b= 1, 2, \ldots, M$, by 
\bea 
{\widehat H}_{a, \, b}^{(\ell)'}(\{ \lambda \}_M, \mu_b; \, \{ \xi_k \}_{N_s})
& = & {\frac {\sinh \eta} {\sinh(\lambda_a- \mu_b)}}    
\Big( a(\mu_b) \prod_{k=1; k \ne a}^{M} \sinh(\lambda_k - \mu_b + \eta)  
\non \\ 
& & \qquad - d^{(\ell)'}(\mu_b; \{ \xi_k \}_{N_s}) 
\prod_{k=1; k \ne a}^{M} \sinh(\lambda_k - \mu_b - \eta)  
\Big) \, . 
\eea
If $\{\lambda_{\beta} \}_M$ satisfy the Bethe ansatz equations of 
the spin-$\ell/2$ XXZ spin chain with inhomogeneity parameters $\xi_k$, 
we define the expression  
$S^{(\ell)'}_M(\{ \mu_k\},  \{ \lambda \}; \{ \xi_k \}_{N_s}) $ by 
\be 
S^{(\ell)'}_M( \{ \mu_k\}_M, \, \{ \lambda \}_M; \{ \xi_k \}_{N_s}) 
= {\frac 
{\det {\widehat H}^{(\ell)'}
(\{ \lambda_{\alpha} \}_M, \{ \mu_k \}_M; \{\xi_k \}_{N_s})} 
 {\prod_{1 \le j < k \le M} \sinh(\mu_j-\mu_k)  
\prod_{1 \le \alpha < \beta \le M} \sinh(\lambda_{\beta}-\lambda_{\alpha})}} 
\, .
\ee
In eq.(\ref{eq:limit-scalar-product}), sending $\epsilon$ to zero, 
for a given set of integers $\{ a_j, \, a_k^{'} \}$ we have 
\bea 
& & \lim_{\epsilon \rightarrow 0}  
S^{(1)}_M( \{ \mu_k(\epsilon) \}_{N+\ell} \setminus 
\{\mu_{a_j}(\epsilon), \mu_{a_k^{'}}(\epsilon) \}_{\ell+P}, 
\{ \lambda_{\beta}(\epsilon) \}_M ; \{ w_j^{(\ell; \, \epsilon)} \}_{L}) 
\non \\ 
& & = S^{(\ell)'}_M(\{ \mu_k \}_{N+\ell} \setminus 
\{\mu_{a_j}, \mu_{a_k^{'}} \}_{\ell+P},
\{ \lambda_{\beta} \}_M ; \{ \xi_k \}_{N_s}) \, . 
\eea
We summarize the result as follows: 
\begin{lemma}
Let $\{ \lambda_{\gamma} \}_M$ be a solution 
of the Bethe ansatz equations for the spin-$\ell/2$ chain 
with inhomogeneity parameters $\xi_k$ ($1 \le k \le N_s$), 
and $\{ \mu_k \}_N$ a set of arbitrary parameters.  
We assume that $\{ \lambda_{\beta}(\epsilon) \}$ is a solution of 
the Bethe ansatz equations for the spin-1/2 chain with 
$w_j=w_j^{(\ell; \, \epsilon)}$ ($1 \le j \le L$) and 
it approaches $\{ \lambda_{\gamma} \}_M$ continuously at $\epsilon=0$.  
We express the matrix elements of a product 
of the spin-1/2 operators in the limit of sending 
$\epsilon$ to 0 in terms 
of the modified scalar product $S^{(\ell)'}$ as follows.      
\bea
& & 
\lim_{\epsilon \rightarrow 0} 
\langle 0 | \prod_{k=1}^{N} C^{(\ell \, p; \, \epsilon)}(\mu_k) 
\, 
T^{(\ell \, p; \, \epsilon)}_{\varepsilon_1, \, \varepsilon_1^{'}}(w_1^{(\ell; \, \epsilon)}) \cdots 
T^{(\ell  \, p; \, \epsilon)}_{\varepsilon_{\ell}, \, \varepsilon_{\ell}^{'}}
(w_{\ell}^{(\ell; \, \epsilon)}) 
\, \prod_{\beta=1}^{M} 
B^{(\ell \, p; \, \epsilon)}(\lambda_{\beta}(\epsilon)) | 0 \rangle 
\non \\ 
& = & 
%\phi_{\ell}(\{ \lambda_{\beta} \}; \{w_j^{(\ell)} \}) \,    
 \sum_{\{ a_j, \, a_k^{'} \} } 
G_{\{ a_j, \,  a_k^{'} \}}^{(\varepsilon_{\alpha}^{'})_{\ell}, \, 
(\varepsilon_{\beta})_{\ell}}((\mu_k)_{N+\ell})   \nonumber \\ 
& & \times \,    
S_M^{(\ell)'}(\{ \mu_1, \ldots, \mu_{N}, 
w_1^{(\ell)}, \ldots, w_{\ell}^{(\ell)} \} 
\setminus \{\mu_{a_j}, \mu_{a_k^{'}} \}_{\ell+P}, 
\{ \lambda_{\gamma} \}_M; 
\{ \xi_k \}_{N_s})  \, . 
\eea
Here, $P=N-M$ and we have set 
$\mu_{N+j}=w_j^{(\ell)}$ for $j=1, 2, \ldots, \ell$. 
\label{lem:mat-mS}
\end{lemma}

Through Lemma \ref{lem:mat-mS} we show the following: 
\begin{proposition}
Let $\{ \lambda_{\gamma} \}_M$ be a solution 
of the Bethe-ansatz equations for the spin-$\ell/2$ chain 
with inhomogeneity parameters $\xi_k$ ($1 \le k \le N_s$), 
and $\{ \mu_k \}_N$ a set of arbitrary parameters.  
We assume that $\{ \lambda_{\beta}(\epsilon) \}$ is a solution of 
the Bethe ansatz equations for the spin-1/2 chain with 
$w_j=w_j^{(\ell; \, \epsilon)}$ ($1 \le j \le L$).  
Then, 
every spin-$\ell/2$ form factor associated with the Bethe roots 
$\{ \lambda_{\gamma} \}_M$ can be expressed 
as a sum of the spin-1/2 scalar products. 
For instance, we evaluate the form factor of the spin-$\ell/2$ elementary 
operator as follows.   
\bea
& & \widehat{F}^{i_1, \, j_1 \, (\ell \, w)}_{k=1}
(\{ \mu_{\alpha} \}_N, \{ \lambda_{\beta} \}_M)   
= \langle 0 | \prod_{\alpha=1}^{N} C^{(\ell \, w)}(\mu_a) \cdot 
\widehat{ E}^{i_1, \, j_1 \, (\ell \, w)}_1 \cdot \prod_{\beta=1}^{M} 
B^{(\ell \, w)}(\lambda_{\beta}) | 0 \rangle \non \\ 
%
%& & = \widehat{ N}_{i_1, \, j_1}^{(\ell)} \,  
%e^{\sigma(w) (\sum_k \mu_k - \sum_{\gamma} \lambda_{\gamma})} \, 
%\phi_{\ell}(\{ \lambda_{\beta} \}; \{w_j^{(\ell)} \}) \, \non \\ 
%&  & \times \, \sum_{ \{ \varepsilon_{\beta} \} } \,   
%\lim_{\epsilon \rightarrow 0} 
%\langle 0 | \prod_{\alpha=1}^{N} C^{(\ell \, p; \, \epsilon)}(\mu_a) \, 
%T^{(\ell \, p; \, \epsilon)}_{\varepsilon_1, \, \varepsilon_1^{'}}
%(w_1^{(\ell; \, \epsilon)}) \cdots 
%T^{(\ell  \, p; \, \epsilon)}_{\varepsilon_{\ell}, \, \varepsilon_{\ell}^{'}}
%(w_{\ell}^{(\ell; \, \epsilon)}) \, 
%\prod_{\beta=1}^{M} B^{(\ell \, p; \, \epsilon)}
%(\lambda_{\beta}(\epsilon)) | 0 \rangle \non \\ 
%
& & = \widehat{ N}_{i_1, \, j_1}^{(\ell)} \,  
e^{\sigma(w) (\sum_k \mu_k-\sum_{\gamma} \lambda_{\gamma})} \, 
\phi_{\ell}(\{ \lambda_{\beta} \}; \{w_j^{(\ell)} \}) \,   
\sum_{ (\varepsilon_{\beta}(j_1) )_{\ell} } \,   
 \sum_{\{ a_j, \, a_k^{'} \} } 
G_{\{ a_j, \,  a_k^{'} \}}^{(\varepsilon_{\alpha}^{'}(i_1))_{\ell}, \, 
(\varepsilon_{\beta}(j_1))_{\ell}}((\mu_k)_{N+\ell})   \nonumber \\ 
& & \times \,      
S_M^{(\ell)'}(\{ \mu_1, \ldots, \mu_{N}, 
w_1^{(\ell)}, \ldots, 
w_{\ell}^{(\ell)} \} \setminus \{\mu_{a_j}, 
\mu_{a_k^{'}} \}_{\ell+P}, 
\{ \lambda_{\gamma} \}_M; 
\{ \xi_k \}_{N_s})  . 
 \label{eq:red-FF-explicit}
\eea
Here we have fixed a sequence $\varepsilon^{'}_{\alpha}(i_1)$. 
\end{proposition}
It is easy to show 
the formula corresponding to (\ref{eq:red-FF-explicit}) 
for a given product of the spin-$\ell/2$ elementary operators such as 
$\widehat{ E}^{i_1, \, j_1 \, (\ell \, w)}_1 
\widehat{ E}^{i_2, \, j_2 \, (\ell \, w)}_2 \cdots 
\widehat{ E}^{i_m, \, j_m \, (\ell \, w)}_m$. 

%%%%%%%%%%%%%%%%%%%%%%%%%%%%%%%%%%%
%
%         Sec 6
%
\setcounter{equation}{0} 
\renewcommand{\theequation}{6.\arabic{equation}}
\section{Spin-$s$ XXZ correlation functions in a massless region}

Applying the reduction formula we now derive the multiple-integral representations of the correlation functions of the integrable spin-$s$ XXZ spin chain in a region of the massless regime: $0 \le \zeta < \pi/2s$. Here we remark that integer $2s$ corresponds to integer $\ell$ of $V^{(\ell)}$. We show only the main results. In fact, we derive them by following mainly the procedures 
of Ref. \cite{DM2} except for the evaluation of the expectation values of products of the spin-$s$ operators. 

Let us review the main procedures for deriving the multiple-integral representation of the spin-$s$ XXZ correlation functions, briefly. First, we introduce the spin-$s$ elementary operators as the basic blocks for constructing the local operators of the integrable spin-$s$ XXZ spin chain. Secondary, we reduce them into a sum of products of the spin-1/2 elementary operators, which we express in temrms of the matrix elements of the spin-1/2 monodromy matrix through the spin-1/2 QISP formula. We then evaluate their scalar products with Slavnov's formula of the Bethe-ansatz scalar products, we we have shown in \S 5. Here, the expectation value of a physical quantity is expressed as a sum of the ratios of the Bethe-ansatz scalar products to the norm of the ground-state Bethe-ansatz eigenvector, and the ratios are expressed in terms of the determinants of some matrices.    Thirdly, by solving the integral equations for the matrices in the thermodynamic limit, we derive the multiple-integral representation of the correlation functions. Here, the integrals and their solutions for the spin-$s$ case 
are given in Ref. \cite{DM2}.

%%%%%%%%%%%%%%%%%%%%%%%%%%%%%%%%%%%%%%%%%%%%%%%%%
%
%  Sec 6.1
%
\subsection{Conjecture of the spin-$s$ Ground-state solution}

Let us now consider the ground state of the integrable spin-$s$ XXZ spin chain 
in the massless regime.    
Here we remark that integer $2s$ corresponds to integer $\ell$ 
of $V^{(\ell)}$. 
In the massless regime we set $\eta= i \zeta$ with $0 \le \zeta < \pi$. 
For the spin-$s$ case, in the region $0 \le \zeta < \pi/2s$ 
we assume that the spin-$s$ ground state  
$| \psi_g^{(2s)} \rangle$ is 
given by $N_s/2$ sets of the $2s$-strings:   
\begin{equation} 
\lambda_{a}^{(\alpha)} 
= \mu_a - (\alpha- 1/2) \eta + \delta_a^{(\alpha)} \, , \quad  
\mbox{for} \, \, a=1, 2, \ldots, N_s/2 \, \,  
\mbox{and} \, \,  \alpha = 1, 2, \ldots, 2s .  
\end{equation} 
Here we also assume that string deviations $\delta_a^{(\alpha)}$ are  
small enough when $N_s$ is large enough. 
In terms of $\lambda_{a}^{(\alpha)}$, 
the spin-$s$ ground state associated with grading $w$ 
is given by 
\begin{equation} 
 | \psi_g^{(2s \, w)} \rangle = 
\prod_{a=1}^{N_s/2} \prod_{\alpha=1}^{2s} 
{B}^{(2s \, w)}(\lambda_a^{(\alpha)}; \{\xi_b \}_{N_s}) | 0 \rangle . 
\end{equation}
Here we have $M$ Bethe roots with $M= 2s \, N_s/2 = s N_s$.  

%%%%%%%%%%%%%%%%%%%%%%%%%%%%%%%%%%%%%%%%%%%%%%%%%%%%%%%%%%%
%
%  Sec 6.2
%
\subsection{Multiple-integral representations for arbitrary matrix elements}

Let us now formulate the multiple-integral representations 
of the spin-$s$ XXZ correlation functions in the general case 
for the massless region: $0 \le \zeta < \pi/2s$.  
We define the zero-temperature correlation function 
for a given product of the general spin-$s$ elementary operators 
with grading $w$ 
 $\widehat{ E}_1^{i_1 , \, j_1 \, (2s \, p)} \cdots 
 \widehat{ E}_m^{i_m, \, j_m \, (2s \, p)}$, which are   
$(2s+1) \times (2s+1)$ matrices, by  
\begin{equation}
\widehat{ F}_m^{(2s \, p)}(\{i_k, j_k \}) =  \langle \psi_g^{(2s \, w)} | 
\prod_{k=1}^{m} \widehat{ E}_k^{i_k , \, j_k \, (2s \, w)} 
|\psi_g^{(2s \, w)}  \rangle / \langle \psi_g^{(2s \, w)} 
| \psi_g^{(2s \, w)} \rangle \, . 
\label{eq:def-CF}
\end{equation}

For the $m$th product of elementary operators, 
we introduce the sets of variables $\varepsilon_{\alpha}^{[k] \, '}$s 
and $\varepsilon_{\beta}^{[k]}$s ($1 \le k \le m$) such that 
the number of $\varepsilon_{\alpha}^{[k] \, '}=1$ with $1 \le a \le 2s$ 
is given by $i_k$ 
and the number of $\varepsilon_{\beta}^{[k]}=1$ with $1 \le b \le 2s$ 
by $j_k$, respectively. 
Here, the variables $\varepsilon_{\alpha}^{[k] \, '}$ 
and $\varepsilon_{\beta}^{[k]}$ take only two values 0 or 1 
(see also Corollary \ref{cor:productE}).  
We then express them by integers $\varepsilon_j^{'}$s and $\varepsilon_j$s for 
$j=1, 2, \ldots, 2sm$ as follows: 
\bea 
\varepsilon_{2s (k-1)+ \alpha}^{'}  & = & \varepsilon_{\alpha}^{[k] \, '} \quad \mbox{for} \quad \alpha=1, 2, \ldots, 2s;  k=1, 2, \ldots, m,  \nonumber \\  
\varepsilon_{2s (k-1)+ \beta} & = & \varepsilon_{\beta}^{[k]} \quad  
\mbox{for} \quad 
\beta=1, 2, \ldots, 2s;  k=1, 2, \ldots, m. 
\eea

For given sets of $\varepsilon_j$ and $\varepsilon_j^{'}$ for 
$j=1, 2, \ldots, 2sm$ we define 
$\mbox{\boldmath$\alpha$}^{-}$ by the set of integers $j$ 
satisfying $\varepsilon_j^{'}=1$ and 
$\mbox{\boldmath$\alpha$}^{+}$ by 
the set of integers $j$ satisfying $\varepsilon_j=0$:  
\begin{equation} 
\mbox{\boldmath$\alpha$}^{-}(\{ \varepsilon_j^{'} \}) 
= \{ j ; \, \varepsilon_j^{'}=1 \}  
\, , \quad 
\mbox{\boldmath$\alpha$}^{+}(\{ \varepsilon_j \}) 
= \{ j ; \, \varepsilon_j=0 \}  \, . 
\label{eq:def-aa'}
\end{equation} 
We denote by $r$ and $r^{'}$ the number 
of elements of the set $\mbox{\boldmath$\alpha$}^{-}$ 
and $\mbox{\boldmath$\alpha$}^{+}$, respectively. Due to charge conservation, 
we have $r + r^{'} = 2sm$. 
Precisely, we have $r= \sum_{k=1}^{m} i_k$ and 
$r^{'}= 2sm - \sum_{k=1}^{m} j_k$.     

For sets ${\bm \alpha}^{-}$ and ${\bm \alpha}^{+}$, which correspond to 
$\{ \varepsilon_a^{'} \}$ and $\{ \varepsilon_b \}$, respectively,   
we define integral variables ${\tilde \lambda}_j$ for $j \in {\bm \alpha}^{-}$ 
and ${\tilde \lambda}^{'}_{j}$ for $j \in {\bm \alpha}^{+}$, 
respectively, by the following:   
\begin{equation} 
({\tilde \lambda}^{'}_{j^{'}_{max}}, \ldots, 
{\tilde \lambda}^{'}_{j^{'}_{min}},  {\tilde \lambda}_{j_{min}}, 
{\tilde \lambda}_{j_{max}})
=(\lambda_1, \ldots, \lambda_{2sm}) \, . 
\end{equation}

We now introduce a matrix 
$S=S\left( (\lambda_j)_{2sm}; (w_j^{(2s)})_{2sm} \right)$. 
For each integer $j$ satisfying $1 \le j \le 2sm$,  
we define $\alpha(\lambda_j)$ by $\alpha(\lambda_j)= \gamma$ 
for an integer $\gamma$ satisfying $1 \le \gamma \le 2s$  
if $\lambda_j$ is related to an integral variable $\mu_j$ through 
$\lambda_j = \mu_j - (\gamma - 1/2) \eta$ 
or if $\lambda_j$ takes a value close to $w_k^{(2s)}$ with $\beta(k)=\gamma$.
Thus, $\mu_j$ corresponds to the ``string center'' of $\lambda_j$. 
Here we have defined $\beta(j)$ by 
\begin{equation}
\beta(j) = j - 2s [[(j-1)/2s]] \quad  (1 \le j \le M). 
\label{df:beta}
\end{equation}   
Here $[[x]]$ denotes the greatest integer less than or equal to $x$. 
We define the $(j,k)$ element of the matrix $S$ by   
\begin{equation} 
S_{j,k} = \rho(\lambda_j - w_k^{(2s)} + \eta/2) \, 
\delta(\alpha(\lambda_j), \beta(k)) \, , \quad {\rm for} \quad 
j, k= 1, 2, \ldots, 2sm \, .  
\end{equation} 
Here $\rho(\lambda)$ denotes the density of string centers \cite{DM2}, 
and $\delta(\alpha, \beta)$ the Kronecker delta.  
We obtain the following multiple-integral representation:     
\begin{eqnarray} 
& & \widehat{ F}^{(2s \, w)}_{m}(\{i_k, j_k\}) =  
\quad \widehat{ C}^{(2s)}(\{i_k, j_k \}) \, \times \nonumber \\ 
& & 
 \times \left( \int_{-\infty+ i \epsilon}^{\infty+ i \epsilon}
%+ \int_{-\infty - i \zeta + i \epsilon}^{\infty - i \zeta  + i \epsilon}
+ \cdots 
+ \int_{-\infty - i 
%\widetilde{ \zeta_s} 
(2s-1) \zeta 
+ i \epsilon}
^{\infty - i 
%\widetilde{\zeta_s} 
(2s-1) \zeta  
+ i \epsilon} \right)  d \lambda_1 
\cdots 
\left( \int_{-\infty+ i \epsilon}^{\infty+ i \epsilon}
%+ \int_{-\infty - i \zeta + i \epsilon}^{\infty - i \zeta  + i \epsilon}
+ \cdots 
+ \int_{-\infty - i 
%\widetilde{\zeta_s } 
(2s-1) \zeta 
+ i \epsilon}^{\infty - i 
%\widetilde{ \zeta_s }
(2s-1) \zeta 
 + i \epsilon} 
\right)  d \lambda_{r^{'}} 
\nonumber \\  
& & \times \left( \int_{-\infty - i \epsilon}^{\infty - i \epsilon}
%+ \int_{-\infty - i \zeta - i \epsilon}^{\infty - i \zeta  - i \epsilon}
+ \cdots 
+ \int_{-\infty - i 
%\widetilde{\zeta_s} 
(2s-1) \zeta 
- i \epsilon}
^{\infty - i 
%\widetilde{\zeta_s} 
(2s-1) \zeta 
 - i \epsilon} \right)  
d \lambda_{r^{'} + 1} 
%{\widetilde{r}}
%\nonumber \\ & & 
\cdots 
\left( \int_{-\infty - i \epsilon}^{\infty - i \epsilon}
%+ \int_{-\infty - i \zeta - i \epsilon}^{\infty - i \zeta - i \epsilon}
+ \cdots 
+ \int_{-\infty - i 
%\widetilde{\zeta_s} 
(2s-1) \zeta 
 - i \epsilon}
^{\infty - i 
%\widetilde{\zeta_s} 
(2s-1) \zeta 
 - i \epsilon} 
\right)  d \lambda_{2sm}
\nonumber \\ 
& & \quad \times 
\sum_{{\bm \alpha}^{+}(\{ \epsilon_j \})} 
Q(\{ \varepsilon_j, \varepsilon_j^{'} \}; \lambda_1, \ldots, \lambda_{2sm}) \, 
{\rm det}
S(\lambda_1, \ldots, \lambda_{2sm}) \, . 
\label{eq:MIR}
\end{eqnarray}
Here the sum of ${\bm \alpha}^{+}(\{ \varepsilon_j \})$
is taken over all $\{ \varepsilon_j \}$ corresponding to 
$\{ \varepsilon_b^{[k]} \}$ $(1 \le k \le m)$ 
such that the number of $\varepsilon_b^{[k]}=1$ 
with $1 \le b \le 2s$ is given by $j_k$.  
$Q(\{ \varepsilon_j, \varepsilon_j^{'} \}; \lambda_1, \ldots, \lambda_{2sm})$ 
is given by     
\begin{eqnarray} 
& & 
Q(\{ \varepsilon_j, \varepsilon_j^{'} \}; \lambda_1, \ldots, \lambda_{2sm}) 
\nonumber \\ 
& & 
%(-1)^{\alpha_{+}} 
=(-1)^{r^{'}} 
{\frac { \prod_{j \in {\bm \alpha}^{-}( \{ \varepsilon_j^{'} \} ) }
\left( \prod_{k=1}^{j-1} 
\sinh({\tilde \lambda}_{j} - w_k^{(2s)} + \eta) 
\prod_{k=j+1}^{2sm} \sinh({\tilde \lambda}_{j} - w_k^{(2s)} ) \right)}
{\prod_{1 \le k < \ell \le 2sm} 
\sinh(\lambda_{\ell} - \lambda_{k} + \eta + \epsilon_{\ell, k})} } 
\nonumber \\ 
& \times &  
{\frac { \prod_{j \in {\bm \alpha}^{+}( \{ \varepsilon_j \} ) } 
\left( \prod_{k=1}^{j-1} 
\sinh({\tilde \lambda}^{'}_{j} - w_k^{(2s)} - \eta) 
\prod_{k=j+1}^{2sm} \sinh({\tilde \lambda}^{'}_{j} - w_k^{(2s)} ) \right)}
{\prod_{1 \le k < \ell \le 2sm} 
\sinh(w_{k}^{(2s)} - w_{\ell}^{(2s)})} } \, . \label{eq:Q}
\end{eqnarray}
In the denominator we set $\epsilon_{k, \ell}= i \epsilon$ for 
${Im}(\lambda_k -\lambda_{\ell}) > 0$ and 
$\epsilon_{k, \ell}= - i \epsilon$ for 
${Im}(\lambda_k -\lambda_{\ell}) < 0$,  where
$\epsilon$ is an infinitesimally small positive number.  
The coefficient $\widehat{ C}^{(2s)}(\{i_k, j_k \})$ is given by 
\bea 
\widehat{ C}^{(2s)}(\{i_k, j_k \})  
& = & \prod_{k=1}^{m} \widehat{ N}_{i_k, \, j_k}^{(\ell)} \non \\ 
& = & \prod_{k=1}^{m} 
\left( {\frac {g(j_k)} {g(i_k)}} \, 
{\frac {F(2s, i_k)} {F(2s, j_k)}} q^{i_k(2s- i_k)/2 - j_k(2s - j_k)/2} 
\right) \, . \label{eq:Coef} 
\eea
Here we have made use of (\ref{eq:reduction-formula-product}) and 
(\ref{eq:red-FF-explicit}). 
If we put $g(2s, j)=\sqrt{F(2s, j)}$ for $j=0, 1, \ldots, 2s$ 
into (\ref{eq:Coef}),  we have  
\begin{equation} 
\widehat{ C}^{(2s)}(\{i_k, j_k \}) = 
\prod_{k=1}^{m} \sqrt{  
\left[ \begin{array}{c}
      2s \\
      i_k  
      \end{array} 
 \right]_q  
\left[ \begin{array}{c}
      2s \\
      j_k  
      \end{array} 
 \right]_q^{-1} } 
%
% q^{(\ell-1)(i_k - j_k)} 
%
%
\, .  
\end{equation}
In (\ref{eq:Q}) we may take any ${\bm \alpha}^{-}(\{ \varepsilon_j^{'} \})$ 
corresponding to $\varepsilon_{\alpha}^{[k] \, '}$s 
for $k=1, 2, \ldots, m$, as far as  
the number of $\varepsilon_{\alpha}^{[k] \, '}=1$ with $1 \le {\alpha} \le 2s$ 
is given by $i_k$ for each $k$.  

We can show the symmetric expression for the multiple-integral 
representation of the spin-$s$ correlation function 
$\widehat{ F}^{(2s \, w)}_{m}(\{i_k, j_k\})$ as follows. 
\begin{eqnarray}
&&  
\widehat{ F}^{(2s \, w)}_{m}(\{i_k, j_k\}) = 
 \frac{\widehat{ C}^{(2s)}(\{i_k, j_k \})}
{\prod_{1 \leq \alpha < \beta \leq 2s}
\sinh^{m}(\beta-\alpha )\eta}     
\prod_{1\leq k < l \leq m}
\frac{\sinh^{2s}(\pi(\xi_k-\xi_l)/\zeta)}
{\prod^{2s}_{j=1}\prod^{2s}_{r=1}\sinh(\xi_k-\xi_l+(r-j)\eta)}   
\non \\
&&  \times \sum_{\sigma \in {\cal S}_{2sm}/({\cal S}_m)^{2s} } 
({\rm sgn} \, \sigma) \, 
 \prod^{r^{'}}_{j=1}  
\int^{\infty+ i \epsilon }_{-\infty + i \epsilon}   
d \mu_{\sigma j}  \prod^{2sm}_{j=r^{'} +1}
\int^{\infty - i \epsilon }_{-\infty - i \epsilon} d \mu_{\sigma j} 
\nonumber \\  
& &
%\times \prod_{k=1}^{m} \sum_{\{ \varepsilon_b^{[k]} \} }
%
\sum_{\{ \varepsilon_b^{[1]} \} } \cdots \sum_{\{ \varepsilon_b^{[k]} \} } \, 
Q^{'}(\{ \varepsilon_j, \varepsilon_j^{'} \}; \lambda_{\sigma 1}, \ldots, 
\lambda_{\sigma(2sm)})) \, 
\left( \prod^{2sm}_{j=1} 
{\frac {\prod^{m}_{b=1} \prod^{2s-1}_{\beta=1}
\sinh(\lambda_{j}-\xi_b+ \beta \eta)}
{\prod_{b=1}^{m} \cosh(\pi(\mu_{j}-\xi_b)/\zeta)}} \right) \non \\ 
& &\times \, {\frac {i^{2sm^2}} { (2 i \zeta)^{2sm} }} \, 
\prod^{2s}_{\gamma=1} 
\prod_{1 \le b < a \le m}
\sinh(\pi(\mu_{2s(a-1)+\gamma}-\mu_{2s(b-1)+\gamma})/\zeta) \, .  
%\nonumber \\ 
 \label{eq:CFF2}
\end{eqnarray}
Here $\lambda_j$ are given by $\lambda_j= \mu_j - (\beta(j)-1/2) \eta$ for 
$j=1, \ldots, 2sm$, and (${\rm sgn} \, \sigma$) 
denotes the sign of permutation 
$\sigma \in {\cal S}_{2sm}/({\cal S}_m)^{2s}$. The coefficient 
$\widehat{ C}^{(2s)}(\{i_k, j_k \})$ is given by (\ref{eq:Coef}), and   
$Q^{'}(\{ \varepsilon_j, \varepsilon_j^{'} \}; \lambda_{1}, \ldots, 
\lambda_{2sm}))$ is given by 
$Q^{'}(\{ \varepsilon_j, \varepsilon_j^{'} \}; \lambda_{1}, \ldots, 
\lambda_{2sm})) = 
Q(\{ \varepsilon_j, \varepsilon_j^{'} \}; \lambda_{1}, \ldots, 
\lambda_{2sm})) \times 
{\prod_{1 \le k < \ell \le 2sm} 
\sinh(w_{k}^{(2s)} - w_{\ell}^{(2s)})}$.  
We recall that the sums over $\{ \varepsilon_{\beta}^{(k)} \}$ 
are taken over all 
$\varepsilon_{\beta}^{(k)}$s for $1 \le k \le m$ 
such that the number of integers $\beta$ with 
$\varepsilon_{\beta}^{(k)}=1$ and $1 \le \beta \le \ell$ 
is equal to $j_k$ for each $k$. 

In Appendix E we shall explain 
the derivation of the symmetric expression for 
the multiple-integral representation of the correlation functions.

%%%%%%%%%%%%%%%%%%%%%%%%%%%%%%%%%%%%
%
%  Sec 6.3 
%
\subsection{Relations due to the spin inversion symmetry}

We shall derive a consequence of the spin-inversion symmetry. We shall assume 
that it should hold for the ground state of the integrable spin-$s$ XXZ chain associated with even $L$. 
Here we recall $L=2s N_s$ and $N_s$ is the number of the lattice sites of 
the spin-$\ell/2$ XXZ chain and $L$ that of the spin-1/2 XXZ chain 
associated with it.

Let us denote by $| \psi_g^{(2s \, w; \, 0)} \rangle$    
%the largest-eigenvalue state 
the Bethe-ansatz eigenstate with $M=s N_s$ down-spins 
for the spin-1/2 XXZ transfer matrix under zero magnetic field. 
It is given by $| \psi_g^{(2s \, w; \, 0)} \rangle = 
\prod_{\gamma=1}^{M} B^{(2s \, w; \, 0)} (\lambda_{\gamma}) | 0 \rangle$ 
with $M=L/2$ where inhomogeneity parameters are given by   
the $N_s$ pieces of the complete $2s$-strings, $w_j^{(2s)}$ 
for $j=1, 2, \ldots, L$.  
We now assume that it has the spin inversion symmetry: 
\be 
U | \psi_g^{(2s \, w; \, 0)} \rangle 
= \pm | \psi_g^{(2s \, w; \, 0)} \rangle \qquad  
\mbox{for} \quad U=\prod_{j=1}^{L} \sigma_j^{x} . 
\ee 
It leads to symmetry relations as follows.
\be 
\langle \psi_g^{(2s \, w; \, 0)} | \, 
e_1^{\varepsilon_1^{'}, \varepsilon_1} \cdots 
e_{2s}^{\varepsilon_{2s}^{'}, \varepsilon_{2s}} 
\, 
| \psi_g^{(2s \, w; \, 0)} \rangle  
= \langle \psi_g^{(2s \, w; \, 0)} | \, 
e_1^{1- \varepsilon_1^{'}, \, 1 - \varepsilon_1} \cdots 
e_{2s}^{ 1 - \varepsilon_{2s}^{'}, \, 1 - \varepsilon_{2s}} \, 
| \psi_g^{(2s \, w; \, 0)} \rangle \, . \label{eq:spin-inv}
\ee

%%%%%%%%%%%%%%%
%
% derivation 
%

For the XXX case where the parameter $q$ is given by 1, we can show that the 
spin-1/2 XXX transfer matrix with arbitrary inhomogeneity parameters $w_j$ 
has the $SU(2)$ symmetry and hence every Bethe-ansatz eigenvector 
is a highest weight vector of the $SU(2)$. Thus, the Bethe eigenvector 
with $S^Z=0$ has the total spin 0, and hence it is invariant under any 
rotational operation. Therefore, it has the spin inversion symmetry. 
For the XXZ case where $q$ is not equal to 1,  the 
spin-1/2 XXX transfer matrix with arbitrary inhomogeneity parameters $w_j$ 
does not have the SU(2) symmetry, in general.  However, we assume that 
it has the spin inversion symmetry such as the XXX case. 

%%%%%%%%%%%%%%%

Applying the spin-inversion symmetry (\ref{eq:spin-inv}) we derive symmetry 
relations among the expectation values of local or global operators. 

For an illustration, let us evaluate the one-point function 
in the spin-1 case with $i_1=j_1=1$, 
$\langle E_1^{1, \, 1 \, (2 \, p)} \rangle$. 
Setting $\varepsilon_1^{'}=0$ and $\varepsilon_2^{'}=1$ we decompose the 
spin-1 elementary operator in terms of a sum of products of the spin-1/2 ones  
\be 
\langle \psi_g^{(2 \, p)} | E_1^{1, \, 1 \, (2 \, p)} 
| \psi_g^{(2 \, p)} \rangle 
= 
\langle \psi_g^{(2 \, p ; \, 0)} | e_1^{0, \, 0} e_2^{1, \, 1} 
| \psi_g^{(2 \, p ; \, 0)} \rangle
+ \langle \psi_g^{(2 \, p ; \, 0)} | e_1^{0, \, 1} e_2^{1, \, 0} 
| \psi_g^{(2 \, p ; \, 0)} \rangle 
\, .  
\ee
Through the symmetry relations (\ref{eq:sym-epsilon_a'}) 
with respect to $\varepsilon_{\alpha}^{'}$  
we have the following equalities:   
\bea 
& & \langle \psi_g^{(2 \, p ; \, 0)} | e^{0, \, 0}_{1} e_2^{1, \, 1} 
| \psi_g^{(2 \, p ; \, 0)} \rangle  
= \langle \psi_g^{(2\, p ; \, 0)} | e^{1, \, 0}_{1}e_2^{0, \, 1} 
| \psi_g^{(2 \, p ; \, 0)} \rangle 
\, , \non \\  
& & \langle \psi_g^{(2 \, p ; \, 0)} | e^{1, \, 1}_{1}e_2^{0, \, 0} 
| \psi_g^{(2 \, p ; \, 0)} \rangle  
= \langle \psi_g^{(2 \, p ; \, 0)} | e^{0, \, 1}_{1} e_2^{1, \, 0} 
| \psi_g^{(2 \, p ; \, 0)} \rangle .  
\eea 
From spin-inversion symmetry (\ref{eq:spin-inv}) we have 
\bea 
& & \langle \psi_g^{(2 \, p ; \, 0))} | e^{0, \, 0}_{1}e_2^{1, \, 1} 
| \psi_g^{(2 \, p ; \, 0))} \rangle  
= \langle \psi_g^{(2 \, p ; \, 0)} | e^{1, \, 1}_{1}e_2^{0, \, 0} | \psi_g^{(2 \, p ; \, 0)} \rangle 
\, , \non \\  
& & \langle \psi_g^{(2 \, p ; \, 0) )} | e^{0, \, 1}_{1}e_2^{1, \, 0} 
| \psi_g^{(2 \, p ; \, 0)} \rangle  
= \langle \psi_g^{(2 \, p ; \, 0)} | e^{1, \, 0}_{1}e_2^{0, \, 1} | \psi_g^{(2 \, p ; \, 0)} \rangle 
\eea
and hence we have the equalities of the four terms. 
We therefore obtain the following: 
\be 
\langle \psi_g^{(2)} | E_1^{1, \, 1 \, (2 \, p)} | \psi_g^{(2)} \rangle 
= 2 \, \langle \psi_g^{(2 \, p ; \, 0)} | e_1^{0, \, 0} e_2^{1, \, 1} | \psi_g^{(2 \, p ; \, 0)} \rangle\, . 
\ee   

We thus derive the double-integral representation 
of the one-point function $\langle E_1^{1, \, 1 \, (2 \, p)} \rangle$ 
given in  Ref. \cite{DM2}. 
For the spin-1 case, each of 
 the one-point functions is given by the double integral 
associated with a single product of the spin-1/2 operators.

\section*{Acknowledgements}
The author would like to thank C. Matsui, K. Motegi and J. Sato 
for useful comments. 

\appendix 

%%%%%%%%%%%%%%%%%%%%%%%%%%%%%%%%%%%%%%%%%%%%%%%
%
%  Sec A 
%
\setcounter{equation}{0} 
 \renewcommand{\theequation}{A.\arabic{equation}}
\section{Derivation of Proposition \ref{prop:E-w} }

It follows from Lemma \ref{lem:PE=EP=E} that we have 
$E^{i, \, j \, (\ell \, w)} = P^{(\ell)} E^{i, \, j \, (\ell \, w)}$.   
We therefore have 
\bea
E^{i, \, j \, (\ell \, w)} & = & P^{(\ell)} 
\sum_{ (\varepsilon_{\alpha}^{'}(i))_{\ell}}  
\sum_{ (\varepsilon_{\beta}(j))_{\ell} } 
g_{ij}(\varepsilon_{\alpha}^{'}(i), \varepsilon_{\beta}(j) )  \, \, 
\sigma_{a(1)}^{-} \cdots \sigma_{a(i)}^{-} 
|| \ell, 0 \rangle  \langle \ell, 0 || 
\sigma_{b(1)}^{+} \cdots \sigma_{b(j)}^{+}   
\non \\ 
& = & ||\ell, i\rangle \sum_{(\varepsilon_{\alpha}^{'}(i))_{\ell} } 
\sum_{ (\varepsilon_{\beta}(j))_{\ell} } 
g_{ij}( \varepsilon_{\alpha}^{'}(i), \varepsilon_{\beta}(j) ) \, \, 
\langle \ell, i || \sigma_{a(1)}^{-} \cdots \sigma_{a(i)}^{-} 
|| \ell, 0 \rangle   \langle \ell, 0 || 
\sigma_{b(1)}^{+} \cdots \sigma_{b(j)}^{+}   
\non \\ 
& = & ||\ell, i\rangle 
\sum_{ (\varepsilon_{\beta}(i))_{\ell}} 
\sum_{ (\varepsilon_{\alpha}^{'}(j))_{\ell}} 
g_{ij}(\varepsilon_{\alpha}^{'}(i), \varepsilon_{\beta}(j) ) \, 
q^{a(1)+ \cdots + a(i)-i}  \non \\  
& & \times \, \langle \ell, i || \sigma_{a(1)}^{-} \cdots \sigma_{a(i)}^{-} 
|| \ell, 0 \rangle  q^{-(a(1)+ \cdots + a(i)-i)} \quad 
 \langle \ell, 0 || \sigma_{b(1)}^{+} \cdots \sigma_{b(j)}^{+}
\eea
Applying Lemma \ref{lem:independance} we have 
\bea
E^{i, \, j \, (\ell \, w)} & = & ||\ell, i \rangle 
\sum_{(\varepsilon_{\beta}(i))_{\ell} } 
\left( \sum_{ ( \varepsilon_{\alpha}^{'}(j))_{\ell} } 
g_{ij}(\varepsilon_{\alpha}^{'}(i), \varepsilon_{\beta}(j) ) \, 
q^{a(1)+ \cdots + a(i)-i} \right) \non \\ 
& & \times \, \langle \ell, i || \sigma_{a(1)}^{-} \cdots \sigma_{a(i)}^{-} 
|| \ell, 0 \rangle  q^{-(a(1)+ \cdots + a(i)-i)} \quad 
 \langle \ell, 0 || \sigma_{b(1)}^{+} \cdots \sigma_{b(j)}^{+}
\eea
Taking the sum over sequences $(\varepsilon_{\alpha}^{'}(i))_{\ell}$ 
with Lemma \ref{lem:g-sum} we have 
\bea
E^{i, \, j \, (\ell \, w)} & = & ||\ell, i \rangle 
\left[ \begin{array}{c}
\ell \\
i 
\end{array} \right] 
\, 
\left[ \begin{array}{c}
\ell \\
j 
\end{array} \right]_{q}^{-1} 
\, 
q^{i(i-1)/2-j(j-1)/2} 
\langle \ell, i || \sigma_{a(1)}^{-} \cdots \sigma_{a(i)}^{-} 
|| \ell, 0 \rangle  q^{-(a(1)+ \cdots + a(i)-i)} 
\non \\
& & \quad \times \sum_{ (\varepsilon_{\beta}(j))_{\ell} } 
 \langle \ell, 0 || \sigma_{b(1)}^{+} \cdots \sigma_{b(j)}^{+}
q^{b(1) + \cdots + b(j)-j} \label{eq:middle}
\eea 
We move the conjugate vector $\langle \ell, i ||$ to the left 
in (\ref{eq:middle}), and through (\ref{eq:eps-ab}) 
we express $\sigma_{a(1)}^{-} \cdots \sigma_{a(i)}^{-} 
|| \ell, 0 \rangle 
\langle \ell, 0 || \sigma_{b(1)}^{+} \cdots \sigma_{b(j)}^{+}$ 
in terms of the spin-1/2 elementary operators,  
we have 
\bea
E^{i, \, j \, (\ell \, w)} & = & ||\ell, i \rangle \langle \ell, i || \, 
\left[ \begin{array}{c}
\ell \\
i 
\end{array} \right] 
\, 
\left[ \begin{array}{c}
\ell \\
j 
\end{array} \right]_{q}^{-1} 
\, 
q^{i(i-1)/2-j(j-1)/2} 
\non \\ 
& & \times  \sum_{ (\varepsilon_{\beta}(j))_{\ell} } 
e_1^{\varepsilon_1^{'}(i), \, \varepsilon_1(j)} \cdots 
e_{\ell}^{\varepsilon_{\ell}^{'}(i), \, \varepsilon_{\ell}(j)} \,  
q^{-(a(1)+ \cdots + a(i)-i)} q^{b(1) + \cdots + b(j)-j} 
\label{eq:E-sum-e}
\eea
Applying the gauge transformation of Lemma \ref{lem:inv-gauge}  
to (\ref{eq:E-sum-e}) we obtain Proposition \ref{prop:E-w}.

%%%%%%%%%%%%%%%%%%%%%%%%%%%%%%%%%%%%%%%%%%%%%%%
%
%  Sec B 
%
\setcounter{equation}{0} 
\renewcommand{\theequation}{B.\arabic{equation}}

\section{Reduction of spin-$\ell/2$ Hermitian elementary operators}

Let us introduce vectors $\widetilde{|| \ell, n \rangle}$ 
which are Hermitian conjugate to $\langle \ell, n ||$ when 
$|q|=1$ for positive integers $\ell$ with $n=0, 1, \ldots, \ell$ \cite{DM2}. 
Setting the norm of $\widetilde{|| \ell, n \rangle}$ 
such that $\langle \ell, n || \,  \widetilde{|| \ell, n \rangle}=1$, 
vectors $\widetilde{|| \ell, n \rangle}$ are given by   
\begin{equation} 
\widetilde{|| \ell, n \rangle} = 
\sum_{1 \le i_1 < \cdots < i_n \le \ell} \sigma_{i_1}^{-} 
\cdots \sigma_{i_n}^{-} ||\ell, \,  0 \rangle 
q^{-(i_1 + \cdots + i_n) + n \ell - n(n-1)/2} 
\left[ 
\begin{array}{cc} 
\ell \\ 
n 
\end{array} 
 \right]_q \, 
q^{-n(\ell-n)} 
\left( 
\begin{array}{cc} 
\ell \\ 
n 
\end{array} 
 \right)^{-1} \, . 
\end{equation}
We define the spin-$\ell/2$ Hermitian elementary matrices 
associated with homogeneous grading, 
$\widetilde{E}^{i, \, j \, (\ell, \, +)}$, by 
\be
\widetilde{E}^{i, \, j \, (\ell, \, +)} = 
\widetilde{|| \ell, i \rangle}  
\langle \ell, j || \, . 
\ee
Introducing ${\widetilde g}_{i,j}$ by 
\be 
\widetilde{|| \ell, i \rangle}  
\langle \ell, j || 
= 
\sum_{( \varepsilon_{\alpha}^{'} )_{\ell}} 
\sum_{ ( \varepsilon_{\beta} )_{\ell} } 
{\widetilde g}_{i,j}( \varepsilon_{\alpha}^{'}(i) \}, 
 \varepsilon_{\beta}(j) ) 
e_1^{\varepsilon_1^{'}(i), \, \varepsilon_1(j)} \cdots   
e_{\ell}^{\varepsilon_{\ell}^{'}(i), \, \varepsilon_{\ell}(j)} \, , 
\ee
 we have 
\be 
{\widetilde g}_{i,j}( \varepsilon_{\alpha}^{'}(i), 
\varepsilon_{\beta}(j))
= 
\left[ 
\begin{array}{c}  
\ell \\ 
i 
\end{array} 
\right]_q 
\left[ 
\begin{array}{c}  
\ell \\ 
j 
\end{array} 
\right]_q^{-1} 
\left( 
\begin{array}{c}  
\ell \\ 
i 
\end{array} 
\right)_q \, q^{i(i-1)/2-j(j-1)/2} 
\, q^{-(a(1) + \cdots +a(i) - i) + (b(1) + \cdots + b(j) -j)} 
\ee
We derive the reduction formula 
for the Hermitian elementary operators 
$\widetilde{E}^{i, \, j \, (\ell, \, +)}$ as follows. 
\bea 
& & \widetilde{E}^{i, \, j \, (\ell, \, +)} 
 = {\widetilde P}^{(\ell)} \widetilde{E}^{i, \, j \, (\ell, \, +)}  \non \\ 
&  & = 
\left[ 
\begin{array}{c}  
\ell \\ 
i 
\end{array} 
\right]_q 
\left[ 
\begin{array}{c}  
\ell \\ 
j 
\end{array} 
\right]_q^{-1} 
\left( 
\begin{array}{c}  
\ell \\ 
i 
\end{array} 
\right) \, q^{i(i-1)/2-j(j-1)/2} \widetilde{ || \ell, i \rangle} 
\non \\ 
& & \times 
\sum_{ ( \varepsilon_{\beta}(j) )_{\ell} } 
\sum_{ (\varepsilon_{\alpha}^{'}(i) )_{\ell} } 
\left( \langle \ell, i || \sigma_{a(1)}^{-}  \cdots \sigma_{a(i)}^{-} 
|| \ell, 0 \rangle q^{-(a(1) + \cdots + a(i) -i)}    
\right) \langle \ell, 0 || \sigma_{b(1)}^{+}  \cdots \sigma_{b(j)}^{+}  
q^{b(1) + \cdots + b(j) -i} \, .  \non \\ 
\eea
Here, applying Lemma \ref{lem:independance} we show that  
the inside of the parentheses (or the round brackets) 
is independent of $a(k)$s.  
Making use of the following:  
\be 
\sum_{( \varepsilon_{\alpha}^{'}(i))_{\ell} } 1 = 
\left( 
\begin{array}{c}  
\ell \\ 
i 
\end{array} 
\right) \, ,  
\ee
we thus have 
\bea
\widetilde{E}^{i, \, j \, (\ell, \, +)} 
& = &  
\left[ 
\begin{array}{c}  
\ell \\ 
i 
\end{array} 
\right]_q 
\left[ 
\begin{array}{c}  
\ell \\ 
j 
\end{array} 
\right]_q^{-1} 
\left( 
\begin{array}{c}  
\ell \\ 
i 
\end{array} 
\right)^{-1} \, q^{i(i-1)/2-j(j-1)/2} \widetilde{ || \ell, i \rangle} 
\langle \ell, i ||
\non \\ 
& & \times 
\left( 
\begin{array}{c}  
\ell \\ 
i 
\end{array} 
\right) \, 
\sum_{ (\varepsilon_{\beta}(j) )_{\ell} }    
 \sigma_{a(1)}^{-}  \cdots \sigma_{a(i)}^{-} 
|| \ell, 0 \rangle  \langle \ell, 0 || \sigma_{b(1)}^{+} 
 \cdots \sigma_{b(j)}^{+} \,  q^{-(a(1) + \cdots + a(i) -i)}  
q^{b(1) + \cdots + b(j) -i} \non 
%\\ 
\eea
\bea
& = &  
\left[ 
\begin{array}{c}  
\ell \\ 
i 
\end{array} 
\right]_q 
\left[ 
\begin{array}{c}  
\ell \\ 
j 
\end{array} 
\right]_q^{-1}
 q^{i(i-1)/2-j(j-1)/2} \, \widetilde{ || \ell, i \rangle} 
\langle \ell, i || \, e^{-(i-j) \xi_1} \, 
\sum_{ \{ \varepsilon_{\beta} \} }  
\chi_{1 \cdots \ell} \,  e_1^{\varepsilon_1^{'}, \, \varepsilon_{1}} 
\cdots e_{\ell}^{\varepsilon_{\ell}^{'}, \, \varepsilon_{\ell}}     
\, \chi_{1 \cdots \ell}^{-1} \, . \non \\ 
\label{eq:tilde-red-formula}
\eea
Here we have applied Lemma \ref{lem:inv-gauge} to derive the 
last line of eq. (\ref{eq:tilde-red-formula}).

%%%%%%%%%%%%%%%%%%%%%%%%%%%%%%%%%%%%%%%%%%%%%%%
%
%  Sec C 
%
\setcounter{equation}{0} 
\renewcommand{\theequation}{C.\arabic{equation}}
\section{Non-regularity of the transfer matrix} 

Let us consider the case of $L=3$. We introduce 
 $b_{0j}$ and $c_{0j}^{\pm}$ for $j=1, 2, 3$ by  
 $b_{0j}=b(\lambda-w_j^{(2)})$ and $c_{0j}^{\pm} = 
\exp(\pm (\lambda - w_j^{(2)})) c(\lambda - w_j^{(2)})$ 
for $j=1, 2, 3$, respectively. 

The matrix elements of the operator $A_{123}^{(1 \, +)}
(\lambda)$ in the sector of $M=1$ are given by  
\be 
\left. A_{123}^{(1 \, +)} (\lambda) \right|_{M=1}
=
\left(
\begin{array}{ccc} 
b_{03} & c_{02}^{+}c_{03}^{-} & c_{01}^{+} b_{02} c_{03}^{-} \\  
0 & b_{02} & c_{01}^{+} c_{02}^{-} \\ 
0 & 0 & b_{01} 
\end{array}
\right) \, , 
\ee
and  those of the operator $A_{123}^{(1 \, +)}
(\lambda)$ in the sector of $M=1$ are given by 
\be 
\left. D_{123}^{(3 \, +; \, 0)} (\lambda) \right|_{M=1}
=
\left(
\begin{array}{ccc} 
b_{01} b_{02} & 0 &  \\  
b_{01} c_{02}^{-}c_{03}^{+} & b_{01}b_{03} & 0 \\ 
c_{01}^{-} c_{03}^{+}   & c_{01}^{+} c_{02}^{+} b_{03} & b_{02} b_{03} 
\end{array}
\right) \, . 
\ee
Let us set $w_1=w_1^{(2)}= \xi_1$,  
 $w_2=w_2^{(2)}= \xi_1-\eta$, and 
 $w_3= \xi_2$.  Setting $\lambda=\xi_1$ we have 
\be 
\left. \left( A_{123}^{(2 \, +; \, 0)} (\xi_1) + 
D_{123}^{(2 \, +; \, 0)} (\xi_1) 
\right) 
\right|_{S^Z=1/2}
=
\left(
\begin{array}{ccc} 
b_{13}  & {\frac q {[2]_q}} c_{13}^{-}  
&  {\frac 1 {[2]_q}}c_{13}^{-} \\  
0 & {\frac 1 {[2]_q}} & {\frac {q^{-1}} {[2]_q}}  \\ 
c_{13}^{+}   & {\frac q {[2]_q}} b_{13} & {\frac 1 {[2]_q}}  b_{13} 
\end{array}
\right) \, . 
\ee
Here, the second and the third columns are parallel.  Thus, 
the determinant of the spin-1/2 transfer matrix in the sector of $M=1$ 
is non-regular.

%%%%%%%%%%%%%%%%%%%%%%%%%%%%%%%%%%%%%%%%%%%%%%%
%
%  Sec D 
%
\setcounter{equation}{0} 
\renewcommand{\theequation}{D.\arabic{equation}}
%%%%%%%%%%%%%%%%%%%%%%%%%%%%%%%%%%%%%%
\section{Reducing spin-$\ell/2$ Bethe states 
 with principal grading}

In order to evaluate the spin-$s$ form factors 
for the spin-$s$ elementary operators $E^{i, \, j \, (\ell \, p)}$ associated with principal grading, 
we first transform them into those of homogeneous grading,  and then apply to them the formula for expressing 
the spin-$s$ elementary operators in terms of a sum of products  of 
spin-1/2 elementary operators.   

Let us recall the gauge transformation 
$\chi^{(1, \, \ell)}_{0, \, 1 2 \cdots N_s}$ which 
maps the higher-spin transfer matrix associated with principal grading 
of type $(1, {\ell}^{\otimes N_s})$ 
to that of homogeneous grading: 
\bea 
T^{(1, \, \ell +)}_{0, \, 1 2 \cdots N_s}(\lambda) & = & 
\chi_{0, \, 1 2 \cdots N_s}^{(1, \, \ell)}  T^{(1, \, \ell \, p)}(\lambda)  
\left( \chi_{0, \, 1 2 \cdots N_s}^{(1, \, \ell)}\right)^{-1} 
\non \\ 
& = & 
\left( 
\begin{array}{cc} 
\chi_{1 2 \cdots N_s}^{(\ell)} 
A^{(\ell \, p)}(\lambda) \chi_{1 2 \cdots N_s}^{(\ell) \quad -1}  
& e^{-\lambda} \chi_{1 2 \cdots N_s}^{(\ell)} 
B^{(\ell \, p)}(\lambda) \chi_{1 2 \cdots N_s}^{(\ell) \quad -1} \\ 
e^{\lambda} \chi_{1 2 \cdots N_s}^{(\ell)} 
C^{(\ell \, p)}(\lambda) \chi_{1 2 \cdots N_s}^{(\ell) \quad -1}  
&  \chi_{1 2 \cdots N_s}^{(\ell)} 
D^{(\ell \, p)}(\lambda) 
\chi_{1 2 \cdots N_s}^{(\ell) \quad -1}
\end{array} 
\right) \, . 
\eea
We also recall that the $C$ operator acting 
on the tensor product of the spin-$\ell/2$ representations 
$(V^{(\ell)})^{\otimes N_s}$   
is derived from the $C$ operator  
acting on the tensor product 
of the spin-1/2 representations $(V^{(1)})^{\otimes L}$ 
multiplied by the projection operators: 
$$ 
C^{(\ell \, +)}(\mu) = P^{(\ell)}_{1 2 \cdots L} 
C^{(\ell \,  +; \, 0)}(\mu) P^{(\ell)}_{1 2 \cdots L} \, e^{\mu} \, . 
$$  
Here, through the spin-1/2 gauge transformation we have 
$$ 
C^{(\ell \, +; \, 0)}(\mu) = 
\chi_{1 2 \cdots L} \, 
C^{(\ell \, p; \, 0)}(\mu) \, \chi_{1 2 \cdots L}^{-1} \, e^{\mu} . 
$$
We therefore have $\langle \{ \mu_{\alpha} \}_N^{(\ell, \, p)} |$ as  
\bea
 \langle 0 | \prod_{\alpha=1}^{N} 
C^{(\ell \, p)}(\mu_{\alpha}) 
& = & \langle 0 | \prod_{\alpha=1}^{N} \left( 
\left( \chi_{ 1 \cdots N_s}^{(\ell)} \right)^{-1}   
P^{(\ell)}_{1 \cdots L} \, \chi_{1 \cdots L}  
C^{(\ell \, p; \, 0)}(\mu_{\alpha})  \chi_{1 \cdots L}^{-1} \, 
P^{(\ell)}_{1 \cdots L} 
\left( \chi_{ 1 \cdots N_s}^{(\ell)} \right) 
\right) 
\non \\ 
& = & \langle 0 | \prod_{k=1}^{N}  C^{(\ell \, p; \, 0)}(\mu_{k}) 
\, \cdot \,  \chi_{ 1 \cdots L}^{-1}  
P^{(\ell)}_{1 \cdots L} \chi_{ 1 \cdots N_s}^{(\ell)}  \, . 
\eea 
Precisely we derive it as follows: 
\bea 
\langle 0 | \prod_{k=1}^{N} 
C^{(\ell \, p)}(\mu_{k}) 
& = & \langle 0 | \prod_{k=1}^{N} \left\{ 
\left( \chi_{ 1 \cdots N_s}^{(\ell)} \right)^{-1} 
  C^{(\ell \, +)}(\mu_{k}) e^{-\mu_k} 
\left( \chi_{ 1 \cdots N_s}^{(\ell)} \right) \right\}
\non \\ 
& = & \langle 0 | \prod_{k=1}^{N} \left\{ C^{(\ell \, +)}(\mu_{k}) e^{-\mu_k} 
\right\}  \, \, \cdot \, \chi_{ 1 \cdots N_s}^{(\ell)} 
\non \\ 
& = & \langle 0 | \prod_{k=1}^{N} \left( 
P^{(\ell)}_{1 \cdots L} 
 C^{(\ell \, +; \, 0)}(\mu_{k}) e^{-\mu_k}  
P^{(\ell)}_{1 \cdots L} \right) 
\, \, \cdot \, \chi_{ 1 \cdots N_s}^{(\ell)} \non \\ 
& = & \langle 0 | \prod_{k=1}^{N} \left\{ 
 C^{(\ell \, +; \, 0)}(\mu_{k}) e^{-\mu_k} \right\} 
\, \, \cdot \, P^{(\ell)}_{1 \cdots L} \chi_{ 1 \cdots N_s}^{(\ell)} \, . 
\eea
Here we have made use of the commutation relation 
with the projection $P^{(\ell)}_{1 \cdots L}$. Thus, we have 
\bea
\langle \{ \mu_{\alpha} \}_N^{(\ell, \, p)} | & = & \langle 0 | \prod_{k=1}^{N} \left( \chi_{ 1 \cdots L} 
 C^{(\ell \, p; \, 0)}(\mu_{k}) \chi_{ 1 \cdots L}^{-1} \right) 
\, \, \cdot \, P^{(\ell)}_{1 \cdots L} \chi_{ 1 \cdots N_s}^{(\ell)} \non \\ 
& = & \langle 0 | \prod_{k=1}^{N}  
 C^{(\ell \, p; \, 0)}(\mu_{k}) \, \, \cdot \, 
\chi_{ 1 \cdots L}^{-1} 
P^{(\ell)}_{1 \cdots L} \chi_{ 1 \cdots N_s}^{(\ell)} \, . \label{eq:left-p}
\eea
Similarly we have  
\bea
| \{ \mu_{\alpha} \}_N^{(\ell, \, p)} \rangle 
& = &  \prod_{\alpha=1}^{N} 
B^{(\ell \, p)}(\lambda_{\alpha}) | 0 \rangle 
\non \\ 
& = &  \prod_{\alpha=1}^{N} \left( 
\left( \chi_{ 1 \cdots N_s}^{(\ell)} \right)^{-1}   
P^{(\ell)}_{1 \cdots L} \,  \chi_{ 1 \cdots L}
B^{(\ell \, p; \, 0)}(\lambda_{\alpha}) 
\chi_{ 1 \cdots L}^{-1} \, 
P^{(\ell)}_{1 \cdots L} 
\left( \chi_{ 1 \cdots N_s}^{(\ell)} \right) 
\right) 
\non \\ 
& = & 
\left( \chi_{ 1 \cdots N_s}^{(\ell)} \right)^{-1} 
P^{(\ell)}_{1 \cdots L} \,  \chi_{ 1 \cdots L} \, \cdot \, 
 \prod_{\alpha=1}^{M}  
B^{(\ell \, p; \, 0)}(\lambda_{\alpha}) | 0 \rangle \, . \label{eq:right-p}
\eea

%%%%%%%%%%%%%%%%%%%%%%%%%%%%%%%%%%%%%%%%%%%%%%%
%
%  Sec E 
%
\setcounter{equation}{0} 
\renewcommand{\theequation}{E.\arabic{equation}}
%%%%%%%%%%%%%%%%%%%%%%%%%%%%%%%%%%%%%%%
%\input{symmetric}

\section{Symmetric multiple-integral representations}

%\vskip 12pt   

For the spin-1 case, 
let us express the double sum 
$\sum_{c_1}^{M^{'}} \sum_{c_2}^{M} f(c_1, c_2)$ in the symmetric form 
which leads to the symmetric expression of the multiple-integral representations.  Here $c_1$ and $c_2$ run through from 1 to $M$ corresponding to 
all the 2-strings of the ground state.  
 
%\vskip 12pt

Recall that variable $c_j$ ($1 \le  j \le 2s m$) takes 
 integers from 1 to $M=N_s$ which correspond to $N_s/2$ sets of 2-strings. 
We express them in terms of integers $a(j, \beta)$ for $\beta=1, 2$, 
where $a(j, \beta)$ take integral values from 1 to $N_s/2$. 
We first express the sum over $c_1$ in terms of $a(1, \beta)$ as follows. 
\be 
\sum_{c_1} =  \sum_{a(1,1)=1}^{M/2} + \sum_{a(1,2)=1}^{M/2}  
\ee
More precisely we have 
\be 
\sum_{c_1} f(c_1) =  \sum_{a(1,1)=1}^{M/2} 
f( 2 (a(1,1)-1) + 1)  +  
\sum_{a(1,2)=1}^{M/2} f( 2(a(1,2)-1) + 1)  
\ee
For the spin-1 case with one-point function ($m=1$) we have 
\bea 
\sum_{c_1=1}^{M^{'}} \sum_{c_2=1}^{M} f(c_1, c_2)  
& = & 
\left( \sum_{a(1,1)=1}^{M^{'}/2} + \sum_{a(1,2)=1}^{M^{'}/2} \right)  
\left( \sum_{a(2,1)=1}^{M/2} + \sum_{a(2,2)=1}^{M/2} \right)  
 f(c_1, c_2)  
\non \\ 
& = & 
\left( 
\sum_{a(1,1)=1}^{M^{'}/2} \sum_{a(2,1)=1}^{M/2} + 
\sum_{a(1,1)=1}^{M^{'}/2} \sum_{a(2,2)=1}^{M/2} + 
\sum_{a(1,2)=1}^{M^{'}/2} \sum_{a(2,1)=1}^{M/2} +  
\sum_{a(1,2)=1}^{M^{'}/2} \sum_{a(2,2)=1}^{M/2} \right)  
 f(c_1, c_2)  
\non \\ 
& = & 
\left( 0  + 
\sum_{a(1,1)=1}^{M^{'}/2} \sum_{a(2,2)=1}^{M/2} + 
\sum_{a(1,2)=1}^{M^{'}/2} \sum_{a(2,1)=1}^{M/2} +  0  \right)  
 f(c_1, c_2) 
\eea
Here we recall that  $f(c_1, c_2) $ vanishes  
if the types of string rapidities $c_1$ and $c_2$ are the same.

Let us now introduce variables 
$a_j$ ($j=1,2$) which correspond to the centers of the 2-strings. 
We define an integer-valued variable ${\hat c}_j$ 
which is a  function of $a_j$ as follows  
\be 
{\hat c}_j = 2 (a_j -1) + \beta(j) 
\ee
Then, in terms of permutations $\pi$ in the symmetric group ${\cal S}_2$ 
we express the sum as follows 
\bea 
& & \sum_{c_1}^{M^{'}} \sum_{c_2}^{M} f(c_1, c_2)  
= \sum_{a_1=1}^{M^{'}/2} \sum_{a_2=1}^{M/2} f({\hat c}_1, {\hat c}_2)   
+ \sum_{a_2=1}^{M^{'}/2} \sum_{a_1=1}^{M/2} f({\hat c}_2, {\hat c}_1) 
\non \\ 
& & = \sum_{a_{e 1}=1}^{M^{'}/2} \sum_{a_{e 2}=1}^{M/2} 
f({\hat c}_{e 1}, {\hat c}_{e 2})   
+ \sum_{a_{(12) 1} =1}^{M^{'}/2} \sum_{a_{(12)2} =1}^{M/2} 
f({\hat c}_{(12) 1} , {\hat c}_{(12) 2} )
\non \\ 
& & = \sum_{\pi \in {\cal S}} 
\sum_{a_{\pi 1} =1}^{M^{'}/2} \sum_{a_{\pi 2} =1}^{M/2} 
f({\hat c}_{\pi 1} , {\hat c}_{\pi 2} ) \, . 
\eea
We thus have  
\be 
\sum_{c_1=1}^{M^{'}} \sum_{c_2=1}^{M} f(c_1, c_2)  
= \sum_{\pi \in {\cal S}} 
\sum_{a_{\pi 1} =1}^{M^{'}/2} \sum_{a_{\pi 2} =1}^{M/2} 
f({\hat c}_{\pi 1} , {\hat c}_{\pi 2} )
\ee
The result leads to the symmetric expression of the multiple-integral 
representation.

%%%%%%%%%%%%%%%%%%%%%%%%%%%%%%%%%%%%%%
%
%  Sec Ref
%


\begin{thebibliography}{99}
\bibitem{SCP} By applying the QISP formula 
we can express the local operators of the $N$-state 
superintegrable chiral Potts chain in terms of the operator-valued 
matrix elements of the monodromy matrix of the integrable spin-$(N-1)/2$ XXZ 
spin chain under some twisted boundary conditions.
\bibitem{Assumption} 
Here we assume that for a given solution $\{ \lambda_{\gamma} \}$ of the 
the Bethe ansatz equations of the spin-$\ell/2$ chain 
there is a solution of the Bethe ansatz equations, 
$\{ \lambda_{\gamma}(\epsilon) \}$, 
for the spin-1/2 chain with inhomogeneity parameters given by 
the almost complete $\ell$-strings, i.e.,   
$w_j=w_j^{(\ell; \, \epsilon)}$ ($1 \le j \le L)$, 
and also that the solution $\{ \lambda_{\gamma}(\epsilon) \}$ 
approaches $\{ \lambda_{\gamma} \}$ continuously at $\epsilon=0$.
\bibitem{ER2} In Ref. \cite{DM2} eqs. 
[3.27] and [3.28] are not valid, which were derived from eq. [3.26] 
under an invalid assumption that the spin-1/2 monodromy matrix with 
$w_j=w_j^{(\ell; \, \epsilon)}$ should commute with 
projection $P^{(2s)}_{1 2 \cdots L}$ within the order of $\epsilon$.  
%
Thus, the derivation of eq. [4.16] is not valid.  However, 
eq. [4.16] itself is valid,  since we can show it by the revised method. 
Here, for the emptiness formation probability, 
we have only one set ${\bm \alpha}^{+}$ in the sum over sets  
${\bm \alpha}^{+}(\{ \varepsilon_{\beta} \})$. 
%
Moreover, the multiple-integral representation of 
the correlation functions given in eq. [6.14] should be 
replaced with eq. (6.9) of the present paper for the most of the cases.       
Here [a.bc] denote eq. (a.bc) of Ref. \cite{DM2}.
\end{thebibliography}
\end{document}